\documentclass[11pt]{article}

\usepackage[a4paper,margin=1in]{geometry}
\usepackage[T1]{fontenc}
\usepackage{lmodern}
\usepackage{microtype}
\usepackage{amsmath,amssymb,amsthm}
\usepackage{mathtools}
\usepackage{enumitem}
\usepackage{graphicx}
\usepackage{float}
\usepackage{algorithm}
\floatstyle{ruled}
\restylefloat{algorithm}
\usepackage{placeins}

\usepackage{authblk}

\makeatletter
\newcounter{algln@depth}
\newcounter{algln@no}
\newenvironment{enumerate*}{%
  \stepcounter{algln@depth}%
  \ifnum\value{algln@depth}=1 \setcounter{algln@no}{0}\fi
  \list{\stepcounter{algln@no}\arabic{algln@no}.}{%
    \setlength{\topsep}{0.3ex}%
    \setlength{\partopsep}{0pt}%
    \setlength{\itemsep}{0.2ex}%
    \setlength{\parsep}{0pt}%
    \setlength{\leftmargin}{2.2em}%
    \setlength{\labelwidth}{1.6em}%
    \setlength{\labelsep}{0.4em}%
    \setlength{\itemindent}{0pt}%
    \renewcommand{\makelabel}[1]{\hbox to\labelwidth{\hfil##1}}%
  }%
}{%
  \endlist
  \addtocounter{algln@depth}{-1}%
}
\makeatother
\usepackage[round,authoryear]{natbib}
\usepackage[hidelinks]{hyperref}
\usepackage{url}

\newtheorem{theorem}{Theorem}[section]
\newtheorem{lemma}[theorem]{Lemma}
\newtheorem{proposition}[theorem]{Proposition}
\theoremstyle{definition}
\newtheorem{definition}[theorem]{Definition}
\newtheorem{example}[theorem]{Example}
\theoremstyle{remark}
\newtheorem{remark}[theorem]{Remark}

\title{Stringological sequence prediction II:\\
Right-to-left automaticity and related complexity measures}

\author[1,2]{Vanessa Kosoy\thanks{\href{mailto:vanessa@alter.org.il}{\texttt{EMAIL}}}}
\affil[1]{Faculty of Mathematics, Technion, Haifa, Israel}
\affil[2]{Computational Rational Agents Laboratory, Delaware, USA}

\date{}

\begin{document}

\maketitle

\begin{abstract}%
In a previous paper \citep{Kosoy2026a}, we began the study of sequence prediction algorithms adapted to stringological word complexity measures. One measure we considered was left-to-right (most-significant-digit-first) automaticity. Here, we show a statistically and computationally efficient algorithm adapted to the ``dual'' right-to-left (least-significant-digit-first) automaticity, which turns out to be substantially different for our purpose. We also demonstrate a prediction algorithm for a more expressive measure that we call ``arithmetic repetition complexity''. In particular, the latter can be used for predicting the so-called mix-automatic sequences.
\end{abstract}


\section{Introduction}
\label{sec:intro}

Sequence (or ``time series'') prediction is a classical field of study in artificial intelligence, machine learning and statistics \citep{Hutter2024,Cesa-Bianchi2006,Brockwell2002}. Despite the long history of the field, there are few examples of prediction algorithms which \emph{simultaneously} (i) are computationally efficient, (ii) satisfy strong provable mistake bounds, and (iii) display strong performance on a highly versatile class of data distributions. For the sake of concreteness, it seems useful to operationalize (iii) as ``guaranteeing asymptotically near-perfect prediction for rich natural classes of deterministic sequences". Indeed, prediction of deterministic sequences is a common metric of intelligence (\cite{Kyllonen2017,John2003,Carroll1993}).

For any \emph{two} of these 3 natural requirements, there is ample prior work that satisfies them.

Solomonoff induction \citep{Hutter2024} is the gold standard for sequence prediction (assuming no domain-specific knowledge). On any sequence, it asymptotically performs as well as any computable predictor, meeting requirements (ii) and (iii). It has strong conceptual justification as a formalization of Occam's razor. However, Solomonoff induction is uncomputable: it fails requirement (i) catastrophically, making it completely unsuitable for any practical implementation.

Much work on statistics has focused on series of continuous variables (see e.g. \citealp{Brockwell2002}). This led to algorithms requiring assumptions about particular forms of probability distributions, e.g. normal. While some such algorithms meet requirements (i) and (ii), most of them are inapplicable in the discrete setting, and they don't yield interesting classes of deterministic sequences, failing requirement (iii).

In practice, methods based on deep learning are incredibly successful in next-token prediction \citep{Brown2020}, answering requirement (i) and suggesting that requirement (iii) is at least empirically satisfied. However, the theoretical understanding of generalization bounds for deep learning is in its infancy \citep{Fan2020}. In particular, examples where strong rigorous mistake bounds can be proved are lacking, and requirement (ii) fails.

There are computationally efficient algorithms with strong provable generalization bounds for some interesting classes of \emph{stochastic} processes, for example context-tree weighting methods \citep{Kontoyiannis2020}. However, such classes often degenerate in the deterministic case, e.g. admitting only periodic sequences. Hence, requirements (i) and (ii) are satisfied, but requirement (iii) fails.

Notably, the perceptron algorithm for online linear classification \citep{Cesa-Bianchi2006} applied to features computed from the sequence does come close to meeting our desiderata. However, the corresponding class of predictable deterministic sequences is still quite limited.

In this series of papers, we present a novel family of sequence prediction algorithms that are computationally efficient and satisfy provable mistake bounds. The mistake bounds imply an asymptotically vanishing frequency of mistakes for many interesting classes of sequences from combinatorics on words. Thus, our algorithms satisfy (a form of) all 3 of the requirements above.

Formally, we work in the following setting. There is a fixed finite alphabet $\Sigma$ and we are interested in predictors of the form $P:\Sigma^*\to\Sigma$ (here, $P(u)$ is the symbol predicted to follow the prefix $u$). The assumptions about the data are expressed as a \emph{complexity measure} $C:\Sigma^*\to\mathbb{N}$. We are then interested in bounding the number of mistakes $P$ makes on a sequence $u \in \Sigma^*$ in terms of $m \mathrel{:=} C(u)$ and $n \mathrel{:=} |u|$. At the same time, we require that $P$ is computable in polynomial time. Moreover, we wish to simultaneously bound the size of the \emph{internal state} of the predictor in terms of $m$ and $n$ (this can be interpreted as the predictor \emph{compressing} the sequence). The latter is interesting because it leads to predictors that are space efficient and in some cases run in quasilinear\footnote{I.e. linear up to logarithmic factors.} time.

We study multiple natural complexity measures $C$ with strong mistake and compression bounds. (Two such measures were addressed in a previous work \citep{Kosoy2026a}; three are introduced here; further examples will be studied in a sequel.) Conceptually, such complexity measures can be viewed as candidate tractable analogues of Kolmogorov complexity. We strive to simultaneously get a mistake bound of the form $O(m\,\operatorname{poly}(\log n))$ and a compression bound of the form $\operatorname{poly}(m,\log n)$. In particular, when operating on an infinite sequence for which $m$ grows polylogarithmically, such a predictor runs in quasilinear time and polylog space while making only polylog mistakes (see Definition~\ref{def:comp-eff}).

One interesting complexity measure is the size of the smallest automaton that can compute any symbol in the sequence when given the \emph{time index in base $k$} as input, for some fixed integer $k\geq 2$, a concept known as ``$k$-automaticity''. Notably, \emph{infinite} sequences for which this measure is finite have been widely studied in combinatorics on words: they are called \emph{automatic sequences} and have their own rich theory \citep{Allouche2003}. Hence, even this simple complexity measure captures a rich family of interesting examples. More precisely, we get different complexity measures if the automaton is assumed to read digits in left-to-right (LTR) order vs. right-to-left (RTL). For LTR, we proposed an efficient algorithm in the preceding paper \citep{Kosoy2026a}. Here, we propose an algorithm for RTL with $O(k m \log (k m) \log n)$ mistakes and $O(m (k \log m \log n + (\log n)^2))$ compression (Theorem~\ref{thm:rtl-bounds}).

Note that, even though automaton inference has received much attention \citep{Higuera2010}, most existing work is inapplicable to our setting, and our algorithm is novel, to the best of our knowledge.

For most real-world applications, $k$-automaticity seems like a rather contrived complexity measure because it depends on the parameter $k$. There is usually no obvious preferred choice of base in which the time index should be represented. In the LTR case, this problem can be solved by passing from automaticity to the much more versatile \emph{straight-line complexity} (SLC) measure: the size of the smallest \emph{straight-line program} that produces a word. Here, a straight-line program (SLP) is a sort of acyclic context-free grammar often used in stringology (see e.g. \citealp{Burgisser1997}). In order to accomplish an analogous generalization for RTL, we introduce ``arithmetic repetition complexity'' (ARC). ARC is the size of the smallest system of equations that uniquely specifies the word, where each equation can either specify the symbol at a single position, or mandate the equality of two subsequences sampled from \emph{arithmetic progressions}. We also show some connections between ARC and ``zipline programs'': a variant of straight-line programs where the word concatenation operation is replaced by the ``zip'' operation, i.e. round-robin interleaving of words.

In \citet{Kosoy2026a} we observed that SLC is slow-growing for \emph{morphic sequences}, an important class of sequences in combinatorics on words (see e.g. \citealp{Allouche2003}). ARC has the same property, and in addition (via the zipline program connection) it has polylogarithmic growth for the so-called mix-automatic sequences introduced in \citet{Endrullis2013}.

Unfortunately, we were unable to find a sequence prediction algorithm for ARC with a meaningful compression bound. However, we do find a polynomial-time algorithm with a mistake bound of $O(m \log n)$ (Theorem~\ref{thm:arc-muwu}). In the sequel, we plan to demonstrate a further weakening of this complexity measure which admits a quasilinear compression bound. Whether these gaps are unavoidable or are artificats of our methods is left as an open problem.

\section{Setting}
\label{sec:setting}

In this section, we recall the framework of stringological sequence prediction from \citet{Kosoy2026a}. We define the online prediction protocol in terms of state-based algorithms, introduce the notion of stringological complexity measures, and establish rigorous criteria for statistical and computational efficiency.

\subsection{Preliminaries and Notation}
\label{sec:preliminaries-and-notation}

$\mathbb{N}$ stands for the set of natural numbers including 0.

Let $\Sigma$ be a fixed finite\footnote{We assume $\Sigma$ is finite purely for ease of presentation. We could instead assume that $\Sigma=\mathbb{N}$, in which case the factors of $\log |\Sigma|$ in the bounds would be replaced by $\log M$, where $M$ is the highest number that actually appeared in the sequence so far.} alphabet. We denote the set of finite words over $\Sigma$ by $\Sigma^*$ and the set of right-infinite sequences by $\Sigma^\omega$. For a word $u \in \Sigma^*$, we denote its length by $|u|$. The $i$-th symbol in a word $u$ is denoted $u[i]$. We use 0-based indexing, such that
\begin{equation*}
  u = u[0]u[1]\dots u[|u|-1]
\end{equation*}
The notation $u[i:j]$ refers to the factor (subword) $u[i]\dots u[j-1]$. The notation $u[:j]$ is the same as $u[0:j]$. A word $v$ is a prefix of $u$, denoted $v \sqsubseteq u$, if $u = vw$ for some $w \in \Sigma^*$. Logarithms are taken to base 2 unless otherwise specified.

\subsection{The Prediction Protocol}
\label{sec:the-prediction-protocol}

We operate in the standard deterministic online prediction setting. To discuss memory and time constraints rigorously, we model the predictor not merely as a function of the past history, but as a state-based machine.

\begin{definition}
\label{def:predictor}
A \emph{predictor} is a tuple $\Pi = (\mathcal{S}, s_{\mathrm{init}}, \mathcal{U}, \mathcal{P})$, where:
\begin{itemize}
  \item $\mathcal{S} \subseteq \{0,1\}^*$ is the set of possible internal states (represented as binary strings).
  \item $s_{\mathrm{init}} \in \mathcal{S}$ is the initial state.
  \item $\mathcal{U}: \mathcal{S} \times \Sigma \to \mathcal{S}$ is the \emph{state-update function}. It takes the current state and the most recent observation to produce the next state.
  \item $\mathcal{P}: \mathcal{S} \to \Sigma$ is the \emph{state-prediction function}. It maps the current state to a predicted next symbol.
\end{itemize}
\end{definition}

The prediction process proceeds in rounds $t = 0,1,\dots$ for a target sequence $x \in \Sigma^T$. At step $t$:
\begin{enumerate}
  \item The predictor makes a prediction $\hat{x}_t = \mathcal{P}(s_t)$.
  \item The true symbol $x[t]$ is revealed.
  \item The predictor updates its internal state: $s_{t+1} = \mathcal{U}(s_t, x[t])$.
\end{enumerate}
We assume $s_0 = s_{\mathrm{init}}$. The total number of mistakes made by $\Pi$ on a word $x \in \Sigma^T$ is denoted by:
\begin{equation*}
  M_\Pi(x) \mathrel{:=} \left|\left\{t \in \{0,\dots,T-1\} \mid \hat{x}_t \neq x[t]\right\}\right|
\end{equation*}
We also introduce notation for the maximal size of the state during the processing of a word:
\begin{equation*}
  S_\Pi(x) \mathrel{:=} \max_{0 \leq t < T} |s_t|
\end{equation*}

\subsection{Efficiency Criteria}
\label{sec:efficiency-criteria}

We evaluate performance against the inherent structural complexity of the individual sequence, rather than a probabilistic prior.

\begin{definition}
\label{def:word-complexity}
A \emph{word complexity measure} is a function $C: \Sigma^* \to \mathbb{N}$ that satisfies the following conditions:
\begin{itemize}
  \item \textbf{(Polynomial bound.)} There exists a polynomial $p \in \mathbb{N}[x]$ s.t. for all $u \in \Sigma^*$,
  \begin{equation}
    C(u) \leq p(|u|)
    \label{eq:word-com-bound}
  \end{equation}
  \item \textbf{(Approximate monotonicity.)} There exists a polynomial $q \in \mathbb{N}[x,y]$ s.t. for all $u \in \Sigma^*$ and $t \leq |u|$,
  \begin{equation}
    C(u[:t]) \leq q(\log |u|, \log C(u)) \cdot (C(u)+1)
    \label{eq:mono}
  \end{equation}
\end{itemize}
\end{definition}

We seek predictors that are statistically efficient, space efficient and time efficient. We bound these resources in terms of the sequence's complexity $C(x)$ and its logarithmic length $\log |x|$.

\begin{definition}
\label{def:stat-eff}
\textbf{(Statistical Efficiency.)} The predictor $\Pi$ is \emph{statistically efficient} with respect to $C$ if the number of mistakes is quasilinear in the complexity:
\begin{equation}
  M_\Pi(x) \leq (C(x)+1) \cdot \operatorname{poly}(\log |x|)
  \label{eq:def-stat-eff}
\end{equation}
\end{definition}

Notice that any polynomial in $\log C(x)$ can be absorbed into Equation~\eqref{eq:def-stat-eff} due to Equation~\eqref{eq:word-com-bound}.

\begin{definition}
\label{def:comp-feas}
\textbf{(Computational Feasibility.)} The predictor $\Pi$ is \emph{computationally feasible} if
\begin{itemize}
  \item The time to compute $\hat{x}_t = \mathcal{P}(s_t)$ is bounded by $\operatorname{poly}(|s_t|)$.
  \item The time to compute $s_{t+1} = \mathcal{U}(s_t, x[t])$ is bounded by $\operatorname{poly}(|s_t|)$.
  \item $S_\Pi(x)$ is bounded by $\operatorname{poly}(|x|)$.
\end{itemize}
\end{definition}

\begin{definition}
\label{def:comp-eff}
\textbf{(Computational Efficiency.)} The predictor $\Pi$ is \emph{computationally efficient} with respect to $C$ if it's computationally feasible, and for all $x \in \Sigma^*$:
\begin{equation}
  S_\Pi(x) \leq \operatorname{poly}(C(x), \log |x|)
  \label{eq:gen-com-bound}
\end{equation}
\end{definition}

We can thus think of $s_t$ as a compressed representation of the history $x[:t]$, and we refer to inequalities of the form \eqref{eq:gen-com-bound} as ``compression bounds''. (In many examples, this compression is lossless, but it doesn't have to be.) In particular, this bound implies that if the sequence complexity grows polylogarithmically, then the space complexity and the per-round processing time are also polylogarithmic\footnote{Since per-round processing time is polynomial in $|s_t|$ and $|s_t|$ is polylogarithmic in $t$ due to inequality~\eqref{eq:gen-com-bound}.}.

\section{Automaticity}
\label{sec:auto}

In this section, we investigate \emph{automaticity}, our simplest example of a word complexity measure. Automaticity is the number of states in the smallest finite automaton taking the base-$k$ representation of the time index $t$ as input (\citet{Allouche2003}). Depending on the direction in which the automaton reads the digits---Most Significant Digit First or Least Significant Digit First---we obtain two distinct complexity measures: LTR-automaticity and RTL-automaticity.

Fix an integer base $k \geq 2$. For any length $m \in \mathbb{N}$ and integer $0\leq t < k^m$, let $\langle t \rangle_k^m \in [k]^m$ denote the standard base-$k$ representation of $t$, padded with leading zeros to length $m$. That is, $\langle t \rangle_k^m = w_0 w_1 \dots w_{m-1}$ such that $t = \sum_{j=0}^{m-1} w_j k^{m-1-j}$.

We model the structure of a sequence $x$ via deterministic finite automata with output (DFAOs) that compute $x[t]$ from this padded representation.

\begin{definition}
\label{def:auto}
Let $x \in \Sigma^*$ be a finite string.
\begin{itemize}
  \item The \emph{LTR $k$-automaticity} of $x$, denoted $\mathrm{AC}^{\mathrm{L}}_k(x)$, is the minimal number of states in a DFAO $M$ (with input alphabet $[k]$) such that there exists an integer $m$ satisfying $|x| \leq k^m$ and:
  \begin{equation*}
    \forall t < |x| : x[t] = M(\langle t \rangle_k^m)
  \end{equation*}
  This corresponds to processing the time index in Left-to-Right (MSDF) order.
  \item The \emph{RTL $k$-automaticity} of $x$, denoted $\mathrm{AC}^{\mathrm{R}}_k(x)$, is the minimal number of states in a DFAO $M$ such that there exists an integer $m$ satisfying $|x| \leq k^m$ and:
  \begin{equation*}
    \forall t < |x| : x[t] = M((\langle t \rangle_k^m)^{\mathrm{R}})
  \end{equation*}
  This corresponds to processing the time index in Right-to-Left (LSDF) order, where $(\cdot)^{\mathrm{R}}$ denotes string reversal.
\end{itemize}
\end{definition}

Notice that for any $x \sqsubseteq y$, it's obvious that $\mathrm{AC}^{\mathrm{R}}_k(x) \leq \mathrm{AC}^{\mathrm{R}}_k(y)$ and $\mathrm{AC}^{\mathrm{L}}_k(x) \leq \mathrm{AC}^{\mathrm{L}}_k(y)$. Hence, Equation~\eqref{eq:mono} is satisfied for these measures. It is also easy to see\footnote{This is demonstrated by an automaton that has a state for every element of $[k]^{\leq m}$ with $m=\lceil \log_k |x| \rceil$, s.t. $M(\langle t \rangle_k^m)$ detects the exact $t$ and outputs $x[t]$ accordingly.} that
\begin{equation*}
  \max(\mathrm{AC}^{\mathrm{R}}_k(x),\mathrm{AC}^{\mathrm{L}}_k(x)) \leq k^2 \max(|x|,1)
\end{equation*}
Therefore $\mathrm{AC}^{\mathrm{R}}_k$ and $\mathrm{AC}^{\mathrm{L}}_k$ are word complexity measures. 

\begin{example}
\label{ex:thue_morse}
  Let $\Sigma = \{0, 1\}$. We define the \emph{Thue-Morse sequence} recursively by a family of prefixes $u_n$ of length $2^n$:
  \[
    u_0 = 0, \quad u_{n+1} = u_n \overline{u_n}
  \]
  where $\overline{w}$ denotes the bitwise negation of $w$. Note that for all $n$, $u_n \sqsubseteq u_{n+1}$ and hence there is a unique $u_\infty \in \Sigma^\omega$ s.t. $u_n = u_\infty[:2^n]$. The beginning of this sequence is
  \[
    u_\infty=01101001100101101001011001101001 \dots
  \]

  It's possible to show that for all $n$ and $i<2^n$
  \[
    u_n[i] \equiv \sum_{j=0}^{n-1} \langle i \rangle_2^n [j] \pmod{2}
  \]

  From this, it is easy to see that $u_n[i]$ can be computed by an automaton with two states in either reading direction. Hence, for any $n\ge1$,
  \[
    \mathrm{AC}^{\mathrm{L}}_2(u_n) = \mathrm{AC}^{\mathrm{R}}_2(u_n) = 2.
  \]
\end{example}

See Appendix~\ref{sec:apx-auto} for a few more examples.

\subsection{Tabular Plurality Algorithm}

In the preceding paper \citet{Kosoy2026a} we proved the existence of a statistically and computationally efficient predictor for LTR automaticity. For sequences with low RTL automaticity, we now introduce the \textbf{Tabular Plurality (TP)} algorithm. This algorithm effectively learns the transition table of a DFAO by detecting identities between subsequences corresponding to arithmetic progressions.


The internal state of TP consists of:
\begin{itemize}
\item The time index $n \in \mathbb{N}$.
  \item A \emph{depth} parameter $m^* \in \mathbb{N}$.
  \item For each $0 \leq m \leq m^*$, a set of \emph{states} $Q_m \subseteq [k^m]$, where each $q \in Q_m$ is the minimal known residue of a distinct equivalence class at depth $m$.
  \item For each $0 \leq m < m^*$, a \emph{transition function} $\delta_m : Q_m \times [k] \to Q_{m+1}$, where $\delta_m (q, i)$ is the minimal state in $Q_{m+1}$ consistent with all observations involving the residue class of $q$ and digit $i$.
  \item An \emph{output function} $\tau : Q_{m^*} \to \Sigma$.
\end{itemize}


Tabular Plurality's prediction algorithm is Algorithm~\ref{alg:tp-predict}, which uses the subroutines $\mathsf{Eval}$ (Algorithm~\ref{alg:tp-eval}) and $\mathsf{Diverge}$ (Algorithm~\ref{alg:tp-diverge}). The update algorithm is designed to preserve certain invariants, as captured by the following lemma. We use the notation

\begin{align*}
\rho_0(r)&:=0,
\nonumber\\
\rho_{m+1}(r)&:=
  \delta_m
  \left(
    \rho_m(r\bmod k^m),
    \left\lfloor r/k^m\right\rfloor
  \right).
\end{align*}

\begin{lemma}
\label{lem:model-correct-prelim}
  Consider the state of the predictor after updating on the prefix $x[:n+1]$. The following statements hold.
  \begin{enumerate}
    \item \emph{Data consistency.}
    \[
      \mathsf{Eval}(t)=x[t]
      \qquad\text{for every }t\leq n.
    \]

    \item \emph{Minimality.}  For every $m\leq m^*$ and $q\in Q_m$,
    \[
      q=
      \min\{r\in[k^m] : \rho_m(r)=q\}.
    \]

    \item \emph{Distinguishability.}
    Fix $m\leq m^*$ and distinct states $q_1,q_2\in Q_m$.  Let
    \[
      r:=d_m(q_1,q_2)
    \]
    be their divergence value (as computed by $\mathsf{Diverge}$) in the final state.  Then
    \[
      k^m r+q_1\leq n,
      \qquad
      k^m r+q_2\leq n,
    \]
    and the two observed symbols at these row-$r$ occurrences are different:
    \[
      x[k^m r+q_1]\neq x[k^m r+q_2].
    \]
  \end{enumerate}
\end{lemma}

See Appendix~\ref{sec:tp} for more details.

\begin{algorithm}[htbp]
\caption{Eval($\ell$)}
\label{alg:tp-eval}
  \begin{enumerate*}
    \item $q \leftarrow 0$
    \item \textbf{for} $m = 0$ \textbf{to} $m^* - 1$
      \begin{enumerate*}
        \item $q \leftarrow \delta_m (q, \left\lfloor \ell / k^m \right\rfloor \bmod k)$
      \end{enumerate*}
    \item \textbf{return} $\tau(q)$
  \end{enumerate*}
\end{algorithm}

\begin{algorithm}[htbp]
\caption{Diverge()}
\label{alg:tp-diverge}
  \begin{enumerate*}
    \item \textbf{for} $q, q' \in Q_{m^*}$
      \begin{enumerate*}
        \item \textbf{if} $q = q'$ \textbf{then} $d_{m^*}(q, q') \leftarrow \infty$ \textbf{else} $d_{m^*}(q, q') \leftarrow 0$
      \end{enumerate*}
    \item \textbf{for} $m = m^* - 1$ \textbf{downto} $0$
      \begin{enumerate*}
        \item \textbf{for} $q, q' \in Q_m$
          \begin{enumerate*}
            \item $d_m (q, q') \leftarrow \min_{i < k} (k \cdot d_{m+1}(\delta_m (q,i), \delta_m (q',i)) + i)$
          \end{enumerate*}
      \end{enumerate*}
    \item \textbf{return} $\{d_m\}_{m \leq m^*}$
  \end{enumerate*}
\end{algorithm}

\begin{algorithm}[htbp]
\caption{TP Predict($n$)}
\label{alg:tp-predict}
  \begin{enumerate*}
    \item $m \leftarrow 0$; $q \leftarrow 0$
    \item $q' \leftarrow \delta_0 (q, n \bmod k)$
    \item \textbf{while} $m<m^*-1$ and $q' \equiv n \pmod{k^{m+1}}$
      \begin{enumerate*}
        \item $m \leftarrow m + 1$; $q \leftarrow q'$
        \item $q' \leftarrow \delta_m (q, \left\lfloor n/k^m \right\rfloor \bmod k)$
      \end{enumerate*}
    \item $\{d_m\}\leftarrow\mathsf{Diverge}()$
    \item $\ell \leftarrow \left\lfloor n / k^{m+1} \right\rfloor$
    \item $D \leftarrow \{\tilde{q} \in Q_{m+1} : d_{m+1}(q', \tilde{q}) \geq \ell\}$
    \item \textbf{for} $\tilde{q} \in D$
      \begin{enumerate*}
        \item $\nu(\tilde{q}) \leftarrow \left|\{(\hat{q}, i) \in Q_m \times [k] : \delta_m (\hat{q}, i) = \tilde{q} \text{ and } k^m \cdot i + \hat{q} < n \bmod k^{m+1}\}\right|$
      \end{enumerate*}
    \item $q^* \leftarrow \operatorname*{arg\,max}_{\tilde{q} \in D} \nu(\tilde{q})$
    \item \textbf{return} $\mathsf{Eval}(k^{m+1} \cdot \ell + q^*)$
  \end{enumerate*}
\end{algorithm}

\FloatBarrier

Using Tabular Plurality, we obtain the following theorem.

\begin{theorem}
\label{thm:rtl-bounds}
There exists a predictor (TP) which is statistically and computationally efficient with respect to $\mathrm{AC}^{\mathrm{R}}_k$. Specifically, for any $n\geq2$ and $x \in \Sigma^n$, letting $m = \mathrm{AC}^{\mathrm{R}}_k(x)$ and $\sigma \mathrel{:=} |\Sigma|$, we have:
\begin{enumerate}
  \item \textbf{Statistical Efficiency:} The number of mistakes is bounded by
  \begin{equation}
    M_{\mathrm{TP}}(x) = O(k m \log (k m) \log n)
    \label{eq:rtl-mistakes}
  \end{equation}
  \item \textbf{Computational Efficiency:} The predictor satisfies the compression bound
  \begin{equation}
    S_{\mathrm{TP}}(x) = O(m (k \log m \log n + (\log n)^2 + \log \sigma))
    \label{eq:rtl-compression}
  \end{equation}
\end{enumerate}
\end{theorem}

Here, the predictor receives $k$ as a parameter, and its time complexity is polynomial in the numerical value of $k$. See Appendix~\ref{sec:tp} for the proof.

\section{Arithmetic Repetition Complexity}
\label{sec:arc}

In the preceding paper \citet{Kosoy2026a} we showed that $\mathrm{AC}^{\mathrm{L}}_k$ can be generalized to straight-line complexity ($\operatorname{SLC}$), a base-independent measure that captures concatenative repetition.
Section~\ref{sec:auto} treated the dual, right-to-left automaticity $\mathrm{AC}^{\mathrm{R}}_k$.
The goal of this section is to introduce a base-independent complexity measure for the interleaving coming from the right-to-left structure, while still retaining the concatenative examples covered by $\operatorname{SLC}$.

The resulting object is \emph{Arithmetic Repetition Complexity} ($\operatorname{ARC}$).
Rather than describing a word by a program, we describe it by a small set of literal symbols together with arithmetic copy constraints.
Every non-literal position is justified by a previous position lying in a matching finite arithmetic progression.

\subsection{Arithmetic Repetition Systems}
\label{sec:arithmetic-repetition-systems}

We start by fixing the notation for finite arithmetic progressions.
An \emph{arithmetic interval} is a tuple $\alpha=(r,i,a,b)$ where $r \in \mathbb{N}^+$, $i < r$, and $a,b \in \mathbb{N}$ with $a < b$.
It denotes the strictly increasing partial map with domain $[b-a]$ given by
\begin{equation*}
  \alpha(t) \mathrel{:=} r(a+t) + i.
\end{equation*}
We write $\operatorname{st}(\alpha)\mathrel{:=}r$ for the step and $|\alpha|\mathrel{:=}b-a$ for the length of the progression.
For a word $u \in \Sigma^*$, $u[\alpha]$ denotes the subword obtained by reading $u$ at the positions $\alpha(t)$ for which $\alpha(t) < |u|$, in increasing order of $t$.

\begin{definition}
\label{def:ars}
An \emph{arithmetic repetition system} (ARS) for a word $u \in \Sigma^n$ is a pair $R=(L,\Phi)$ where $L \subseteq \Sigma \times [n]$ and $\Phi$ is a finite set of ordered pairs of arithmetic intervals, satisfying:
\begin{itemize}
  \item For any $(a,i) \in L$, $u[i]=a$.
  \item For any $(\alpha,\beta) \in \Phi$, $|\alpha|=|\beta|$ and $u[\alpha] = u[\beta]$.
  \item For any $i < n$, either $(u[i],i) \in L$, or there are $(\alpha,\beta) \in \Phi$ and $t \in \operatorname{dom}(\alpha)$ such that
  \begin{equation*}
    \alpha(t) < \beta(t) = i.
  \end{equation*}
\end{itemize}

The size of $R$ is defined by $|R|\mathrel{:=}|L|+|\Phi|$.
\end{definition}

The last condition orients each repetition towards the future: a non-literal symbol at position $i$ must be copied from an earlier position belonging to the same arithmetic comparison.
Thus, although the equality condition between $u[\alpha]$ and $u[\beta]$ is symmetric, the ordered pair $(\alpha,\beta)$ records how the system predicts later symbols from earlier ones.

\begin{definition}
\label{def:arc}
The \emph{Arithmetic Repetition Complexity} of a word $x$, denoted $\operatorname{ARC}(x)$, is the minimum size of an arithmetic repetition system for $x$.
\end{definition}

It's easy to see that ARC is a word complexity measure (Proposition~\ref{prp:arc-measure}). See Appendix~\ref{sec:apx-props} for an example of an ARS and a proof of this fact.

\subsection{Upper Bounds on ARC}
\label{sec:upper-bounds-on-arc}

We now show how $\operatorname{ARC}$ is stronger than (i.e. upper bounded by)  a number of other complexity measures. The propositions in this subsection are proved in Appendix~\ref{sec:apx-arc-bounds}.

First, $\operatorname{ARC}$ captures the patterns that Lempel-Ziv 77 compression can detect. Here, $\operatorname{LZC}$ is the number of factors in the minimal LZ77 factorization (see Appendix~\ref{sec:apx-arc-bounds} for the definition; notably, this is one of the complexity measures we considered in \citet{Kosoy2026a}):

\begin{proposition}
\label{prp:arc-vs-lzc}
For every word $x$,
\begin{equation*}
  \operatorname{ARC}(x) \leq \operatorname{LZC}(x).
\end{equation*}
\end{proposition}

Recall the definition of a straight-line program (see also e.g. \citealp{Burgisser1997}):

\begin{definition}
\label{def:slp}
A \emph{straight-line program} (SLP) over $\Sigma$ is $P=(Q,q_0,\delta)$, where $Q$ is a finite set, $q_0 \in Q$ and $\delta:Q\to\overline{Q}^*$, where $\overline{Q}\mathrel{:=}Q \sqcup \Sigma$. We require that
\begin{itemize}
  \item For all $q \in Q$, $|\delta(q)|\geq 2$.
  \item $\delta$ is acyclic, i.e. there are no $p_0,p_1 \dots p_{l-1} \in Q$ s.t. for all $i<l$, $p_{(i+1) \bmod l}$ appears in $\delta(p_i)$.
  \item $q_0$ is the unique element of $Q$ that doesn't appear in $\delta(q)$ for any $q \in Q$.
\end{itemize}

We define $\operatorname{val}_P: \overline{Q} \to \Sigma^*$ recursively by

\begin{itemize}
  \item For $a \in \Sigma$, $\operatorname{val}_P(a)\mathrel{:=}a$.
  \item For $q \in Q$, let $n\mathrel{:=}|\delta(q)|$. Then, $\operatorname{val}_P(q)\mathrel{:=}\operatorname{val}_P(\delta(q)[0])\dots \operatorname{val}_P(\delta(q)[n-1])$.
\end{itemize}

We define the \emph{value} of $P$ as $\operatorname{val}(P)\mathrel{:=}\operatorname{val}_P(q_0)$.
\end{definition}

An SLP is essentially an acyclic directed graph with vertices $\overline{Q}$ and edges $(q,q')$ where $q'$ appears in $\delta(q)$. The outgoing edges of every vertex are ordered (by the position inside $\delta(q)$), $q_0$ is the unique source in $Q$, and the elements of $\Sigma$ are sinks.

In \citet{Kosoy2026a} we considered the complexity measure $\operatorname{SLC}$, which is the size of the smallest straight-line program for a given word. Analogously, here we consider \emph{zipline programs}: a generalization where the concatenation operation is replaced by the $\operatorname{zip}$ operation defined below. Given $u_0, u_1, \dots, u_{r-1} \in \Sigma^*$, define $\operatorname{zip}(u_0,\dots,u_{r-1})$ by round-robin interleaving and stopping just before the first requested symbol that does not exist: if $m\mathrel{:=}\min_{i<r} |u_i|$ and $j<r$ is minimal with $|u_j|=m$, then
\begin{equation}
  |\operatorname{zip}(u_0,\dots,u_{r-1})| \mathrel{:=} r m + j
  \label{eq:zip-length}
\end{equation}
and, for $i < r m + j$,
\begin{equation}
  \operatorname{zip}(u_0,\dots,u_{r-1})[i] \mathrel{:=} u_{i \bmod r}\left[\left\lfloor \frac{i}{r} \right\rfloor\right].
  \label{eq:zip}
\end{equation}

\begin{definition}
\label{def:zlp}
A \emph{zipline program} (ZLP) over $\Sigma$ is $P=(Q,q_0,\delta)$, where $Q$ is a finite set, $q_0 \in Q$ and $\delta:Q\to\overline{Q}^*$, where $\overline{Q}\mathrel{:=}Q \sqcup \Sigma$. We require the same conditions as in Definition~\ref{def:slp}. We define $\operatorname{val}_P: \overline{Q} \to \Sigma^*$ recursively by

\begin{itemize}
  \item For $a \in \Sigma$, $\operatorname{val}_P(a)\mathrel{:=}a$.
  \item For $q \in Q$, let $n\mathrel{:=}|\delta(q)|$. Then, $\operatorname{val}_P(q)\mathrel{:=}\operatorname{zip}(\operatorname{val}_P(\delta(q)[0]), \dots, \operatorname{val}_P(\delta(q)[|\delta(q)|-1]))$.
\end{itemize}

We define the \emph{value} of $P$ as $\operatorname{val}(P)\mathrel{:=}\operatorname{val}_P(q_0)$.
\end{definition}

Thus, a ZLP is the same combinatorial structure as an SLP, but assigned a different interpretation ($\operatorname{val}$). See Appendix~\ref{sec:apx-props} for a few examples and basic properties of ZLP.

The size of a ZLP is defined as the number of edges in the graph, i.e. $|P|\mathrel{:=}\sum_{q \in Q} |\delta(q)|$. This allows us to define zipline complexity.

\begin{definition}
\label{def:zlc}
The \emph{zipline complexity} of a word $x$, denoted $\operatorname{ZLC}(x)$, is the minimum size of a zipline program $P$ such that $x \sqsubseteq \operatorname{val}(P)$.
\end{definition}

It's easy to see that $\operatorname{ZLC}$ is a word complexity measure. Notably, we don't know whether there exists a statistically efficient and computationally feasible predictor for $\operatorname{ZLC}$. However, the following proposition together with Theorem~\ref{thm:arc-muwu} below implies that there is one for $\sqrt[3]{n} \cdot \operatorname{ZLC}(x)$.

\begin{proposition}
\label{prp:arc-vs-zlc}
For every $n \in \mathbb{N}$ and $x \in \Sigma^n$,
\begin{equation*}
  \operatorname{ARC}(x) \leq 2 \left(\sqrt[3]{2 n}+1\right) \cdot \operatorname{ZLC}(x).
\end{equation*}
\end{proposition}

We also leave it as a (purely combinatorial) open problem to establish whether the factor $O(\sqrt[3]{n})$ in the above is tight. There is, however, a natural subclass of ZLP for which this factor can be avoided.

\begin{definition}
\label{def:sync-zlp}
A ZLP $P=(Q,q_0,\delta)$ is called \emph{synchronous} when there exists a function $\phi: Q \to \mathbb{N}$ such that:
\begin{itemize}
  \item $\phi(q_0) = 1$
  \item For every $q,q' \in Q$, if $q'$ appears in $\delta(q)$ then $\phi(q') = |\delta(q)| \cdot \phi(q)$.
\end{itemize}

The \emph{synchronous zipline complexity} of a word $x$, denoted $\operatorname{sZLC}(x)$, is the minimum size of a synchronous ZLP $P$ s.t. $x \sqsubseteq \operatorname{val}(P)$.
\end{definition}

In other words, for any vertex $q$ and two paths from the root $q_0$ to $q$, the product of vertex out-degrees along these paths has to be the same.

The intuition behind this definition can be understood when thinking of how the structure of $\operatorname{val}(P)$ depends on $P$. Every $q \in Q$ corresponds to a collection of arithmetic progression subsequences of $\operatorname{val}(P)$ which all look like prefixes of $\operatorname{val}_P(q)$. Each individual subsequence corresponds to a particular path from $q_0$ to $q$. The step of the corresponding arithmetic progression is the product of the out-degrees along the path. Hence, the synchronicity condition ensures that all progressions corresponding to a given $q \in Q$ have the same step.

In contrast to Proposition~\ref{prp:arc-vs-zlc}, we have,

\begin{proposition}
\label{prp:arc-vs-szlc}
For every $n \in \mathbb{N}$ and $x \in \Sigma^n$,
\begin{equation*}
  \operatorname{ARC}(x) \leq \operatorname{sZLC}(x).
\end{equation*}
\end{proposition}

As the following proposition demonstrates, the zip operation plays, for $\mathrm{AC}^{\mathrm{R}}_k$, a role analogous to the role played by concatenation in SLP for $\mathrm{AC}^{\mathrm{L}}_k$.\footnote{See Proposition 4.3 of \citet{Kosoy2026a}.} In particular, ARC accomplishes the goal of generalizing $\mathrm{AC}^{\mathrm{R}}_k$ in a way that removes the dependence on $k$.

\begin{proposition}
\label{prp:zlc-vs-acr}
For every base $k \geq 2$, length $n\geq2$ and word $x \in \Sigma^n$,
\begin{equation*}
  \operatorname{sZLC}(x) \leq k\lceil\log_k n\rceil \cdot \mathrm{AC}^{\mathrm{R}}_k(x).
\end{equation*}
\end{proposition}

$\operatorname{ARC}$ is a substantially more expressive complexity measure than $\operatorname{SLC}$ or $\mathrm{AC}^{\mathrm{R}}_k$. In particular, in Appendix~\ref{sec:apx-mix-automatic} we show that the so-called \emph{mix-automatic} sequences have polylogarithmic $\operatorname{sZLC}$ (and therefore $\operatorname{ARC}$) growth.

\subsection{Prediction from ARC}
\label{sec:prediction-from-arc}

The main algorithmic result is that $\operatorname{ARC}$ admits a statistically efficient and computationally feasible predictor. The algorithm is based on methods developed in \cite{Joshi2026} for ``multivalued" online learning. We do not know whether a computationally \emph{efficient} (in the sense of Definition~\ref{def:comp-eff}, i.e. including a compression bound) predictor is possible: this is left as an open problem\footnote{That said, in the following papers in the series, we will show that there is a non-trivial class of ZLP for which the analogue of ZLC \emph{does} admit a computationally efficient predictor.}.

\begin{theorem}
\label{thm:arc-muwu}
There exists a statistically efficient and computationally feasible predictor (ARC-MUWU) for $\operatorname{ARC}$.
For any $n \geq 2$ and $x \in \Sigma^n$, denote $m\mathrel{:=}\operatorname{ARC}(x)$.
Then ARC-MUWU satisfies
\begin{equation}
  M_{\mathrm{ARC-MUWU}}(x) = O(m \log n).
  \label{eq:arc-muwu-mistakes}
\end{equation}
\end{theorem}

See Appendix~\ref{sec:apx-muwu-arc} for the proof.

\section*{Acknowledgments}

This work was supported by the Advanced Research+Invention Agency (ARIA) of the United
Kingdom, Survival and Flourishing Corp,
and Coefficient Giving in San Francisco, California. The author wishes to thank Vinayak Pathak for reviewing a draft and providing substantive suggestions. She also thanks Sarah Lurie for suggestions on improving the writing style. 

\bibliographystyle{plainnat}
\bibliography{main}

\appendix
\section{Tabular Plurality}
\label{sec:tp}

In this appendix, we describe the Tabular Plurality (TP) algorithm and prove Theorem~\ref{thm:rtl-bounds}. The algorithm operates by maintaining an approximate model of the unknown DFAO that computes $x[t]$ from the reversed (LSDF) base-$k$ representation of $t$. The model is organized as a layered transition table, indexed by digit depth. On each observation, the algorithm corrects inconsistent transitions; predictions are made by a plurality vote over candidate states.

We begin with an auxiliary combinatorial result on sequential prediction with branching structure, then present the algorithm in full, and conclude with the proof of Theorem~\ref{thm:rtl-bounds}.

\subsection{Branch Prediction}
\label{sec:branch-pred}

The following game captures the combinatorial core of the TP mistake analysis.

\begin{definition}
\label{def:bp-game}
  The \emph{branch prediction game} is a sequential game proceeding as follows.
  \begin{itemize}
    \item The state consists of a finite set $E$ of \emph{elements}, a rooted tree $T = (V, E_T)$, and a mapping $f : E \to V$ assigning each element to a vertex.
    \item Initially, $E = \emptyset$ and $V = \{v_0\}$ (a single root vertex).
    \item On each round, nature either:
    (a) \emph{introduces} a new element $e$: nature may first extend $T$ by attaching a (possibly empty) chain of new vertices $u_1, \dots, u_p$, with each $u_{i+1}$ a child of $u_i$, where $u_0$ is a pre-existing vertex. Nature then assigns $f(e)$ to $u_p$.
    (b) \emph{activates} an existing element $e \in E$:
      \begin{enumerate}
        \item If $f(e)$ is a leaf, nature extends $T$ by attaching a new child $u$ to $f(e)$, and sets $f(e) \leftarrow u$. No prediction is made.
        \item If $f(e)$ is an internal vertex, the learner selects a child $a$ of $f(e)$ as its prediction. Nature then either selects an existing child $b$ of $f(e)$ or attaches a new child $b$ to $f(e)$. In either case, $f(e) \leftarrow b$. A \emph{mistake} occurs if $a \neq b$.
      \end{enumerate}
  \end{itemize}
\end{definition}

Because activation always moves an element strictly downward, once an element enters the subtree rooted at a vertex $u$ it remains in that subtree for the remainder of the game. In particular, the quantity
\begin{equation*}
  N(u) \mathrel{:=} \left|\{e \in E : f(e) \text{ is in the subtree rooted at } u\}\right|
\end{equation*}
is non-decreasing over time for every $u \in V$.

We propose the following prediction strategy for this game.

\begin{definition}
\label{def:mb}
  The \emph{Plurality Branch} (PB) algorithm operates as follows. Whenever a prediction is required at an internal vertex $v = f(e)$ (case (b.2) of Definition~\ref{def:bp-game}), the learner predicts the child $c$ of $v$ maximizing $N(c)$ --- that is, the child whose current subtree contains the most elements. Ties are broken arbitrarily.
\end{definition}

Note that the element $e$ being activated sits at $v$ itself, not in any child's subtree, and so does not contribute to any $N(c)$.

We now state the main result of this subsection.

\begin{proposition}
\label{prop:branch-pred}
  Under the PB algorithm, the total number of mistakes in the branch prediction game (Definition~\ref{def:bp-game}) is at most $3 |E| \left\lfloor \log_2 |E| \right\rfloor$, where $|E|$ is the final number of elements.
\end{proposition}

\begin{proof}
  We augment the game with \emph{phantom traversals}. Whenever a new element $e$ is introduced at a vertex $v \neq v_0$, we retroactively declare that $e$ has traversed the entire root-to-$v$ path --- that is, at each vertex $w$ on this path, we add a phantom selection of the child leading toward $v$. Phantom traversals are purely an analytical device: they do not affect PB, which predicts based on the actual subtree counts $N(c)$.

  For each vertex $v$ with children $u_1, u_2, \dots$, let $c_i$ denote the total number of augmented (real plus phantom) selections of $u_i$ at $v$ over the course of the game. Equivalently, $c_i$ is the total number of elements whose augmented root-to-final-vertex path passes through $u_i$. Assume the children are sorted so that $c_1 \geq c_2 \geq \dots$, and call $u_1$ the \emph{dominant branch} at $v$.

  Observe that at any time $t$, $N_t (u_i)$ equals the number of elements present at time $t$ whose augmented path passes through $u_i$. Hence $N_t (u_i) \leq c_i$, with equality at the end of the game.

  We bound the mistakes at $v$ by separating them by nature's choice $b$:

  \begin{itemize}
    \item \emph{Nature picks a non-dominant branch ($b = u_i$ for some $i \geq 2$).} Each such real selection is a real (non-phantom) traversal into $u_i$, and there are at most $c_i$ such traversals over the course of the game. Hence the total number of such rounds at $v$ is at most $\sum_{i \geq 2} c_i$; each contributes at most one mistake.
    \item \emph{Nature picks the dominant branch ($b = u_1$).} If PB predicts $a \neq u_1$, then by definition of PB, $N_t (a) \geq N_t (u_1)$ at the time of the round. Since $a$ is non-dominant, $N_t (a) \leq c_a \leq c_2$, whence $N_t (u_1) \leq c_2$. After the round $N(u_1)$ increases by $1$. Hence at most $c_2 + 1$ rounds of this type can occur: once $N(u_1) > c_2$, we have $N(u_1) > c_i$ for every non-dominant $u_i$, so PB predicts $u_1$ and no further mistake of this type is possible.
  \end{itemize}

  Summing, mistakes at $v$ $\leq \sum_{i \geq 2} c_i + c_2 + 1 \leq 2 \sum_{i \geq 2} c_i + 1$. If there are any mistakes at $v$ then $c_2 \geq 1$ (otherwise $v$ has a single child which implies no mistakes), so $2 \sum_{i \geq 2} c_i + 1 \leq 3 \sum_{i \geq 2} c_i$. The right-hand side is 3 times the number of augmented non-dominant traversals at $v$.

  It remains to bound the total number of augmented non-dominant traversals across all vertices. For each element $e \in E$, define the \emph{cohort} of $e$ at a vertex $v$ on its augmented root-to-final-vertex path as the set of elements whose augmented path also passes through $v$. Since every element's augmented path begins at the root, the cohort at $v_0$ has size $|E|$. When $e$ makes a non-dominant traversal at $v$ (real or phantom), it follows a child through which at most half the cohort passes (the dominant child accounts for at least as many as any other). Hence the cohort of $e$ drops by a factor of at least $2$ at every non-dominant traversal. Since the cohort starts at $|E|$ and is always at least $1$, each element contributes at most $\left\lfloor \log_2 |E| \right\rfloor$ non-dominant augmented traversals. Summing, the total number of augmented non-dominant traversals is at most $|E| \left\lfloor \log_2 |E| \right\rfloor$.

  Applying the factor-of-$3$ overhead, the total number of mistakes is at most $3 |E| \left\lfloor \log_2 |E| \right\rfloor$.
\end{proof}

The following matching lower bound shows that Proposition~\ref{prop:branch-pred} is tight up to constant factors.

\begin{proposition}
\label{prop:branch-lower}
  For any $t\geq1$, nature has a strategy in the branch prediction game with $|E| = t$ that forces $\Omega(t \log t)$ mistakes against any deterministic prediction algorithm.
\end{proposition}

\begin{proof}
  Nature's strategy is as follows. All $t$ elements are introduced at the root $v_0$. Nature then processes vertices iteratively: while there exists a vertex $v$ with at least $2$ elements present (i.e., $f(e) = v$ for $\geq 2$ elements $e$), nature activates these elements one by one.

  \begin{itemize}
    \item The first activated element $e_1$ finds $v$ is a leaf (case (b.1)): nature attaches a child $c_1$ and sets $f(e_1) \leftarrow c_1$. No prediction is made.
    \item The second activated element $e_2$ finds $v$ internal with one child $c_1$: the learner is forced to predict $c_1$, and nature attaches a new child $c_2$ and sets $f(e_2) \leftarrow c_2$. This is a mistake.
    \item Each subsequent activated element finds $v$ internal with children $c_1, c_2$: the learner predicts one child, and nature selects the other. This is a mistake.
  \end{itemize}

  Hence every activation at an internal vertex produces a mistake (the second activation forces a mistake because nature creates a new child; every subsequent activation forces a mistake because nature picks the child the learner did not predict). Because all elements are introduced at the root and never elsewhere, no phantom traversals arise and nature's strategy does not depend on the specific prediction rule --- the argument applies to any deterministic learner, PB included.

  For each vertex $v$ processed by nature, let $R(v) \mathrel{:=} \left|\{e \in E : f(e) = v\}\right|$ denote the number of elements located at $v$ at the moment nature begins processing $v$. Since $v$ is selected only when $R(v) \geq 2$, nature forces exactly $R(v) - 1$ mistakes at $v$ (all activations except the first, which encounters a leaf). After all elements at $v$ have been moved to children, nature repeats at vertices where elements have accumulated. The process terminates when every element is the sole occupant of its vertex.

  The resulting structure is a binary tree $T^*$ whose internal nodes correspond to vertices where splitting occurred and whose $t$ leaves correspond to the final positions of the $t$ elements. The tree structure depends on the learner's predictions (which control how elements are distributed between $c_1$ and $c_2$ at each vertex), but every internal node $v$ contributes exactly $R(v) - 1$ mistakes regardless. The total number of mistakes is therefore $\sum_{v \in I} (R(v) - 1)$, where $I$ denotes the set of internal nodes of $T^*$.

  It remains to show this sum is $\Omega(t \log t)$ for every possible learner strategy. Define $M : \mathbb{N} \to \mathbb{N}$ as the minimum total number of mistakes achievable by any learner when all $r$ elements start at a common vertex and nature plays the above strategy. Then $M(1) = 0$ (a single element is never split). For $r \geq 2$, after nature processes the root (forcing $r - 1$ mistakes), the learner's predictions determine a split of the $r$ elements into two groups of sizes $r_1$ and $r_2 = r - r_1$ sent to children $c_1$ and $c_2$ respectively, with $r_1, r_2 \geq 1$ (child $c_1$ receives at least $e_1$ and child $c_2$ receives at least $e_2$). Nature then recurses on each child. The learner chooses the split to minimize total mistakes, giving
\[
  M(r) = (r - 1) + \min_{r_1 + r_2 = r, \,r_1 \geq 1, \,r_2 \geq 1} [M(r_1) + M(r_2)]
\]

  \textbf{Claim:} $M(r) \geq 1/2 \,r \log_2 r$ for all $r \geq 1$.

  We prove this by strong induction. The base cases $r = 1$ and $r = 2$ are immediate: $M(1) = 0 \geq 0$ and $M(2) = 1 \geq 1$. Now let $r \geq 3$ and suppose $M(r') \geq 1/2 \,r' \log_2 r'$ for all $1 \leq r' < r$. Consider an arbitrary split $r_1 + r_2 = r$ with $r_1, r_2 \geq 1$. We distinguish two cases.

  \textbf{Case 1:} $r_1 \geq 2$ and $r_2 \geq 2$. By the inductive hypothesis,
\[
  M(r_1) + M(r_2) \geq 1/2(r_1 \log_2 r_1 + r_2 \log_2 r_2)
\]
  Since $\phi(x) = x \log_2 x$ is convex on $(0, \infty)$ (its second derivative is $1/(x \ln 2) > 0$), applying Jensen's inequality to the uniform distribution over the two values $r_1, r_2$ gives
\[
  1/2(r_1 \log_2 r_1 + r_2 \log_2 r_2) \geq (r/2) \log_2 (r/2)
\]
  and therefore $r_1 \log_2 r_1 + r_2 \log_2 r_2 \geq r \log_2(r/2) = r \log_2 r - r$. Hence,
\[
  M(r) \geq (r - 1) + 1/2(r \log_2 r - r) = 1/2 r \log_2 r + (r/2 - 1) \geq 1/2 r \log_2 r
\]
  where the final inequality uses $r \geq 2$.

  \textbf{Case 2:} $\min(r_1, r_2) = 1$. Without loss of generality, $r_1 = 1$ and $r_2 = r - 1 \geq 2$. Then $M(1) = 0$ and by the inductive hypothesis,
\[
  M(r) \geq (r - 1) + 1/2(r - 1) \log_2(r - 1) = (r - 1) \left(1 + 1/2 \log_2(r - 1)\right)
\]
  We claim that for all $r \geq 3$,
\begin{equation}
  (r - 1)(1 + 1/2 \log_2(r - 1)) \geq 1/2 r \log_2 r
  \label{eq:case2-ineq}
\end{equation}
  To verify this, define $h : [3, \infty) \to \mathbb{R}$ by $h(r) \mathrel{:=} (r - 1)(2 + \log_2(r-1)) - r \log_2 r$, so that \eqref{eq:case2-ineq} is equivalent to $h(r) \geq 0$. We check the boundary and monotonicity.

  \emph{Boundary:} $h(3) = 2(2 + 1) - 3 \log_2 3 = 6 - 3 \log_2 3 \approx 1.25 > 0$.

  \emph{Monotonicity:} For $r > 3$, differentiating (treating $\log_2$ as $\ln/\ln 2$),
\[
  h'(r) = (2 + \log_2(r-1)) + (r-1) \cdot 1/((r-1) \ln 2) - \log_2 r - r \cdot 1/(r \ln 2)
\]
\[
  = 2 + \log_2(r - 1) - \log_2 r + 1/(\ln 2) - 1/(\ln 2)
\]
\[
  = 2 + \log_2 \left((r-1)/r\right) = 2 - \log_2 \left(r/(r-1)\right)
\]
  For $r \geq 3$, we have $r/(r-1) \leq 3/2$, so $h'(r) \geq 2 - \log_2(3/2) \approx 2 - 0.585 > 0$. Hence $h$ is strictly increasing on $[3, \infty)$, and since $h(3) > 0$, we conclude $h(r) > 0$ for all $r \geq 3$.

  In both cases, $M(r) \geq 1/2 \,r \log_2 r$, completing the induction. We conclude that the total number of mistakes is at least $M(t) \geq 1/2 \,t \log_2 t = \Omega(t \log t)$.
\end{proof}

\begin{remark}
\label{rem:branch-known-tree}
  If the learner had oracle access to the \textbf{final} tree structure, it could instead predict the child of $v$ whose subtree contains the most \textbf{eventual} leaves. The cohort-halving argument would then bound each element's non-dominant traversals by $\left\lfloor \log_2 |V_{\mathrm{leaf}}| \right\rfloor$ rather than $\left\lfloor \log_2 |E| \right\rfloor$, yielding a bound of $O(|E| \log |V_{\mathrm{leaf}}|)$. PB as defined above uses only the current subtree populations, not the eventual ones; this is why the bound depends on $|E|$ rather than $|V_{\mathrm{leaf}}|$. In our application to the TP algorithm, the tree structure is not known in advance, and Proposition~\ref{prop:branch-pred} is the relevant bound.
\end{remark}

\subsection{Algorithm Description}
\label{sec:tp-alg}

We now describe the Tabular Plurality (TP) algorithm. Throughout, we fix a base $k \geq 2$ and an alphabet $\Sigma$. The algorithm receives $k$ as a parameter.

\subsubsection{Structural Intuition}

The algorithm's state is best viewed as a candidate DFAO for the unknown sequence, organized to mirror the right-to-left (least-significant-digit-first) reading order. After reading the $m$ low-order base-$k$ digits of a time index $t$ --- equivalently, after fixing $t \bmod k^m$ --- the DFAO occupies some state at depth $m$, and this state summarizes everything the model can predict about $x[t]$ from this prefix alone. Two residues $q, q' \in [k^m]$ are identified with a common state at depth $m$ exactly when the observed data cannot distinguish them: the arithmetic subsequences $(x[k^m j + q])_j$ and $(x[k^m j + q'])_j$ agree at every $j$ for which both indices have been observed. The set $Q_m$ records one representative residue per current equivalence class; the transition table $\delta_m$ records how a class refines once the next digit is appended; and the output map $\tau$ at depth $m^*$ closes off the construction.

The evolution of this state is governed by three invariants, established jointly in Lemma~\ref{lem:model-correct}:

\begin{itemize}
  \item \textbf{Data consistency:} the current model agrees with every observation seen so far --- no observed datum ever falsifies the model.
  \item \textbf{Minimality:} among the residues identified with a state $q \in Q_m$, $q$ is labeled by the smallest. This canonicalizes state identity, so that each equivalence class has a unique identifier, and is enforced both at state creation and at every transition reassignment.
  \item \textbf{Distinguishability:} two states $q_1, q_2 \in Q_m$ are kept distinct only when the data demands it, in the form of a pair of observed indices at which they yield different symbols.
\end{itemize}

These invariants jointly determine the algorithm's behaviour. Data consistency drives the corrective updates: a transition is rerouted whenever it would falsify a fresh observation. Distinguishability makes the algorithm conservative about splitting classes, defaulting to merge whenever the existing states remain compatible with the data; this in turn caps the size of $Q_m$ at the state count of any DFAO witnessing $\mathrm{AC}^{\mathrm{R}}_k (x)$ (Lemma~\ref{lem:state-bound}). Minimality fixes a unique canonical name per class.

The transitions $\delta_m$ may nevertheless misroute time indices not yet observed, since several states in $Q_{m+1}$ can be consistent with the available evidence. This ambiguity is resolved at prediction time by a plurality vote among the candidates currently consistent with the data; the rationale is deferred to the Predict subsubsection below.

\subsubsection{Subroutines}

Before describing the main operations, we define two helper subroutines.

\textbf{Evaluation.} Given the current state $(m^*, \{Q_m\}, \{\delta_m\}, \tau)$, the subroutine $\mathsf{Eval}(\ell)$ evaluates the current model at round number $\ell \in \mathbb{N}$ (see Algorithm~\ref{alg:tp-eval}).

\textbf{Divergence.} The subroutine $\mathsf{Diverge}()$ computes, for each $m \leq m^*$ and each pair $q, q' \in Q_m$, the \emph{divergence index} $d_m (q, q') \in \mathbb{N} \sqcup\{\infty\}$: the smallest non-negative integer $r$ such that $\mathsf{Eval}(k^m \cdot r + q) \neq \mathsf{Eval}(k^m \cdot r + q')$ for $q \neq q'$, or $\infty$ if $q = q'$. Concretely, $d_m (q, q')$ is computed bottom-up (see Algorithm~\ref{alg:tp-diverge}).

The key property is that $d_m (q, q') \geq \ell$ asserts that $\mathsf{Eval}(k^m \cdot r + q) = \mathsf{Eval}(k^m \cdot r + q')$ for all $r < \ell$. That is, the subsequences corresponding to $q$ and $q'$ coincide on the first $\ell$ indices. The divergence index is the tool by which the algorithm determines whether two states are ``consistent": $d_{m+1}(\tilde{q}, q') \geq \ell$ means $\tilde{q}$ and $q'$ have been indistinguishable on all observations prior to the current one.

\begin{remark}
\label{rem:div}
  The divergence tables are \emph{not} stored as part of the predictor's persistent state; they are recomputed from scratch before each call to Update or Predict. This avoids the cost of incremental maintenance and yields a stronger compression bound, at the expense of a (polynomial) increase in per-round time.
\end{remark}

\subsubsection{Initialization}

We initialize the state of the predictor according to Algorithm~\ref{alg:tp-init}. The single state $0 \in Q_0$ is trivially the minimal (and only) residue modulo $k^0 = 1$.

\begin{algorithm}[htbp]
\caption{TP Initialize($a$) --- where $a = x[0]$}
\label{alg:tp-init}
  \begin{enumerate*}
    \item $m^* \leftarrow 0$
    \item $Q_0 \leftarrow \{0\}$
    \item $\tau(0) \leftarrow a$
  \end{enumerate*}
\end{algorithm}

\subsubsection{Preprocessing}

When the sequence index reaches a new power of $k$, the depth $m^*$ must increase. This is handled by the subroutine $\mathsf{Preprocess}$ (Algorithm~\ref{alg:tp-preprocess}), called at the beginning of each Update step.

When a new layer is created, each existing state $q \in Q_{m^* - 1}$ is promoted to the new bottom level $Q_{m^*}$. The transition $\delta_{m^* - 1}(q, 0) = q$ is correct because appending digit $0$ does not change the residue: $0 \cdot k^{m^*} + q \equiv q \pmod{k^{m^*}}$. For digits $i \geq 1$, the true target state $i k^{m^*} + q$ has not yet been observed (since no index $t \geq k^{m^*}$ has appeared), so the algorithm defaults to $\delta_{m^* - 1}(q, i) = 0$, the smallest state in $Q_{m^*}$. This is mandated by the minimal-consistent-transition principle: with no observations to distinguish any states at these positions, all states in $Q_{m^*}$ are equally consistent, and we choose the smallest.

\begin{algorithm}[htbp]
\caption{TP Preprocess($n$) --- called before processing index $n$}
\label{alg:tp-preprocess}
  \begin{enumerate*}
    \item \textbf{if} $n = k^{m^*}$ \textbf{then}
      \begin{enumerate*}
        \item $m^* \leftarrow m^* + 1$
        \item $Q_{m^*} \leftarrow Q_{m^* - 1}$
        \item \textbf{for} $q \in Q_{m^*}$
          \begin{enumerate*}
            \item $\delta_{m^* - 1}(q, 0) \leftarrow q$
            \item \textbf{for} $i = 1$ \textbf{to} $k - 1$ \textbf{do} $\delta_{m^* - 1}(q, i) \leftarrow 0$
          \end{enumerate*}
      \end{enumerate*}
  \end{enumerate*}
\end{algorithm}

\subsubsection{Predict}

See Algorithm~\ref{alg:tp-predict}. To predict $x[n]$, the algorithm proceeds in two phases. The first phase descends through the depth levels along the base-$k$ digits of $n$, following the model's transitions as long as they coincide with the canonical level-$(m+1)$ representative of $n$; the descent halts at the first level where this fails. The second phase resolves the resulting ambiguity by a plurality vote among the level-$(m+1)$ candidates currently indistinguishable on the basis of past data.

\textbf{Phase 1 (exact-residue walk).} Lines 1--5 process the base-$k$ digits of $n$ in LSDF order. At the start of each iteration, $q \in Q_m$ is the canonical representative of the level-$m$ class of $n$ --- the smallest residue in $[k^m]$ that the model identifies with $n \bmod k^m$ --- and $q' \mathrel{:=} \delta_m (q, \left\lfloor n/k^m \right\rfloor \bmod k) \in Q_{m+1}$ is the model's level-$(m+1)$ target. If $q' \equiv n \pmod{k^{m+1}}$, then by minimality $q' = n \bmod k^{m+1}$ is itself the canonical representative at the next level, and the descent continues. Otherwise $q' < n \bmod k^{m+1}$, so the model identifies the level-$(m+1)$ class of $n$ with a class whose canonical representative is the smaller $q'$ --- a tentative choice that past observations may not yet uniquely determine, and which Phase 2 reconsiders.

\textbf{Phase 2 (plurality vote).} Set $\ell \mathrel{:=} \left\lfloor n / k^{m+1} \right\rfloor$. Two states $\tilde{q}, q' \in Q_{m+1}$ are \emph{compatible at depth $\ell$} when $d_{m+1}(q', \tilde{q}) \geq \ell$ --- that is, when their associated arithmetic-progression subsequences $j \mapsto \mathsf{Eval}(k^{m+1} j + q')$ and $j \mapsto \mathsf{Eval}(k^{m+1} j + \tilde{q})$ agree for all $j < \ell$, an equality which by Data Consistency is the same as agreement of the corresponding past observations of $x$. Line 8 collects all states in $Q_{m+1}$ compatible with $q'$ into the set $D$ (which automatically contains $q'$). Each $\tilde{q} \in D$ is, on the evidence available so far, an equally valid candidate for the level-$(m+1)$ class of $n$.

For each $\tilde{q} \in D$, line 10 computes a population $\nu(\tilde{q})$, defined as the number of pairs $(\hat{q}, i) \in Q_m \times [k]$ for which $\delta_m (\hat{q}, i) = \tilde{q}$ --- meaning the model has placed the level-$(m+1)$ residue $\hat{r} \mathrel{:=} i k^m + \hat{q}$ in $\tilde{q}$'s class --- subject to the additional constraint $\hat{r} < n \bmod k^{m+1}$. This constraint isolates the level-$(m+1)$ residues which, by time $n$, have been carried one step further into their arithmetic progression $j \mapsto k^{m+1} j + \hat{r}$ than $n \bmod k^{m+1}$ has. The plurality winner $q^* \mathrel{:=} \operatorname*{arg\,max}_{\tilde{q} \in D} \nu(\tilde{q})$ is the level-$(m+1)$ candidate supported by the largest cohort of similarly-situated past observations.

Line 12 returns $\mathsf{Eval}(k^{m+1} \ell + q^*)$. The chosen index has the same level-$(m+1)$ ``high part" $\ell$ as $n$, but with $q^*$ in place of the unobserved residue $n \bmod k^{m+1}$; since (as the proof shows) $q^* < n \bmod k^{m+1}$, the index lies strictly below $n$, and Data Consistency makes the evaluation coincide with the already-observed value $x[k^{m+1} \ell + q^*]$. The rationale for the plurality rule is the reduction to the Plurality Branch strategy of Definition~\ref{def:mb}: the candidates in $D$ correspond to the children of $n$'s vertex in a level-$m$ refinement tree of subsequence-equivalence classes, and the population $\nu(\tilde{q})$ corresponds to the subtree population under the child indexed by $\tilde{q}$. The formal coupling, and the resulting mistake bound, are the content of Section~\ref{sec:tp-proof}.

\subsubsection{Update}

To absorb the new observation $a = x[n]$, the Update step restores Data Consistency by walking through the depth levels in increasing order and applying the smallest local repair compatible with both past data and the new observation. After Preprocess extends the depth $m^*$ if $n$ has reached a new power of $k$, and Diverge freshly recomputes the divergence tables on which the consistency tests below will be based, the main loop processes levels $m = 0, \dots, m^* - 1$ in turn. At each level, three sub-questions are asked: whether $n$'s level-$m$ class is currently represented in $Q_m$ at all (line 4); whether the model's level-$(m+1)$ target is inconsistent with the new observation (line 6); and, if so, whether the inconsistency can be repaired by rerouting to an existing alternative target (line 7) or whether a new state must be created (lines 10--15).

\textbf{Detection (lines 4--6).} The line-4 guard succeeds exactly when $n \bmod k^m \in Q_m$, in which case Minimality forces $q \mathrel{:=} n \bmod k^m$ to be the unique selection. Set $i \mathrel{:=} \left\lfloor n/k^m \right\rfloor \bmod k$ (the next base-$k$ digit), $\ell \mathrel{:=} \left\lfloor n/k^{m+1} \right\rfloor$ (the level-$(m+1)$ high part of $n$), and $q' \mathrel{:=} \delta_m (q, i) \in Q_{m+1}$ (the model's current level-$(m+1)$ target). If the guard fails, no anchor exists at this level and the iteration is silently skipped. Otherwise, line 6 tests whether the current transition is genuinely inconsistent. The first conjunct $q' \not\equiv n \pmod{k^{m+1}}$ excludes the case where the model already encodes $n$'s level-$(m+1)$ residue. When the conjunct holds, $q' < n \bmod k^{m+1}$ instead, and $k^{m+1} \ell + q' < n$ is a previously observed index; by Data Consistency, $\mathsf{Eval}(k^{m+1} \ell + q') = x[k^{m+1} \ell + q']$, and the second conjunct asks whether this past value disagrees with $a$. When both conjuncts hold, the new observation has refuted the model's identification of $n$'s class with that of $q'$, and a corrective action at level $m$ is required.

\textbf{Correction by reassignment (lines 7--8).} The preferred remedy is to reroute $\delta_m (q, i)$ to an alternative state $\tilde{q} \in Q_{m+1}$ satisfying two compatibility conditions: $d_{m+1}(\tilde{q}, q') \geq \ell$, which by the divergence semantics ensures that $\tilde{q}$ and $q'$ agree on all observations made through index $\ell - 1$, so that the rerouting does not retroactively break Data Consistency; and $\mathsf{Eval}(k^{m+1} \ell + \tilde{q}) = a$, which certifies that $\tilde{q}$ correctly predicts the new observation. If such a $\tilde{q}$ exists, line 8 commits the reassignment.

\textbf{Correction by new-state creation (lines 9--15).} If no compatible alternative exists, the data has separated $n$'s class from every state in $Q_{m+1}$ compatible with $q'$. The algorithm responds by introducing $q_{\mathrm{new}} \mathrel{:=} n \bmod k^{m+1}$ --- the canonical minimal representative of $n$'s level-$(m+1)$ class --- adding it to $Q_{m+1}$, and committing $\delta_m (q, i) \leftarrow q_{\mathrm{new}}$. To preserve Data Consistency on past indices previously routed through $q'$, the outgoing structure of $q_{\mathrm{new}}$ is initialized to make it evaluate identically to $q'$. When $m + 1 < m^*$ (line 14), this is achieved by inheriting transitions verbatim, $\delta_{m+1}(q_{\mathrm{new}}, j) \leftarrow \delta_{m+1}(q', j)$ for every digit $j$: the two states then unfold to the same level-$(m+2)$ targets and yield identical predictions everywhere. When $m + 1 = m^*$ (line 15), the new state lies on the bottom layer with no transitions to copy, and we instead anchor it directly with $\tau(q_{\mathrm{new}}) \leftarrow a$.

A new-state creation at an interior level $m + 1 < m^*$ does not, by itself, restore Data Consistency at $n$: since $q_{\mathrm{new}}$ inherits its outgoing structure from $q'$, $\mathsf{Eval}(n)$ now routes through $q_{\mathrm{new}}$ and produces the same wrong value as before. On iteration $m + 1$, however, the line-4 guard succeeds via $q_{\mathrm{new}} \in Q_{m+1}$, and the line-6 test at the just-inherited transition is again inconsistent --- triggering another corrective action. Thus a single mismatch cascades downward through consecutive levels --- propagating new-state creations along the residue chain $n \bmod k^{m+1}, n \bmod k^{m+2}, \dots$ --- until either an existing-state reassignment succeeds or the cascade reaches the bottom layer, where line 15 anchors the leaf with $\tau(q_{\mathrm{new}}) \leftarrow a$. The proof of Lemma~\ref{lem:model-correct} verifies that the loop indeed terminates with all three invariants restored.

\begin{algorithm}[htbp]
\caption{TP Update($n$, $a$) --- where $a = x[n]$}
\label{alg:tp-update}
  \begin{enumerate*}
    \item $\mathsf{Preprocess}(n)$
    \item $\{d_m\}\leftarrow\mathsf{Diverge}()$
    \item \textbf{for} $m = 0$ \textbf{to} $m^* - 1$
    \item \quad \textbf{if} can select $q \in Q_m$ with $q \equiv n \pmod{k^m}$ \textbf{then}
    \item \quad\quad $i \leftarrow \left\lfloor n / k^m \right\rfloor \bmod k$; \quad  $\ell \leftarrow \left\lfloor n / k^{m+1} \right\rfloor$; \quad  $q' \leftarrow \delta_m (q, i)$
    \item \quad\quad \textbf{if} $q' \not\equiv n \pmod{k^{m+1}}$ \textbf{and} $\mathsf{Eval}(k^{m+1} \cdot \ell + q') \neq a$ \textbf{then}
    \item \quad\quad\quad \textbf{if} can select $\tilde{q} \in Q_{m+1}$ with $d_{m+1}(\tilde{q}, q') \geq \ell$ and $\mathsf{Eval}(k^{m+1} \cdot \ell + \tilde{q}) = a$ \textbf{then}
    \item \quad\quad\quad\quad $\delta_m (q, i)\leftarrow\tilde{q}$
    \item \quad\quad\quad \textbf{else}
    \item \quad\quad\quad\quad $q_{\mathrm{new}} \leftarrow n \bmod k^{m+1}$
    \item \quad\quad\quad\quad add $q_{\mathrm{new}}$ to $Q_{m+1}$
    \item \quad\quad\quad\quad $\delta_m (q, i) \leftarrow q_{\mathrm{new}}$
    \item \quad\quad\quad\quad \textbf{if} $m + 1 < m^*$ \textbf{then}
    \item \quad\quad\quad\quad\quad \textbf{for} $j = 0$ \textbf{to} $k - 1$ \textbf{do} $\delta_{m+1}(q_{\mathrm{new}}, j)\leftarrow\delta_{m+1}(q', j)$
    \item \quad\quad\quad\quad \textbf{else} $\tau(q_{\mathrm{new}}) \leftarrow a$
  \end{enumerate*}
\end{algorithm}

\FloatBarrier

\subsection{Proof of Theorem~\ref{thm:rtl-bounds}}
\label{sec:tp-proof}

We now prove that the TP algorithm is statistically and computationally efficient with respect to RTL $k$-automaticity. Throughout this section, let $x \in \Sigma^*$ with $\mathrm{AC}^{\mathrm{R}}_k (x) = s$, and fix a witnessing DFAO $M^* = (Q^*, q^*_0, \delta^*, \tau^*)$ with $|Q^*| = s$ states. Let $\sigma \mathrel{:=} |\Sigma|$. We write $n$ for the current time index and $m^*$ for the current depth parameter of the algorithm (which satisfies $k^{m^*-1} < n \leq k^{m^*}$, i.e., $m^* = \left\lceil \log_k n \right\rceil$).

The proof unfolds in four stages. Section~\ref{sec:model-correct} establishes three invariants of the algorithm's persistent state --- \textbf{Data Consistency}, \textbf{Minimality}, and \textbf{Distinguishability} (Lemma~\ref{lem:model-correct}) --- which together fix the semantics of the model: the stored transitions are consistent with all past observations, each equivalence class has a unique canonical representative, and two states are kept distinct only when an observation separates them. As a direct consequence of Distinguishability, Section~\ref{sec:state-bound} derives the universal bound $\left|Q_m\right| \leq s$ at every depth (Lemma~\ref{lem:state-bound}). The mistake bound \eqref{eq:rtl-mistakes} is then proved in Section~\ref{sec:tp-mistakes} by attributing each TP mistake to a single depth level and coupling the per-level dynamics to an instance of the branch prediction game of Section~\ref{sec:branch-pred}; Proposition~\ref{prop:branch-pred} supplies the per-level mistake count, with Lemma~\ref{lem:state-bound} controlling its parameters. Finally, Section~\ref{sec:tp-compression} establishes the compression bound \eqref{eq:rtl-compression}.

\subsubsection{Model Correctness}
\label{sec:model-correct}

We prove the three invariants that give the TP state its intended meaning.
At depth $m$, a state is a residue in $[k^m]$.  If $q\in Q_m$ and
$r$ is a nonnegative integer, we call the index $k^m r+q$ the
\emph{row-$r$ occurrence} of $q$.  Thus two states at the same depth are
separated by the data once some common row has been observed at which the two
occurrences carry different symbols.  We call level $m^*$ the \emph{output
level}, because the output map $\tau$ is defined on $Q_{m^*}$; levels
$0,1,\ldots,m^*-1$ are the non-output levels.

We use two auxiliary maps.  The first one evaluates the current model after
entering it from an arbitrary state.  For $m\leq m^*$ and $q\in Q_m$, define
\begin{align}
\operatorname{PEval}_{m^*}(q,0)
  &:= \tau(q),
  \nonumber\\
\operatorname{PEval}_{m}(q,r)
  &:=
  \operatorname{PEval}_{m+1}
    \left(
      \delta_m(q,r\bmod k),
      \left\lfloor r/k\right\rfloor
    \right)
  \qquad (m<m^*).
\label{eq:peval-def}
\end{align}
In particular, $\mathsf{Eval}(\ell)=\operatorname{PEval}_0(0,\ell)$.
The second auxiliary map routes a residue to the state currently representing
it:
\begin{align}
\rho_0(r)&:=0,
\nonumber\\
\rho_{m+1}(r)&:=
  \delta_m
  \left(
    \rho_m(r\bmod k^m),
    \left\lfloor r/k^m\right\rfloor
  \right).
\label{eq:route-def}
\end{align}
Unfolding the definitions gives, whenever the terms are defined,
\begin{equation}
\mathsf{Eval}(k^m r+q)
  =
\operatorname{PEval}_m(\rho_m(q),r).
\label{eq:eval-peval-route}
\end{equation}
In particular, $\mathsf{Eval}(\ell)=\tau(\rho_{m^*}(\ell))$ for
$\ell<k^{m^*}$.

During an execution of $\mathsf{Update}$ we also consider intermediate table
configurations, after some assignments have been made but before the update
call has finished.  We call any such intermediate configuration a
\emph{transient state}.  If a statement is about a transient state, then
$\Delta_m$ denotes the divergence table that Algorithm~\ref{alg:tp-diverge}
would compute from that transient state.  By contrast, the symbols $d_m$ used
inside Algorithm~\ref{alg:tp-update} refer to the single divergence table
computed on line~2 of that algorithm, before the loop starts.  When a case
uses superscripts ``old'' and ``new'', they mean the table configurations
immediately before and immediately after the assignment being analysed.
We use only the following two facts about divergence.  For any fixed transient
state and $q_1,q_2\in Q_m$:
\begin{enumerate}
  \item if $\Delta_m(q_1,q_2)\geq L$, then
  \[
    \operatorname{PEval}_m(q_1,r)
    =
    \operatorname{PEval}_m(q_2,r)
    \qquad\text{for every } r<L;
  \]
  \item if
  $\operatorname{PEval}_m(q_1,R)\neq
    \operatorname{PEval}_m(q_2,R)$,
  then $\Delta_m(q_1,q_2)\leq R$.
\end{enumerate}

Let's recall Lemma~\ref{lem:model-correct-prelim}, restating it more precisely.

\begin{lemma}
\label{lem:model-correct}
  \emph{(Model correctness.)}
  Consider the state of the algorithm after either
  $\mathsf{Initialize}(x[0])$ (where we set $n=0$), or after the complete
  execution of $\mathsf{Update}(n,x[n])$.  Then the following statements hold.
  \begin{enumerate}
    \item \emph{Data consistency.}
    \[
      \mathsf{Eval}(t)=x[t]
      \qquad\text{for every }t\leq n.
    \]

    \item \emph{Minimality.}  For every $m\leq m^*$ and $q\in Q_m$,
    \[
      q=
      \min\{r\in[k^m] : \rho_m(r)=q\}.
    \]

    \item \emph{Distinguishability.}
    Fix $m\leq m^*$ and distinct states $q_1,q_2\in Q_m$.  Let
    \[
      r:=d_m(q_1,q_2)
    \]
    be their divergence value in the final state.  Then
    \[
      k^m r+q_1\leq n,
      \qquad
      k^m r+q_2\leq n,
    \]
    and the two observed symbols at these row-$r$ occurrences are different:
    \[
      x[k^m r+q_1]\neq x[k^m r+q_2].
    \]
  \end{enumerate}
\end{lemma}

\begin{proof}
The proof is an induction over the operations of the algorithm.  Instead of
proving Minimality directly after every assignment, we maintain two elementary
properties of individual transition-table entries.

For $h<m^*$, the first property is the \emph{stored-state routing property}:
\begin{equation}
\begin{gathered}
  u\in Q_{h+1},\quad
  p:=u\bmod k^h,\quad
  i:=\left\lfloor u/k^h\right\rfloor \\
  \Longrightarrow\quad
  p\in Q_h \text{ and } \delta_h(p,i)=u.
\end{gathered}
\label{eq:stored-state-routing}
\end{equation}
In words, a state label $u$ at level $h+1$ is reached by following the
level-$h$ state given by the low $h$ digits of $u$ and then reading the next
digit of $u$.

The second property is the \emph{representative upper-bound property}:
\begin{equation}
  \delta_h(p,i)
  \leq
  i k^h+p
  \qquad
  \text{for every }p\in Q_h\text{ and }i\in[k].
\label{eq:representative-upper-bound}
\end{equation}
Here $ik^h+p$ is the level-$(h+1)$ residue obtained by appending digit $i$ to
$p$.  The property says that the transition target is a representative whose
numeric label is no larger than the residue it represents.

The stored-state routing property and the representative upper-bound property
imply Minimality.  First,
\eqref{eq:stored-state-routing} gives, by induction on $m$, that
\begin{equation}
  \rho_m(q)=q
  \qquad\text{for every }q\in Q_m.
\label{eq:stored-routes-to-itself}
\end{equation}
Indeed, the claim is trivial at level $0$, and the induction step is exactly
\eqref{eq:stored-state-routing}.  Conversely, suppose $\rho_{m+1}(r)=q$.
Write
\[
  i:=\left\lfloor r/k^m\right\rfloor,
  \qquad
  u:=r\bmod k^m,
  \qquad
  p:=\rho_m(u).
\]
By the induction hypothesis at depth $m$, $p$ is the least residue routed to
$p$, and hence $p\leq u$.  Using
\eqref{eq:representative-upper-bound},
\[
  q=\delta_m(p,i)
  \leq i k^m+p
  \leq i k^m+u
  =r.
\]
Together with \eqref{eq:stored-routes-to-itself}, this says that every stored
state $q\in Q_m$ is exactly the least residue routed to $q$, which is
Minimality.

We will therefore maintain Data Consistency, the stored-state routing
property, the representative upper-bound property, and Distinguishability.
One minor complication occurs inside an Update call: when a new state is created
at a non-output level $m+1<m^*$, it copies the outgoing transitions of an old
state.  Immediately after the copy, the model may still evaluate the new state
and the old state identically even though the new observation has already
separated them in the data.  We record such pairs temporarily.

More precisely, during a single execution of $\mathsf{Update}(n,a)$, with
$a=x[n]$, a \emph{pending pair at level $h$} is a triple
\[
  (h,\{q_h^+,q_h^-\},L_h),
  \qquad
  q_h^+,q_h^-\in Q_h,
\]
created as follows.  If iteration $h-1$ of the loop creates the interior state
$q_h^+=n\bmod k^h$ by copying the old target $q_h^-$, then
\[
  L_h:=\left\lfloor n/k^h\right\rfloor.
\]
There is at most one pending pair at each level, because a single Update call
processes each level once.  A pending pair carries a \emph{stored separating
row} $L_h$.

Between loop iterations we use the following relaxed version of
Distinguishability.  For a distinct pair $q_1,q_2\in Q_h$, define its current
witness row $R(q_1,q_2)$ as follows:
\[
R(q_1,q_2):=
\begin{cases}
  L_h,
    &\text{if }\{q_1,q_2\}=\{q_h^+,q_h^-\}
      \text{ is pending at level }h,\\
  \Delta_h(q_1,q_2),
    &\text{otherwise.}
\end{cases}
\]
The relaxed invariant is
\begin{equation}
\begin{gathered}
  k^h R(q_1,q_2)+q_1\leq n,
  \qquad
  k^h R(q_1,q_2)+q_2\leq n,
  \\
  x[k^h R(q_1,q_2)+q_1]\neq
  x[k^h R(q_1,q_2)+q_2].
\end{gathered}
\label{eq:relaxed-distinguishability}
\end{equation}
After an Update call finishes, we will convert this relaxed statement back
into the Distinguishability statement of Lemma~\ref{lem:model-correct}.

The base case is immediate.  After $\mathsf{Initialize}(x[0])$ we have
$m^*=0$, $Q_0=\{0\}$, and $\tau(0)=x[0]$.  Thus
$\mathsf{Eval}(0)=x[0]$; the two transition-table properties have no
nontrivial instance; and there are no distinct same-level states to distinguish.

Fix now an update round $n>0$, set $a:=x[n]$, and assume the lemma after the
previous update.

\medskip
\noindent\textbf{After Preprocess.}
If $\mathsf{Preprocess}$ does not increase $m^*$, the previous final
invariants give Data Consistency for $t<n$, the stored-state routing property,
the representative upper-bound property, and the Distinguishability
statement in the lemma.  There are no pending pairs yet.

Assume instead that the old depth is $\bar m$ and that
$n=k^{\bar m}$, so the new depth is $\bar m+1$.  The new set
$Q_{\bar m+1}$ is a copy of $Q_{\bar m}$, and the new transitions are
\[
  \delta_{\bar m}(q,0)=q,
  \qquad
  \delta_{\bar m}(q,i)=0\quad (i\geq 1).
\]
For every $t<n$, the new highest processed digit is $0$, so the extra
transition leaves the old level-$\bar m$ state unchanged.  Hence
$\mathsf{Eval}(t)=x[t]$ still holds for all $t<n$.

Both properties are immediate at the new level.  A
copied level-$\bar m$ state
$q<k^{\bar m}$ is reached through the digit-$0$ transition
$\delta_{\bar m}(q,0)=q$, and the digit-$i$ transitions for $i\geq1$ point to
$0\leq i k^{\bar m}+q$.  Lower levels are unchanged.

For Distinguishability at old levels, take distinct old states
$q_1,q_2\in Q_h$ with $h\leq\bar m$, and let $R$ be their old divergence
row.  The old invariant gives
$k^hR+q_j<n=k^{\bar m}$ for $j\in\{1,2\}$, hence
$R<k^{\bar m-h}$.  On rows $0,1,\ldots,R$, the newly added highest digit is
always $0$.  Therefore a partial evaluation from level $h$ follows the same
old transitions up to level $\bar m$ and then uses
$\delta_{\bar m}(p,0)=p$, so
\[
  \operatorname{PEval}_h^{\mathrm{new}}(q_j,R)
  =
  \operatorname{PEval}_h^{\mathrm{old}}(q_j,R)
  \qquad(j=1,2).
\]
The old divergence row $R$ means that the two old partial evaluations at row
$R$ are unequal.  Hence the two new partial evaluations at row $R$ are unequal,
and the new divergence row
$r^{\mathrm{new}}:=\Delta_h^{\mathrm{new}}(q_1,q_2)$ satisfies
$r^{\mathrm{new}}\leq R$.  Thus
$k^h r^{\mathrm{new}}+q_j\leq k^hR+q_j<n$ for $j=1,2$.  Since the stored
states route to themselves, \eqref{eq:eval-peval-route} gives
\[
  \mathsf{Eval}^{\mathrm{new}}(k^h r^{\mathrm{new}}+q_j)
  =
  \operatorname{PEval}_h^{\mathrm{new}}(q_j,r^{\mathrm{new}})
  \qquad(j=1,2).
\]
By the definition of $r^{\mathrm{new}}$, the two right-hand sides are different.
The indices are below $n$, so the already-proved data consistency after
preprocessing gives
\[
  x[k^h r^{\mathrm{new}}+q_j]
  =
  \mathsf{Eval}^{\mathrm{new}}(k^h r^{\mathrm{new}}+q_j)
  \qquad(j=1,2),
\]
and therefore
\[
  x[k^h r^{\mathrm{new}}+q_1]
  \neq
  x[k^h r^{\mathrm{new}}+q_2].
\]
This is the required separating observation at the new divergence row.
At the new level, take distinct states
$q_1,q_2\in Q_{\bar m+1}=Q_{\bar m}^{\mathrm{old}}$.  They were already
distinct states at the old level $\bar m$.  In the old state, these states were
at the output level, so Algorithm~\ref{alg:tp-diverge} assigned them divergence
row $0$.  Applying the previous Distinguishability statement with $h=\bar m$
and $r=0$ gives
\[
  q_1<n,
  \qquad
  q_2<n,
  \qquad
  x[q_1]\neq x[q_2].
\]
Thus the relaxed invariant holds after preprocessing, with no pending pairs.

\medskip
\noindent\textbf{Loop invariant.}
Before each loop iteration we maintain:
\begin{enumerate}
  \item $\mathsf{Eval}(t)=x[t]$ for every $t<n$;
  \item the stored-state routing property \eqref{eq:stored-state-routing}
        and the representative upper-bound property
        \eqref{eq:representative-upper-bound};
  \item relaxed distinguishability, namely
        \eqref{eq:relaxed-distinguishability}.
\end{enumerate}
If an iteration makes no assignment, these statements are unchanged.  We
therefore consider an iteration $m$ in which line~4 selects a state, line~6 is
true, and the algorithm performs a repair.  Put
\begin{align*}
  q&:=n\bmod k^m,
  &
  i&:=\left\lfloor n/k^m\right\rfloor\bmod k,
  \\
  \ell&:=\left\lfloor n/k^{m+1}\right\rfloor,
  &
  q'&:=\delta_m(q,i),
  \\
  r_+&:=i k^m+q=n\bmod k^{m+1}.
\end{align*}
Line~6 says that $q'\neq r_+$ and that
\begin{equation}
  \mathsf{Eval}(k^{m+1}\ell+q')\neq a.
\label{eq:line6-disagreement}
\end{equation}
The representative upper-bound property \eqref{eq:representative-upper-bound}
gives $q'\leq r_+$, hence in fact $q'<r_+$.
Moreover, $r_+\notin Q_{m+1}$ before this iteration.  Otherwise
\eqref{eq:stored-state-routing}, applied to the state $r_+$, would force
$\delta_m(q,i)=r_+$, contradicting $q'\neq r_+$.  Repeatedly applying
\eqref{eq:stored-state-routing} also shows that no old state at any
higher level has low $(m+1)$-residue equal to $r_+$. 

\smallskip
\noindent\emph{Reassignment.}
Assume line~7 succeeds and chooses $\tilde q\in Q_{m+1}$, so line~8 changes
only the entry
\[
  \delta_m(q,i): q'\longmapsto \tilde q.
\]
By the preceding paragraph, the transition being modified is not the transition required by
\eqref{eq:stored-state-routing} for any state. Hence
\eqref{eq:stored-state-routing} remains valid.

We next check the representative upper-bound property for the modified entry.
By
\eqref{eq:line6-disagreement} and the value test in line~7,
\[
  \mathsf{Eval}(k^{m+1}\ell+q')\neq a,
  \qquad
  \mathsf{Eval}(k^{m+1}\ell+\tilde q)=a,
\]
so $\tilde q\neq q'$.  Earlier iterations have affected only levels at most
$m$; the tables from level $m+1$ downward are therefore exactly the same as
when line~2 computed the divergence table.  Thus line~7 gives
$\Delta_{m+1}(\tilde q,q')\geq\ell$ for the pre-reassignment state.
Earlier loop iterations can create pending pairs only at levels at most $m$,
so no pending pair has yet been created at level $m+1$.  Therefore relaxed
distinguishability for the distinct pair $\tilde q,q'$ yields
\[
  k^{m+1}\Delta_{m+1}(\tilde q,q')+\tilde q
  \leq
  n
  =k^{m+1}\ell+r_+.
\]
As $\Delta_{m+1}(\tilde q,q')\geq\ell$, this implies
$\tilde q\leq r_+=i k^m+q$.  This is exactly the representative
upper-bound property for the new transition.

Now consider data consistency.  Let $t<n$.  If the evaluation path for $t$
does not use the modified entry, its value is unchanged.  Otherwise the path
reaches $q$ at level $m$ and reads digit $i$ next.  Write
$s:=\lfloor t/k^{m+1}\rfloor$.  We claim that $s<\ell$.  If $s=\ell$, then
$t<n$ and the common digit $i$ would force
$t\bmod k^m<q$.  But Minimality, already implied by the two properties
proved above, gives $q\leq t\bmod k^m$, a contradiction.  Hence $s<\ell$, and the divergence
test from line~7 gives
\[
  \operatorname{PEval}_{m+1}(\tilde q,s)
  =
  \operatorname{PEval}_{m+1}(q',s).
\]
Thus the new and old evaluations of $t$ agree.  At the new index $n$, the path
reaches $q$ at level $m$, takes the new transition to $\tilde q$, and returns
\[
  \operatorname{PEval}_{m+1}(\tilde q,\ell)
  =
  \mathsf{Eval}(k^{m+1}\ell+\tilde q)
  =a,
\]
where the first equality uses $\rho_{m+1}(\tilde q)=\tilde q$.  Therefore data
consistency holds through $n$ after a reassignment.

It remains to spell out relaxed distinguishability.  Pending pairs keep their
stored separating rows.  Pairs at levels $h\geq m+1$ are unaffected, because
their partial evaluations start strictly deeper than the modified level-$m$
transition.

Now take a non-pending pair $u_1,u_2\in Q_h$ with $h\leq m$.  Let $R$ be its
witness row before the reassignment, so
\[
  k^hR+u_j\leq n\quad(j=1,2),
  \qquad
  x[k^hR+u_1]\neq x[k^hR+u_2].
\]
The stored-state routing property after the reassignment gives, for $j=1,2$,
\[
  \rho_h^{\mathrm{new}}(u_j)=u_j.
\]
Therefore \eqref{eq:eval-peval-route}, applied in the new state, gives
\[
  \operatorname{PEval}^{\mathrm{new}}_h(u_j,R)
  =
  \mathsf{Eval}^{\mathrm{new}}(k^hR+u_j)
  \qquad(j=1,2).
\]
The two indices are observed, and the data consistency proved above holds for
all indices up to $n$.  Hence
\[
  \operatorname{PEval}^{\mathrm{new}}_h(u_j,R)
  =
  x[k^hR+u_j]
  \qquad(j=1,2).
\]
The two $x$-values are different, so the two new partial evaluations disagree
at row $R$.  By the definition of divergence,
$r^{\mathrm{new}}:=\Delta_h^{\mathrm{new}}(u_1,u_2)\leq R$.  Consequently
$k^h r^{\mathrm{new}}+u_j\leq k^hR+u_j\leq n$ for $j=1,2$.  At the row
$r^{\mathrm{new}}$, the partial evaluations disagree by definition of
$r^{\mathrm{new}}$.  Applying \eqref{eq:eval-peval-route} and data consistency at
these two observed indices gives
\[
\begin{aligned}
  x[k^h r^{\mathrm{new}}+u_1]
  &=
  \mathsf{Eval}^{\mathrm{new}}(k^h r^{\mathrm{new}}+u_1) \\
  &=
  \operatorname{PEval}^{\mathrm{new}}_h(u_1,r^{\mathrm{new}}) \\
  &\neq
  \operatorname{PEval}^{\mathrm{new}}_h(u_2,r^{\mathrm{new}}) \\
  &=
  \mathsf{Eval}^{\mathrm{new}}(k^h r^{\mathrm{new}}+u_2) \\
  &=
  x[k^h r^{\mathrm{new}}+u_2].
\end{aligned}
\]
Thus the relaxed invariant is preserved.

\smallskip
\noindent\emph{Interior-level new state ($m+1<m^*$).}
Assume line~7 fails and $m+1<m^*$.  The algorithm creates
\[
  q_{\mathrm{new}}:=r_+=n\bmod k^{m+1},
\]
adds it to $Q_{m+1}$, sets $\delta_m(q,i)=q_{\mathrm{new}}$, and copies the outgoing
transitions of the old target $q'$:
\begin{equation}
  \delta_{m+1}^{\mathrm{new}}(q_{\mathrm{new}},j)
  =
  \delta_{m+1}^{\mathrm{old}}(q',j)
  \qquad(j\in[k]).
\label{eq:copied-transitions}
\end{equation}
The state $q_{\mathrm{new}}$ is genuinely new by the preliminary observation above.
The same observation ensures that adding this incoming transition does not
violate \eqref{eq:stored-state-routing} for any old state, while
\eqref{eq:stored-state-routing} holds for $q_{\mathrm{new}}$ by construction.
The representative upper-bound property holds for $\delta_m(q,i)$ with equality.
For the copied entries, the old representative upper-bound property and
$q'\leq q_{\mathrm{new}}$ give
\[
  \delta_{m+1}^{\mathrm{new}}(q_{\mathrm{new}},j)
  =
  \delta_{m+1}^{\mathrm{old}}(q',j)
  \leq
  j k^{m+1}+q'
  \leq
  j k^{m+1}+q_{\mathrm{new}}.
\]
All other transition entries are unchanged.

The copy operation gives the identity
\begin{equation}
  \operatorname{PEval}_{m+1}^{\mathrm{new}}(q_{\mathrm{new}},s)
  =
  \operatorname{PEval}_{m+1}^{\mathrm{old}}(q',s)
  \qquad\text{for every }s.
\label{eq:copied-peval}
\end{equation}
Therefore every old index $t<n$ keeps its value: if its path does not use the
modified level-$m$ entry, nothing changes; if it does, the old target $q'$ is
replaced by $q_{\mathrm{new}}$ with the same continuation.  Hence data consistency for
$t<n$ is preserved.

We now verify relaxed distinguishability in detail.  First consider a pair of
states $u_1,u_2\in Q_h$ not involving $q_{\mathrm{new}}$.  If this pair was
already pending before the present iteration, then its stored row and the two
observed indices in \eqref{eq:relaxed-distinguishability} are unchanged.  The
relaxed invariant for this pair is therefore unchanged as well, because it is
a statement only about the data symbols at those two indices.

For pairs not involving $q_{\mathrm{new}}$, it remains to consider a pair
that was not pending before the present iteration.  If $h\geq m+2$, its
partial evaluations start deeper than the levels that changed.  If $h=m+1$, the partial evaluations are still unchanged, since $q_{\mathrm{new}}$ is not involved.  If $h\leq m$, an evaluation that
uses the modified entry now continues from $q_{\mathrm{new}}$ where it formerly
continued from $q'$.  The equality \eqref{eq:copied-peval} says that these two
continuations have identical partial evaluations.  Therefore, for every row
$s$,
\[
  \operatorname{PEval}^{\mathrm{new}}_h(u_j,s)
  =
  \operatorname{PEval}^{\mathrm{old}}_h(u_j,s)
  \qquad(j=1,2).
\]
The two partial-evaluation functions for the pair are unchanged.  Hence their
first disagreement row is unchanged:
\[
  \Delta_h^{\mathrm{new}}(u_1,u_2)
  =
  \Delta_h^{\mathrm{old}}(u_1,u_2).
\]
The previous relaxed invariant was, for this non-pending pair, exactly the
statement that the row
$R:=\Delta_h^{\mathrm{old}}(u_1,u_2)$ is observed on both sides and satisfies
\[
  x[k^hR+u_1]
  \neq
  x[k^hR+u_2].
\]
Since the new divergence row is the same $R$, the same two observed data
symbols separate $u_1$ and $u_2$ after the present iteration.  Thus all old
pairs remain covered by the relaxed invariant.

For the pair $q_{\mathrm{new}},q'$, we create a pending pair at level $m+1$ with
stored row $\ell$.  This stored row is separating in the data: indeed,
\[
  k^{m+1}\ell+q_{\mathrm{new}}=n,
  \qquad
  k^{m+1}\ell+q'<n,
\]
where the strict inequality uses $q'<q_{\mathrm{new}}$.  By data consistency before
this iteration and by \eqref{eq:line6-disagreement},
\[
  x[k^{m+1}\ell+q']
  =
  \mathsf{Eval}^{\mathrm{old}}(k^{m+1}\ell+q')
  \neq
  a
  =x[n].
\]
Thus the pending pair satisfies \eqref{eq:relaxed-distinguishability}.

It remains to handle a pair $q_{\mathrm{new}},q_2$, where
$q_2\in Q_{m+1}^{\mathrm{old}}$ and $q_2\neq q'$.  There was no pending pair at
level $m+1$ before this iteration, since earlier iterations can only create
pending pairs at levels at most $m$.  By \eqref{eq:copied-peval}, the new
divergence between $q_{\mathrm{new}}$ and $q_2$ is the old divergence between
$q'$ and $q_2$; call this row $R$.  The old relaxed invariant, applied to the
non-pending pair $q',q_2$, gives
\begin{equation}
\begin{gathered}
  k^{m+1}R+q'\leq n,
  \qquad
  k^{m+1}R+q_2\leq n,
  \\
  x[k^{m+1}R+q']
  \neq
  x[k^{m+1}R+q_2].
\end{gathered}
\label{eq:old-qprime-q2}
\end{equation}
Since $n=k^{m+1}\ell+q_{\mathrm{new}}$ and $q_{\mathrm{new}}<k^{m+1}$, the first inequality
in \eqref{eq:old-qprime-q2} implies $R\leq\ell$.  If $R<\ell$, then both
$k^{m+1}R+q_{\mathrm{new}}$ and $k^{m+1}R+q'$ are below $n$.  Identity~(\ref{eq:copied-peval}) gives
\[
  \operatorname{PEval}_{m+1}^{\mathrm{new}}(q_{\mathrm{new}},R)
  =
  \operatorname{PEval}_{m+1}^{\mathrm{old}}(q',R).
\]
Because $q_{\mathrm{new}}$ and $q'$ route to themselves,
\eqref{eq:eval-peval-route} rewrites this as an equality of model evaluations
at the two observed indices:
\[
  \mathsf{Eval}^{\mathrm{new}}(k^{m+1}R+q_{\mathrm{new}})
  =
  \mathsf{Eval}^{\mathrm{old}}(k^{m+1}R+q').
\]
Both indices are below $n$, so data consistency before and after the iteration
converts this model equality into
\[
  x[k^{m+1}R+q_{\mathrm{new}}]
  =
  x[k^{m+1}R+q'].
\]
The old relaxed invariant also gives
$x[k^{m+1}R+q']\neq x[k^{m+1}R+q_2]$.  Combining the last two displayed
relations yields
\[
  x[k^{m+1}R+q_{\mathrm{new}}]
  \neq
  x[k^{m+1}R+q_2],
\]
so $q_{\mathrm{new}}$ and $q_2$ are separated at row $R$.

If $R=\ell$, then the candidate $q_2$ met the divergence requirement in
line~7.  As in the reassignment case, the level-$(m+1)$ divergence table used
by line~7 is still current, because earlier loop iterations did not change
levels $m+1$ and below.  Since line~7 failed, $q_2$ must have failed the value
requirement:
\[
  \mathsf{Eval}^{\mathrm{old}}(k^{m+1}\ell+q_2)
  \neq a.
\]
The inequality $k^{m+1}\ell+q_2\leq n$ from
\eqref{eq:old-qprime-q2}, together with $q_2\neq q_{\mathrm{new}}$, gives
$k^{m+1}\ell+q_2<n$.  Since this index is old, data consistency before the
present iteration gives
\[
  x[k^{m+1}\ell+q_2]
  =
  \mathsf{Eval}^{\mathrm{old}}(k^{m+1}\ell+q_2).
\]
Combining with the failed value requirement above and with $a=x[n]$ gives
\[
  x[k^{m+1}\ell+q_2]
  \neq
  a
  =
  x[k^{m+1}\ell+q_{\mathrm{new}}].
\]
This is the separating observation for $q_{\mathrm{new}}$ and $q_2$ at row
$R=\ell$.  Thus relaxed distinguishability is preserved in the interior
new-state case.

\smallskip
\noindent\emph{Output-level new state ($m+1=m^*$).}
Assume line~7 fails and $m+1=m^*$.  After preprocessing, $n<k^{m^*}$, so the
new state is
\[
  q_{\mathrm{new}}:=n\bmod k^{m^*}=n.
\]
The algorithm adds $q_{\mathrm{new}}$ to $Q_{m^*}$, sets
$\delta_{m^*-1}(q,i)=q_{\mathrm{new}}$, and assigns $\tau(q_{\mathrm{new}})=a$.
Both properties (\ref{eq:stored-state-routing}) and (\ref{eq:representative-upper-bound}) are preserved exactly as above.  In particular, the
representative upper-bound property holds with equality for the modified
transition.

No old index $t<n$ can use the modified transition.  If it did, then at level
$m^*-1$ it would be routed through $q$ and would read digit $i$.  Minimality
would give $q\leq t\bmod k^{m^*-1}$, hence
\[
  t\geq i k^{m^*-1}+q=n,
\]
a contradiction.  Thus old data consistency is preserved, and the new index
is routed to $q_{\mathrm{new}}$ and evaluates to $\tau(q_{\mathrm{new}})=a$.

For relaxed distinguishability, first take a pair of old states
$u_1,u_2\in Q_h$.  Let $R$ be its previous relaxed witness row.  The previous
relaxed invariant gives
\[
  k^hR+u_j\leq n\quad(j=1,2),
  \qquad
  x[k^hR+u_1]\neq x[k^hR+u_2].
\]
The output-level assignment has just restored data consistency through $n$.
The stored-state routing property gives
$\rho_h^{\mathrm{new}}(u_j)=u_j$ for $j=1,2$.  Applying
\eqref{eq:eval-peval-route} gives
\[
  \operatorname{PEval}^{\mathrm{new}}_h(u_j,R)
  =
  \mathsf{Eval}^{\mathrm{new}}(k^hR+u_j)
  =
  x[k^hR+u_j]
  \qquad(j=1,2).
\]
The two $x$-values are unequal, hence the new partial evaluations already
disagree at row $R$.  If
$r^{\mathrm{new}}:=\Delta_h^{\mathrm{new}}(u_1,u_2)$, then
$r^{\mathrm{new}}\leq R$, so both row-$r^{\mathrm{new}}$ occurrences are observed.
At this actual divergence row, \eqref{eq:eval-peval-route} and data
consistency give the explicit separation
\[
\begin{aligned}
  x[k^h r^{\mathrm{new}}+u_1]
  &=
  \mathsf{Eval}^{\mathrm{new}}(k^h r^{\mathrm{new}}+u_1) \\
  &=
  \operatorname{PEval}^{\mathrm{new}}_h(u_1,r^{\mathrm{new}}) \\
  &\neq
  \operatorname{PEval}^{\mathrm{new}}_h(u_2,r^{\mathrm{new}}) \\
  &=
  \mathsf{Eval}^{\mathrm{new}}(k^h r^{\mathrm{new}}+u_2) \\
  &=
  x[k^h r^{\mathrm{new}}+u_2].
\end{aligned}
\]

For pairs involving $q_{\mathrm{new}}$, the output-level divergence row is $0$.  If
the other state is $q'$, then $q'<q_{\mathrm{new}}=n$ and
\[
  x[q']
  =
  \mathsf{Eval}^{\mathrm{old}}(q')
  \neq
  a
  =x[q_{\mathrm{new}}],
\]
where the middle inequality is line~6 with $\ell=0$.  If the other state is an
old output-level state $q_2\neq q'$, then $q_2$ satisfied the divergence requirement in
line~7, again because $\ell=0$.  Since line~7 failed,
$\mathsf{Eval}^{\mathrm{old}}(q_2)\neq a$.  Also $q_2<n$: old output-level states were
introduced before time $n$, and preprocessing only copies existing state
labels.  Data consistency before the present iteration gives
\[
  x[q_2]
  =
  \mathsf{Eval}^{\mathrm{old}}(q_2)
  \neq
  a
  =
  x[q_{\mathrm{new}}].
\]
This is the row-$0$ separating observation for $q_{\mathrm{new}}$ and $q_2$.
This proves the relaxed invariant after an output-level new-state step.

\medskip
\noindent\textbf{End of the Update loop.}
The loop invariant gives data consistency for every $t<n$, the stored-state
routing property, the representative upper-bound property, and relaxed
distinguishability.  We first prove
data consistency at the new index $n$.

If the last modifying iteration was a reassignment or an output-level
new-state step, then the argument above already proved
$\mathsf{Eval}(n)=a$.  If no iteration modified the state, we argue as
follows.  Before processing $n$, the current output level $Q_{m^*}$ does not
contain the state label $n$: every state label was added at some earlier
time $n'<n$ as a residue at most $n'$, and preprocessing only copies labels.
Thus $n\notin Q_{m^*}$.  Let $h$ be the least integer in
$\{1,\ldots,m^*\}$ such that
\[
  \rho_h(n\bmod k^h)\neq n\bmod k^h.
\]
Such an $h$ exists because $n<k^{m^*}$ after preprocessing and
$\rho_{m^*}(n)\in Q_{m^*}$.  Set
\begin{align*}
  m&:=h-1,
  &
  q&:=n\bmod k^m,
  \\
  i&:=\left\lfloor n/k^m\right\rfloor\bmod k,
  &
  \ell&:=\left\lfloor n/k^h\right\rfloor,
  \\
  q'&:=\delta_m(q,i).
\end{align*}
By minimality of $h$, the line~4 guard passes at level $m$, and unfolding
\eqref{eq:route-def} gives
\[
  q'=\rho_h(n\bmod k^h)\neq n\bmod k^h.
\]
Thus the residue part of line~6 is true.  Since this iteration did not modify
the state, the value part of line~6 must be false:
\begin{equation}
  \mathsf{Eval}(k^h\ell+q')=a.
\label{eq:no-modification-value}
\end{equation}
Let
\[
  t':=k^h\ell+q'.
\]
The representative upper-bound property \eqref{eq:representative-upper-bound}
gives $q'\leq n\bmod k^h$.  Since $q'\neq n\bmod k^h$, this
inequality is strict, so $t'<n$.  At depth $h$, both $n$ and $t'$ route to
$q'$.  This is true for $n$ by definition of $q'$, and for $t'$ by
\eqref{eq:stored-routes-to-itself}.  For every higher depth, the remaining
base-$k$ digits of $n$ and $t'$ are identical, because
$\lfloor n/k^h\rfloor=\ell=\lfloor t'/k^h\rfloor$.  Induction using
\eqref{eq:route-def} therefore gives
\[
  \rho_{m^*}(n)=\rho_{m^*}(t').
\]
Consequently,
\[
  \mathsf{Eval}(n)
  =
  \tau(\rho_{m^*}(n))
  =
  \tau(\rho_{m^*}(t'))
  =
  \mathsf{Eval}(t')
  =a,
\]
where the last equality is \eqref{eq:no-modification-value}.

It remains to rule out the possibility that the last modifying iteration was
an interior new-state step.  Suppose iteration $m<m^*-1$ created
$q_{\mathrm{new}}=n\bmod k^{m+1}$ by copying the old target $q'$.  In the next
iteration, $m+1$, the line~4 guard passes with $q_{\mathrm{new}}$.  Define
\begin{align*}
  \hat i
    &:=\left\lfloor n/k^{m+1}\right\rfloor\bmod k,
  &
  \hat\ell
    &:=\left\lfloor n/k^{m+2}\right\rfloor,
  \\
  \hat q'
    &:=\delta_{m+1}^{\mathrm{new}}(q_{\mathrm{new}},\hat i)
     =\delta_{m+1}^{\mathrm{old}}(q',\hat i).
\end{align*}
If
$\hat q'=n\bmod k^{m+2}=\hat i k^{m+1}+q_{\mathrm{new}}$, then the old
representative upper-bound property for the copied transition gives
$q_{\mathrm{new}}\leq q'$.  The old representative upper-bound property for
$\delta_m(q,i)=q'$ gives $q'\leq q_{\mathrm{new}}$.  Hence
$q'=q_{\mathrm{new}}$, contradicting
that $q_{\mathrm{new}}$ was not an old state.  The residue part of line~6 is
therefore true at iteration $m+1$.

For the value part, write
$\ell=k\hat\ell+\hat i$.  The copied transition and the fact that levels
$m+2,m+3,\ldots,m^*$ were unchanged imply
\[
  \mathsf{Eval}^{\mathrm{new}}
    (k^{m+2}\hat\ell+\hat q')
  =
  \mathsf{Eval}^{\mathrm{old}}
    (k^{m+1}\ell+q')
  \neq a,
\]
where the final inequality is the line~6 test from iteration $m$.  Thus
iteration $m+1$ must also modify the state, contradicting the assumption that
iteration $m$ was the last modifying iteration.

We have proved Data Consistency through $n$.  The stored-state routing property
and the representative upper-bound property were preserved throughout the loop,
so Minimality follows from the first part of the proof.

Finally we convert relaxed distinguishability to the Distinguishability
statement of Lemma~\ref{lem:model-correct}.
Fix a final level $h$ and distinct states $u_1,u_2\in Q_h$.  Let
$R=R(u_1,u_2)$ be the relaxed witness row.  By
\eqref{eq:relaxed-distinguishability},
\[
  k^hR+u_j\leq n\quad(j=1,2),
  \qquad
  x[k^hR+u_1]\neq x[k^hR+u_2].
\]
The final stored-state routing property gives $\rho_h(u_j)=u_j$.
Combining this with
\eqref{eq:eval-peval-route}, and then with final Data Consistency at the two
observed indices, yields
\[
  \operatorname{PEval}_h(u_j,R)
  =
  \mathsf{Eval}(k^hR+u_j)
  =
  x[k^hR+u_j]
  \qquad(j=1,2).
\]
The two $x$-values are different, so the two partial evaluations disagree at
row $R$.  Since $d_h(u_1,u_2)$ is the first row at which they disagree,
\[
  d_h(u_1,u_2)\leq R.
\]
Consequently
\[
  k^h d_h(u_1,u_2)+u_j
  \leq
  k^hR+u_j
  \leq n
  \qquad(j=1,2).
\]
At the actual divergence row, the partial evaluations are unequal:
\[
  \operatorname{PEval}_h(u_1,d_h(u_1,u_2))
  \neq
  \operatorname{PEval}_h(u_2,d_h(u_1,u_2)).
\]
Using $\rho_h(u_j)=u_j$, \eqref{eq:eval-peval-route}, and final Data
Consistency at the two observed row-$d_h(u_1,u_2)$ occurrences gives
\[
\begin{aligned}
  x[k^h d_h(u_1,u_2)+u_1]
  &=
  \mathsf{Eval}(k^h d_h(u_1,u_2)+u_1) \\
  &=
  \operatorname{PEval}_h(u_1,d_h(u_1,u_2)) \\
  &\neq
  \operatorname{PEval}_h(u_2,d_h(u_1,u_2)) \\
  &=
  \mathsf{Eval}(k^h d_h(u_1,u_2)+u_2) \\
  &=
  x[k^h d_h(u_1,u_2)+u_2].
\end{aligned}
\]
This is exactly Distinguishability.  The induction, and hence the lemma, is
complete.
\end{proof}

\subsubsection{State Bound}
\label{sec:state-bound}

We now show that, after each complete Update call, the sets $Q_m$ contain no more elements than the witnessing DFAO $M^*$ has states. The argument is a direct consequence of Distinguishability (Lemma~\ref{lem:model-correct}).

Fix a padding length $m^{**} \in \mathbb{N}$ such that $|x| \leq k^{m^{**}}$ and
\begin{equation}
  \forall t < |x| : \quad x[t] = M^*((\langle t \rangle^{m^{**}}_k)^R),
  \label{eq:dfa-sem}
\end{equation}
whose existence is guaranteed by $\mathrm{AC}^{\mathrm{R}}_k (x) = s$. Since the algorithm's current depth parameter $m^*$ is bounded above by $\left\lceil \log_k |x| \right\rceil \leq m^{**}$, every $m \leq m^*$ referenced below satisfies $m \leq m^{**}$.

Let $\hat{\delta}^* : Q^* \times [k]^* \to Q^*$ denote the standard extension of $\delta^*$ to finite input words:
\[
  \hat{\delta}^*(q, \epsilon) \mathrel{:=} q, \quad \hat{\delta}^*(q, u a) \mathrel{:=} \delta^*(\hat{\delta}^*(q, u), a).
\]
For each $m \in \{0, \dots, m^{**}\}$, define the \emph{level-$m$ reading map} $\psi^*_m : [k^m] \to Q^*$ by
\[
  \psi^*_m (r) \mathrel{:=} \hat{\delta}^*(q^*_0, \,d_0 d_1 \dots d_{m-1}), \quad \text{where } d_j \mathrel{:=} \left\lfloor r / k^j \right\rfloor \bmod k.
\]
Equivalently, $\psi^*_m (r)$ is the state of $M^*$ after reading the $m$ least-significant base-$k$ digits of $r$ in LSDF order. Iterating the transition yields the \emph{chain rule}
\begin{equation}
  \psi^*_{m'}(r) = \hat{\delta}^*(\psi^*_m (r \bmod k^m), \,d_m d_{m+1} \dots d_{m'-1})
  \label{eq:psi-chain}
\end{equation}
for every $0 \leq m \leq m' \leq m^{**}$ and $r \in [k^{m'}]$. In particular, \eqref{eq:dfa-sem} rewrites as
\begin{equation}
  x[t] = \tau^*(\psi^*_{m^{**}} (t)) \quad \text{for every } t < |x|.
  \label{eq:x-via-psi}
\end{equation}

\begin{lemma}
\label{lem:state-bound}
  \emph{(State bound.)} Consider the state of the algorithm after the complete execution of $\mathsf{Update}(n, x[n])$, or after Initialize (in which case $n = 0$). For every $m \leq m^*$, $|Q_m| \leq s$.
\end{lemma}

\begin{proof}
  We show that $\psi^*_m$ restricted to $Q_m$ is injective. Since $\psi^*_m$ takes values in $Q^*$, this yields $|Q_m| \leq |Q^*| = s$.

  Fix distinct $q_1, q_2 \in Q_m$ and set $r \mathrel{:=} d_m (q_1, q_2)$. By Distinguishability (Lemma~\ref{lem:model-correct}),
\begin{equation}
  k^m r + q_j \leq n \quad \text{for } j \in \{1, 2\}, \quad \text{and} \quad x[k^m r + q_1] \neq x[k^m r + q_2].
  \label{eq:sb-distinguish}
\end{equation}
  In particular, $r$ is a finite natural number; and since $x$ is defined at both indices $k^m r + q_j$, we have $k^m r + q_j < |x| \leq k^{m^{**}}$, whence $r < k^{m^{**} - m}$.

  Suppose, toward a contradiction, that $\psi^*_m (q_1) = \psi^*_m (q_2)$. Let $e_0 e_1 \dots e_{m^{**} - m - 1}$ denote the base-$k$ digits of $r$ in LSDF order, $e_j \mathrel{:=} \left\lfloor r / k^j \right\rfloor \bmod k$. Applying \eqref{eq:psi-chain} with $m' = m^{**}$ to $k^m r + q_j \in [k^{m^{**}}]$ --- and observing that $(k^m r + q_j) \bmod k^m = q_j$ while the digits of $k^m r + q_j$ at positions $m, m+1, \dots, m^{**} - 1$ coincide with $e_0, e_1, \dots, e_{m^{**} - m - 1}$, independent of $j$ --- we obtain
\[
  \psi^*_{m^{**}}(k^m r + q_j) = \hat{\delta}^*(\psi^*_m (q_j), \,e_0 e_1 \dots e_{m^{**} - m - 1}) \quad \text{for } j \in \{1, 2\}.
\]
  By the hypothesis $\psi^*_m (q_1) = \psi^*_m (q_2)$, the right-hand sides agree, hence
\[
  \psi^*_{m^{**}}(k^m r + q_1) = \psi^*_{m^{**}}(k^m r + q_2).
\]
  Invoking \eqref{eq:x-via-psi} at both indices (each strictly less than $|x|$),
\[
  x[k^m r + q_1] = \tau^*(\psi^*_{m^{**}}(k^m r + q_1)) = \tau^*(\psi^*_{m^{**}}(k^m r + q_2)) = x[k^m r + q_2],
\]
  contradicting \eqref{eq:sb-distinguish}. Hence $\psi^*_m (q_1) \neq \psi^*_m (q_2)$, completing the proof of injectivity.
\end{proof}

\subsubsection{Statistical Efficiency}
\label{sec:tp-mistakes}

In this subsection we establish the mistake bound \eqref{eq:rtl-mistakes} of Theorem~\ref{thm:rtl-bounds}. The strategy is a \emph{per-level decomposition}: apart from the initial round and the rounds whose index is a power of $k$, each prediction round is charged to a single \emph{breakpoint depth} --- the depth at which the exact-residue walk of Algorithm~\ref{alg:tp-predict} halts --- and the mistakes at each level $m$ are bounded separately. The exceptional rounds are counted directly. The per-level analysis couples TP's dynamics to an instance of the branch-prediction game of Section~\ref{sec:branch-pred} played on a tree $T^{(m)}$ whose vertices encode a refinement hierarchy of subsequence classes. The link between TP's algorithmic state and these partitions is the content of a \emph{semantic interpretation} lemma (Lemma~\ref{lem:tp-semantics}). Two ingredients then close the per-level count: Proposition~\ref{prop:branch-pred} bounds the mistakes whose game-vertex is \emph{internal}, while the residual mistakes --- those at game-\emph{leaves} of $T^{(m)}$ --- are absorbed by a count of \emph{branching} vertices in $T^{(m)}$, in turn bounded by $s - 1$ via the universal class bound below. Summing over the $O(\log_k |x|)$ levels and adding the direct contribution of the exceptional rounds yields the claimed bound.

Throughout the subsection we retain the setup of Section~\ref{sec:state-bound}: a witnessing DFAO $M^* = (Q^*, q^*_0, \delta^*, \tau^*)$ with $|Q^*| = s$, a padding length $m^{**}$ with $|x| \leq k^{m^{**}}$, and the reading map $\psi^*_m : [k^m] \to Q^*$. Let
\[
  M \mathrel{:=} \left\lceil \log_k |x| \right\rceil .
\]
Then every depth level reached during the run is at most $M$, and $M \leq m^{**}$.

\textbf{Subsequence equivalence.} The intrinsic object underlying the analysis is the following refinement hierarchy. At depth $\ell$, we consider exactly those residues whose arithmetic-progression subsequence has entries defined through row $\ell-1$.

\begin{definition}
\label{def:subseq-equiv}
  For $m \geq 1$ and $\ell \in \mathbb{N}$, define the level-$\ell$ live residue set
\[
  \Lambda^m_\ell \mathrel{:=}
  \begin{cases}
    [k^m] & \text{if } \ell = 0,\\
    \{r \in [k^m] : k^m(\ell-1)+r < |x|\} & \text{if } \ell \geq 1.
  \end{cases}
\]
  For $r,r' \in \Lambda^m_\ell$, write $r \approx^m_\ell r'$ when
\[
  \forall j < \ell : \quad x[k^m j + r] = x[k^m j + r'].
\]
  For each fixed $m,\ell$, this is an equivalence relation on $\Lambda^m_\ell$; denote the class of $r \in \Lambda^m_\ell$ by $[r]^m_\ell$.
\end{definition}

Since $\Lambda^m_{\ell+1} \subseteq \Lambda^m_\ell$, the relation $\approx^m_{\ell+1}$ refines $\approx^m_\ell$ after restricting the latter to $\Lambda^m_{\ell+1}$. The DFAO $M^*$ provides a universal upper bound on the number of classes at every level.

\begin{lemma}
\label{lem:psi-refines}
  For every $1 \leq m \leq m^{**}$, every $\ell \in \mathbb{N}$, and all $r,r' \in \Lambda^m_\ell$: if $\psi^*_m(r)=\psi^*_m(r')$, then $r \approx^m_\ell r'$. Consequently, for every $\ell$, there are at most $s$ different $\approx^m_\ell$-classes in $\Lambda^m_\ell$.
\end{lemma}

\begin{proof}
  Fix $r,r' \in \Lambda^m_\ell$ with $\psi^*_m(r)=\psi^*_m(r')$. For $j<\ell$, the definition of $\Lambda^m_\ell$ ensures that both $k^m j+r$ and $k^m j+r'$ are strictly smaller than $|x|$. The chain rule \eqref{eq:psi-chain} applied at $m'=m^{**}$, with the last $m^{**}-m$ digits of $k^m j+r$ being the base-$k$ digits of $j$ padded with zeros (call them $e_0,e_1,\dots,e_{m^{**}-m-1}$), yields
\[
  \psi^*_{m^{**}}(k^m j+r)
  =
  \hat{\delta}^*(\psi^*_m(r), e_0\dots e_{m^{**}-m-1}).
\]
  The same identity holds with $r'$ in place of $r$; since $\psi^*_m(r)=\psi^*_m(r')$, the right-hand sides agree, and \eqref{eq:x-via-psi} yields $x[k^m j+r]=x[k^m j+r']$. As $j<\ell$ was arbitrary, $r \approx^m_\ell r'$.

  For the class bound, choose one representative from each $\approx^m_\ell$-class. If two representatives had the same $\psi^*_m$-value, the first part would put them in the same class. Thus the representatives inject into $Q^*$, and there are at most $|Q^*|=s$ classes.
\end{proof}

\textbf{Breakpoints and level-$m$ rounds.} Let
\[
  \mathcal{B} \mathrel{:=}
  \{t\in[|x|] : t=0 \text{ or } t=k^a \text{ for some } a\in\mathbb{N}\}
\]
be the set of boundary rounds, and put \(\mathcal{R}\mathrel{:=}[|x|]\setminus\mathcal{B}\). Since
\[
  |\mathcal{B}|\leq 1+\left\lceil\log_k |x|\right\rceil,
\]
these rounds contribute only \(O(\log_k |x|)\) possible mistakes. We therefore carry out the branch-prediction accounting on \(\mathcal{R}\), and add \(|\mathcal{B}|\) back in the final aggregation. For \(n\in\mathcal{R}\), consider the execution of \(\mathsf{Predict}(n)\).

\begin{definition}
\label{def:breakpoint}
  The \emph{breakpoint} $m(n) \in \mathbb{N}$ is the final value of the loop variable $m$ after the while-loop of Algorithm~\ref{alg:tp-predict} exits. The \emph{active pair} is $\iota(n) \mathrel{:=} (q_n,i_n)$ with
\[
  q_n \mathrel{:=} n \bmod k^{m(n)}, \quad i_n \mathrel{:=} \left\lfloor n / k^{m(n)} \right\rfloor \bmod k,
\]
  and the \emph{active residue} is $r_n \mathrel{:=} i_n k^{m(n)} + q_n = n \bmod k^{m(n)+1}$. Write $\ell_n \mathrel{:=} \left\lfloor n/k^{m(n)+1} \right\rfloor$, so that $n=k^{m(n)+1}\ell_n+r_n$.
\end{definition}

\begin{lemma}
\label{lem:breakpoint-range}
  For every $n \in \mathcal{R}$, if $d_n$ denotes the current value of the algorithm's depth parameter $m^*$ at the moment $\mathsf{Predict}(n)$ runs, then $m(n) \in \{0,1,\dots,d_n-1\} \subseteq \{0,1,\dots,M-1\}$.
\end{lemma}

\begin{proof}
  At the moment $\mathsf{Predict}(n)$ runs, every element of $Q_{d_n}$ was introduced either by $\mathsf{Initialize}$ (adding $0<n$) or by line 11 of Algorithm~\ref{alg:tp-update} during some earlier $\mathsf{Update}(n',\cdot)$ with $n'<n$ (adding $q_{\mathrm{new}} = n' \bmod k^{m'+1}$ at some level $m'+1 \leq d_n$, possibly propagated unchanged through subsequent $\mathsf{Preprocess}$ calls). In either case the added value is at most $n'<n$. The algorithm changes depth only at powers of $k$, so for $n\in\mathcal{R}$ the usual current-depth relation is strict: $n<k^{d_n}$. Hence every $q \in Q_{d_n}$ satisfies $q<n=n \bmod k^{d_n}$, so $n \bmod k^{d_n}\notin Q_{d_n}$ and the while-loop cannot reach $m=d_n$.
\end{proof}

For $m \in \{0,\dots,M-1\}$, define the set of \emph{level-$m$ rounds}
\begin{equation*}
  R_m \mathrel{:=} \{ n \in \mathcal{R} : m(n)=m \}.
\end{equation*}
The sets $\{R_m\}_{m<M}$ partition $\mathcal{R}$. Observe that TP's prediction at index $n$ only depends on the algorithm's state through the variables at breakpoint depth $m(n)$: the exact-residue walk in Algorithm~\ref{alg:tp-predict} exits at $m(n)$, the candidate set $D$ is drawn from $Q_{m(n)+1}$, and the prediction is $\mathsf{Eval}(k^{m(n)+1}\ell_n+q^*)$ for some $q^* \in D$. Our strategy is to attribute each TP mistake to a single \emph{level}, namely $m(n)$, and bound the mistakes at each level separately.

\textbf{Semantic interpretation of TP's state.} The link between TP's algorithmic state and the partitions $\approx^{m+1}_\ell$ is furnished by Lemma~\ref{lem:model-correct}.

\begin{lemma}
\label{lem:tp-semantics}
  Suppose $\mathsf{Predict}(n)$ has breakpoint $m$, with active pair $(q_n,i_n)$ and old target $q' \mathrel{:=} \delta_m(q_n,i_n)$. Then:
  \begin{enumerate}
    \item \emph{(Old-target equivalence.)} $r_n \approx^{m+1}_{\ell_n} q'$.
    \item \emph{(Candidate bound.)} Every $\tilde q \in D$ satisfies $\tilde q < r_n$.
    \item \emph{(Candidate set.)} The set $D \subseteq Q_{m+1}$ constructed in line 8 of Algorithm~\ref{alg:tp-predict} satisfies
\[
  D = \{\tilde q \in Q_{m+1} \cap \Lambda^{m+1}_{\ell_n} : \tilde q \approx^{m+1}_{\ell_n} r_n\}.
\]
    \item \emph{(Distinct children.)} Distinct $\tilde q_1,\tilde q_2 \in D$ lie in distinct $\approx^{m+1}_{\ell_n+1}$-classes.
    \item \emph{(Descent classification.)} For every $(\hat q,\hat i) \in Q_m \times [k]$ with $\hat r \mathrel{:=} \hat i k^m+\hat q < r_n$ and $\tilde q \mathrel{:=} \delta_m(\hat q,\hat i) \in D$: $\hat r \approx^{m+1}_{\ell_n+1} \tilde q$ and $\tilde q \leq \hat r$.
  \end{enumerate}
\end{lemma}

\begin{proof}
  Put $h\mathrel{:=}m+1$ for readability.

  \emph{(1).} Since $q_n \in Q_m$ at the time of $\mathsf{Predict}(n)$, Minimality (Lemma~\ref{lem:model-correct}) gives $\rho_m(q_n)=q_n$; unfolding \eqref{eq:route-def} at depth $h$,
\[
  \rho_h(r_n)=\delta_m(\rho_m(q_n),i_n)=\delta_m(q_n,i_n)=q'.
\]
  Minimality applied to the transition $\delta_m(q_n,i_n)=q'$ gives $q'\leq r_n$. Hence, if $\ell_n>0$, both $r_n$ and $q'$ lie in $\Lambda^h_{\ell_n}$; for $\ell_n=0$ this is automatic. For every $j<\ell_n$,
\[
  \mathsf{Eval}(k^h j+r_n)=\operatorname{PEval}_h(q',j),
\]
  and, applying Minimality to $q'\in Q_h$, $\mathsf{Eval}(k^h j+q')=\operatorname{PEval}_h(q',j)$ as well. Hence
\begin{equation}
  \mathsf{Eval}(k^h j+r_n)=\mathsf{Eval}(k^h j+q') \quad \text{for every } j<\ell_n.
  \label{eq:sem-rn-qp}
\end{equation}
  Both indices in \eqref{eq:sem-rn-qp} are strictly smaller than $n$; Data Consistency converts it into $x[k^h j+r_n]=x[k^h j+q']$. As $j<\ell_n$ was arbitrary, $r_n \approx^h_{\ell_n} q'$.

  \emph{(2).} First, $r_n\notin Q_h$. Indeed, the proof of (1) established $\rho_h(r_n)=q'$. If $r_n\in Q_h$, Minimality would also give $\rho_h(r_n)=r_n$, whence $q'=r_n$, contradicting the fact that the while-loop of Algorithm~\ref{alg:tp-predict} exits at $m=m(n)$ only when $q'\not\equiv n \pmod{k^h}$.

  Let $\tilde q\in D$. If $\tilde q=q'$, then Minimality applied to $\delta_m(q_n,i_n)=q'$ gives $\tilde q=q'\leq r_n$, and $\tilde q\neq r_n$ by the preceding paragraph; hence $\tilde q<r_n$.

  If $\tilde q\neq q'$, set $r^*\mathrel{:=}d_h(\tilde q,q')$. Since $\tilde q\in D$, the definition of $D$ gives $r^*\geq \ell_n$. Distinguishability (Lemma~\ref{lem:model-correct}) applied to the distinct pair $\tilde q,q'\in Q_h$ yields $k^h r^*+\tilde q\leq n$. Writing $n=k^h\ell_n+r_n$, this rearranges to
\[
  k^h(r^*-\ell_n)+\tilde q\leq r_n.
\]
  Since $r^*\geq\ell_n$, the left-hand side is at least $\tilde q$, so $\tilde q\leq r_n$. Again $\tilde q\neq r_n$ because $\tilde q\in Q_h$ but $r_n\notin Q_h$, and therefore $\tilde q<r_n$.

  \emph{(3).} We first prove $D$ is contained in the displayed set. Let $\tilde q\in D$. By (2), $\tilde q<r_n$, so $\tilde q\in\Lambda^h_{\ell_n}$ and in fact $\tilde q\in\Lambda^h_{\ell_n+1}$. The condition $d_h(q',\tilde q)\geq\ell_n$ defining $D$ is equivalent to
\[
  \mathsf{Eval}(k^h j+q')=\mathsf{Eval}(k^h j+\tilde q) \quad\text{for all } j<\ell_n.
\]
  For all such $j$, both indices are smaller than $n$, so Data Consistency converts this to $q'\approx^h_{\ell_n}\tilde q$. Combining with (1) and transitivity of $\approx^h_{\ell_n}$ gives $\tilde q\approx^h_{\ell_n}r_n$.

  Conversely, suppose $\tilde q\in Q_h\cap\Lambda^h_{\ell_n}$ and $\tilde q\approx^h_{\ell_n}r_n$. By (1), $q'\approx^h_{\ell_n}r_n$, hence $q'\approx^h_{\ell_n}\tilde q$. For every $j<\ell_n$, the indices $k^h j+q'$ and $k^h j+\tilde q$ are smaller than $n$ (for $\tilde q$, use $j<\ell_n$ and $\tilde q\in\Lambda^h_{\ell_n}$). Data Consistency therefore gives
\[
  \mathsf{Eval}(k^h j+q')=\mathsf{Eval}(k^h j+\tilde q) \quad\text{for all } j<\ell_n,
\]
  i.e. $d_h(q',\tilde q)\geq\ell_n$. Thus $\tilde q\in D$.

    \emph{(4).} Let $\tilde q_1,\tilde q_2\in D$ be distinct and set
  $r^*\mathrel{:=}d_h(\tilde q_1,\tilde q_2)$. By (2), both
  $\tilde q_1$ and $\tilde q_2$ are strictly smaller than $r_n$, hence they lie in
  $\Lambda^h_{\ell_n+1}$. Distinguishability applied to the distinct states
  $\tilde q_1,\tilde q_2\in Q_h$ gives
\[
  k^h r^*+\tilde q_i\leq n \qquad (i=1,2)
\]
  and
\[
  x[k^h r^*+\tilde q_1]\neq x[k^h r^*+\tilde q_2].
\]
  Since $n=k^h\ell_n+r_n$ and $\tilde q_i<r_n<k^h$, the first displayed
  inequality forces $r^*\leq\ell_n$: otherwise
  $k^h r^*+\tilde q_i\geq k^h(\ell_n+1)>n$. Thus $r^*<\ell_n+1$, and the
  second displayed inequality gives a row below $\ell_n+1$ on which the two
  live residues disagree. Therefore
  $\tilde q_1\not\approx^h_{\ell_n+1}\tilde q_2$.

  \emph{(5).} Minimality applied to $\hat q\in Q_m$ gives $\rho_m(\hat q)=\hat q$, and unfolding \eqref{eq:route-def} yields $\rho_h(\hat r)=\tilde q$. Minimality applied to $\tilde q\in Q_h$ then gives
\[
  \tilde q = \min\{r\in[k^h] : \rho_h(r)=\tilde q\} \leq \hat r.
\]
  In particular, $\tilde q\leq\hat r<r_n$, so both $\hat r$ and $\tilde q$ lie in $\Lambda^h_{\ell_n+1}$. Using $\rho_h(\hat r)=\tilde q$ and $\rho_h(\tilde q)=\tilde q$, we get, for every $j\leq\ell_n$,
\begin{equation}
  \mathsf{Eval}(k^h j+\hat r)=\operatorname{PEval}_h(\tilde q,j)=\mathsf{Eval}(k^h j+\tilde q).
  \label{eq:sem-hr-tq}
\end{equation}
  For these $j$, both indices $k^h j+\hat r$ and $k^h j+\tilde q$ lie strictly below $n$ (using $\hat r<r_n$ and $\tilde q\leq\hat r<r_n$), so Data Consistency converts \eqref{eq:sem-hr-tq} into equality of $x$-values. Hence $\hat r\approx^h_{\ell_n+1}\tilde q$.
\end{proof}

\textbf{The level-$m$ tree.} Fix $m\in\{0,\dots,M-1\}$ and put $h\mathrel{:=}m+1$. We define a rooted tree $T^{(m)}$ on which the branch-prediction reduction will run.

\begin{definition}
\label{def:level-m-tree}
  The \emph{level-$m$ tree} $T^{(m)}$ has:
  \begin{itemize}
    \item \emph{Vertices:} pairs $(C,\ell)$ where $0\leq\ell\leq \left\lceil |x|/k^h \right\rceil$ and $C$ is a $\approx^h_\ell$-class in $\Lambda^h_\ell$.
    \item \emph{Root:} $(\Lambda^h_0,0)=([k^h],0)$.
    \item \emph{Edges:} $(C,\ell)\to(C',\ell+1)$ whenever $C'$ is a $\approx^h_{\ell+1}$-class with $C'\subseteq C$.
  \end{itemize}
\end{definition}

The children of $(C,\ell)$ therefore partition $C\cap\Lambda^h_{\ell+1}$. If this intersection is empty, $(C,\ell)$ has no children. Otherwise it has one child for each $\approx^h_{\ell+1}$-class contained in $C$. Each residue $r\in[k^h]$ induces a unique path in $T^{(m)}$ through the levels $0,1,\dots,L_r$, where
\[
  L_r \mathrel{:=} \left|\{j\geq0 : k^h j+r<|x|\}\right|,
\]
with vertex at depth $\ell\leq L_r$ equal to $([r]^h_\ell,\ell)$.

\textbf{Reduction to branch prediction.} We bound TP's mistakes at level $m$ by coupling the algorithm's behaviour to an instance of the branch prediction game (Definition~\ref{def:bp-game}) played on $T^{(m)}$. Set
\[
  \hat Q_m \mathrel{:=} \bigcup_{t\geq0} Q_m^{(t)} \subseteq [k^m]
\]
to be the set of all residues that ever populate $Q_m$ during the run of the algorithm on $x$ (equivalently, the final value of $Q_m$, since states are never removed from $Q_m$), and define the level-$m$ element set as the collection of $(m+1)$-residues lying over $\hat Q_m$ modulo $k^m$:
\[
  E_m \mathrel{:=} \{r\in[k^h] : r\bmod k^m\in\hat Q_m\}
  =
  \{\hat i k^m+\hat q : \hat q\in\hat Q_m,\ \hat i\in[k]\}.
\]

Each $\hat r\in E_m$ participates in the level-$m$ game as a single element, traversing its canonical path in $T^{(m)}$. The element enters the game at the moment its level-$m$ residue $\hat r\bmod k^m$ first appears in $Q_m$, placed at the depth it would occupy had it been participating from the start, and thereafter it descends one step at every real time $n$ with $n\bmod k^h=\hat r$. At descent steps corresponding to rounds in $R_m$ from a game-internal vertex (case (b.2) of Definition~\ref{def:bp-game}), TP's plurality vote on line 11 of Algorithm~\ref{alg:tp-predict} will coincide with the $\mathsf{PB}$ strategy of Definition~\ref{def:mb}. Descent steps from a game-leaf (case (b.1)) produce no game prediction, but TP may still err at the corresponding round; such residual mistakes are absorbed by a count of branching vertices in $T^{(m)}$. The lemmas and game-dynamics description that follow make this coupling precise.

\begin{lemma}
\label{lem:Em-bound}
\[
  |E_m|\leq k|\hat Q_m|\leq ks.
\]
\end{lemma}

\begin{proof}
  Each $\hat r\in E_m$ is uniquely of the form $\hat r=\hat i k^m+\hat q$ with $\hat q\in\hat Q_m\subseteq[k^m]$ and $\hat i\in[k]$. Hence $|E_m|\leq k|\hat Q_m|$, and $|\hat Q_m|\leq s$ by Lemma~\ref{lem:state-bound}.
\end{proof}

For each $\hat r\in E_m$, write $\hat q\mathrel{:=}\hat r\bmod k^m$ and let $t(\hat r)$ denote the time at which $\hat q$ is first added to $Q_m$ (by either $\mathsf{Initialize}$, a $\mathsf{Preprocess}$ promotion, or line 11 of $\mathsf{Update}$). Nothing is ever removed from $Q_m$, so $\hat q\in Q_m^{(t)}$ for all $t\geq t(\hat r)$.

\textbf{Game dynamics.} Each round in the online sequence prediction setting corresponds to a single activation round followed by any number of introduction rounds of the level-$m$ game.

\emph{Introduction round.} For each $\hat r\in E_m$, the element $\hat r$ is introduced at time $t(\hat r)$ and placed at vertex $v_{\mathrm{intro}}(\hat r)\mathrel{:=}([\hat r]^h_{\ell_0(\hat r)},\ell_0(\hat r))$, where
\begin{align*}
\ell_0(\hat r) \mathrel{:=}
\begin{cases}
  \left\lfloor t(\hat r)/k^h \right\rfloor+1 & \text{if } t(\hat r)\bmod k^h \geq \hat r,\\
  \left\lfloor t(\hat r)/k^h \right\rfloor & \text{otherwise.}
\end{cases}
\end{align*}
The level $\ell_0(\hat r)$ is the depth at which an element with residue $\hat r$ would naturally sit at time $t(\hat r)$ had it been participating in the game from the start. The definition of $\Lambda^h_{\ell_0(\hat r)}$ ensures that this vertex exists. (Multiple introductions may share a single real time $t$; we order them arbitrarily.)

\emph{Activation round.} At every real time $n>0$ with both
\[
  n\bmod k^h\in E_m \quad\text{and}\quad n>t(n\bmod k^h),
\]
an activation round of element $r_n^{(m)}\mathrel{:=}n\bmod k^h$ is played. Writing $\ell_n^{(m)}\mathrel{:=}\lfloor n/k^h\rfloor$, the activation moves $r_n^{(m)}$ from its current vertex $([r_n^{(m)}]^h_{\ell_n^{(m)}},\ell_n^{(m)})$ to the child $([r_n^{(m)}]^h_{\ell_n^{(m)}+1},\ell_n^{(m)}+1)$, which exists because $k^h\ell_n^{(m)}+r_n^{(m)}=n<|x|$.

The following lemma characterizes the tree vertex where each element is located on each round of the game.

\begin{lemma}
\label{lem:depth-formula}
  \emph{(Element vertex.)} Fix $m\in\{0,\dots,M-1\}$, $\hat r\in E_m$, and $n\in[|x|]$ with $t(\hat r)<n$. Just before any round processed at time $n$, the element $\hat r$ is at vertex $([\hat r]^h_{D_n(\hat r)},D_n(\hat r))$, where
\[
  D_n(\hat r) \mathrel{:=} \left|\{j\geq0 : k^h j+\hat r<n\}\right|.
\]
  Equivalently, writing $n=k^h\ell+r$ with $\ell\mathrel{:=}\lfloor n/k^h\rfloor$ and $r\mathrel{:=}n\bmod k^h$,
\begin{align*}
D_n(\hat r)=
\begin{cases}
  \ell+1 & \text{if } \hat r<r,\\
  \ell & \text{if } \hat r\geq r.
\end{cases}
\end{align*}
\end{lemma}

\begin{proof}
  The case-based form is immediate from the inequality $k^h j+\hat r<n=k^h\ell+r$: if $\hat r<r$, the valid rows are $j=0,\dots,\ell$, while if $\hat r\geq r$, the valid rows are $j=0,\dots,\ell-1$.

  The element $\hat r$ always stays on its canonical path: introduction places it at $([\hat r]^h_{\ell_0(\hat r)},\ell_0(\hat r))$, and every activation of $\hat r$ moves it from $([\hat r]^h_\ell,\ell)$ to $([\hat r]^h_{\ell+1},\ell+1)$. Hence its vertex is determined by its depth, and it remains to check the depth.

  We proceed by forward induction on $n$, for $n\geq t(\hat r)+1$. In the base case, just before time $t(\hat r)+1$, the element has just been introduced and has not yet been activated. Writing $t(\hat r)=k^h\ell_t+r_t$, the same case analysis as above gives
\[
  D_{t(\hat r)+1}(\hat r)
  =
  \left|\{j\geq0:k^h j+\hat r\leq t(\hat r)\}\right|
  =
  \ell_0(\hat r).
\]
  For the inductive step, between time $n$ and time $n+1$ the depth of $\hat r$ changes only if $n\bmod k^h=\hat r$, in which case the activation of $\hat r$ increases it by one. This is exactly the change in $D_n(\hat r)$, since
\[
  D_{n+1}(\hat r)-D_n(\hat r)
  =
  \left|\{j\geq0:k^h j+\hat r=n\}\right|.
\]
\end{proof}

We fix one harmless convention for the level-$m$ game. Whenever the
$\mathsf{PB}$ rule has several children with maximal current population, its
tie-breaking is chosen to agree with the argmax convention used in line 11 of
Algorithm~\ref{alg:tp-predict}, under the child--candidate correspondence
proved in Lemma~\ref{lem:b2-children-id} below. Proposition~\ref{prop:branch-pred}
is independent of this convention; it is used only to identify TP's chosen
candidate with the child predicted by $\mathsf{PB}$. The next lemma supplies the
needed identity.

\textbf{Population identity.} For a vertex $u$ of $T^{(m)}$, write $N_n(u)$ for the number of game elements located at $u$ or any of its descendants just before any round processed at time $n$.

\begin{lemma}
\label{lem:population}
  \emph{(Population identity.)} Fix $n\in R_m$. Let $D$ be the candidate set constructed by $\mathsf{Predict}(n)$. For every $\tilde q\in D$,
\[
  \nu(\tilde q)=N_n(c^{\tilde q}),
\]
  where $c^{\tilde q}\mathrel{:=}([\tilde q]^{m+1}_{\ell_n+1},\ell_n+1)$ and $\nu$ is the count from line 10 of Algorithm~\ref{alg:tp-predict}.
\end{lemma}

\begin{proof}
  Put $h\mathrel{:=}m+1$, $\ell\mathrel{:=}\ell_n$, and
  $r\mathrel{:=}r_n$. All algorithmic quantities in this proof are taken just
  before $\textsf{Predict}(n)$ runs, equivalently after
  $\textsf{Update}(n-1,\cdot)$.

  Fix $\tilde q\in D$. By Lemma~\ref{lem:tp-semantics} part~(2),
  $\tilde q<r$. Hence
\[
  k^h\ell+\tilde q < k^h\ell+r=n<|x|,
\]
  so $\tilde q\in\Lambda^h_{\ell+1}$ and
  $c^{\tilde q}=([\tilde q]^h_{\ell+1},\ell+1)$ is a vertex of $T^{(m)}$.
  Similarly, every residue $\hat r<r$ lies in $\Lambda^h_{\ell+1}$. Thus the
  following set is well-defined:
\[
  \mathcal S_{\tilde q}
  \mathrel{:=}
  \{\hat r\in E_m :
    \hat r\bmod k^m\in Q_m^{(n-1)},\
    \hat r<r,\
    \hat r\approx^h_{\ell+1}\tilde q
  \}.
\]
  We show that both $\nu(\tilde q)$ and $N_n(c^{\tilde q})$ are equal to
  $|\mathcal S_{\tilde q}|$.

  First consider $\nu(\tilde q)$. The map
  $(\hat q,\hat i)\mapsto \hat i k^m+\hat q$ identifies the pairs counted in
  line 10 of Algorithm~\ref{alg:tp-predict} with residues
  $\hat r\in E_m$ satisfying
  $\hat r\bmod k^m\in Q_m^{(n-1)}$. For such a residue, writing
  $\hat q=\hat r\bmod k^m$ and $\hat i=\lfloor \hat r/k^m\rfloor$,
  Minimality gives $\rho_m^{(n-1)}(\hat q)=\hat q$, and therefore
\[
  \delta_m^{(n-1)}(\hat q,\hat i)
  =
  \rho_h^{(n-1)}(\hat r).
\]
  Hence
\[
  \nu(\tilde q)
  =
  \left|
  \{\hat r\in E_m :
    \hat r\bmod k^m\in Q_m^{(n-1)},\
    \hat r<r,\
    \rho_h^{(n-1)}(\hat r)=\tilde q
  \}
  \right|.
\]
  It remains to compare the last condition with
  $\hat r\approx^h_{\ell+1}\tilde q$.

  Let $\hat r\in E_m$ satisfy
  $\hat r\bmod k^m\in Q_m^{(n-1)}$ and $\hat r<r$. If
  $\rho_h^{(n-1)}(\hat r)=\tilde q$, then, with
  $\hat q=\hat r\bmod k^m$ and $\hat i=\lfloor\hat r/k^m\rfloor$, we have
  $\delta_m^{(n-1)}(\hat q,\hat i)=\tilde q$. Since $\tilde q\in D$,
  Lemma~\ref{lem:tp-semantics} part~(5) gives
  $\hat r\approx^h_{\ell+1}\tilde q$.

  Conversely, assume $\hat r\approx^h_{\ell+1}\tilde q$, and set
  $q_0\mathrel{:=}\rho_h^{(n-1)}(\hat r)$. Minimality gives
  $q_0\leq\hat r<r$ and $\rho_h^{(n-1)}(q_0)=q_0$. For every $j<\ell$,
\[
  \textsf{Eval}^{(n-1)}(k^h j+\hat r)
  =
  \operatorname{PEval}^{(n-1)}_h(q_0,j)
  =
  \textsf{Eval}^{(n-1)}(k^h j+q_0).
\]
  Both indices are below $n$, so Data Consistency converts this equality into
  $q_0\approx^h_\ell \hat r$. Since
  $\hat r\approx^h_{\ell+1}\tilde q$, we also have
  $\hat r\approx^h_\ell\tilde q$; and since $\tilde q\in D$,
  Lemma~\ref{lem:tp-semantics} part~(3) gives
  $\tilde q\approx^h_\ell r$. Therefore $q_0\approx^h_\ell r$, so another
  application of part~(3) gives $q_0\in D$.

  Now write again $\hat q=\hat r\bmod k^m$ and
  $\hat i=\lfloor\hat r/k^m\rfloor$. We have
  $\delta_m^{(n-1)}(\hat q,\hat i)=q_0$, with $q_0\in D$ and $\hat r<r$.
  Lemma~\ref{lem:tp-semantics} part~(5) gives
  $\hat r\approx^h_{\ell+1}q_0$. Together with
  $\hat r\approx^h_{\ell+1}\tilde q$, this puts $q_0$ and $\tilde q$ in the
  same $\approx^h_{\ell+1}$-class. By part~(4), distinct elements of $D$ lie
  in distinct such classes, hence $q_0=\tilde q$. This proves
\[
  \rho_h^{(n-1)}(\hat r)=\tilde q
  \quad\Longleftrightarrow\quad
  \hat r\approx^h_{\ell+1}\tilde q
\]
  for all residues under consideration, and therefore
  $\nu(\tilde q)=|\mathcal S_{\tilde q}|$.

  It remains to identify $N_n(c^{\tilde q})$. An element $\hat r\in E_m$ has
  been introduced before time $n$ exactly when
  $\hat r\bmod k^m\in Q_m^{(n-1)}$. For such an element,
  Lemma~\ref{lem:depth-formula} says that just before any round processed at
  time $n$, its current vertex is
  $([\hat r]^h_{D_n(\hat r)},D_n(\hat r))$, where
\[
  D_n(\hat r)=
  \begin{cases}
    \ell+1 & \text{if } \hat r<r,\\
    \ell & \text{if } \hat r\geq r.
  \end{cases}
\]
  Since $c^{\tilde q}$ has depth $\ell+1$, a current vertex of depth $\ell$
  cannot lie in its subtree. If $\hat r<r$, then the current vertex has the
  same depth as $c^{\tilde q}$, so it lies in the subtree of $c^{\tilde q}$
  exactly when it is equal to $c^{\tilde q}$, i.e. exactly when
  $\hat r\approx^h_{\ell+1}\tilde q$. Therefore
\[
  N_n(c^{\tilde q})
  =
  \left|
  \{\hat r\in E_m :
    \hat r\bmod k^m\in Q_m^{(n-1)},\
    \hat r<r,\
    \hat r\approx^h_{\ell+1}\tilde q
  \}
  \right|
  =
  |\mathcal S_{\tilde q}|.
\]
  Combining the two identities gives
  $\nu(\tilde q)=N_n(c^{\tilde q})$.
\end{proof}

\textbf{Mistake accounting at a single level.} We now relate TP's prediction at round $n\in R_m$ to the level-$m$ branch-prediction game and classify the resulting mistakes. Throughout, $q^*$ and $D$ refer to the values computed by $\mathsf{Predict}(n)$, and $u_n\mathrel{:=}([r_n]^{m+1}_{\ell_n},\ell_n)$ denotes the vertex of $T^{(m)}$ at which $r_n^{(m)}$ sits just before the activation at time $n$ (cf. Lemma~\ref{lem:depth-formula}).

\begin{lemma}
\label{lem:tp-mistake-criterion}
  \emph{(TP mistake criterion.)} For every $n\in R_m$, TP's prediction at round $n$ equals $x[k^{m+1}\ell_n+q^*]$. Consequently, TP makes a mistake at round $n$ if and only if $q^*\not\approx^{m+1}_{\ell_n+1}r_n$.
\end{lemma}

\begin{proof}
  By Lemma~\ref{lem:tp-semantics} part (2), $q^*<r_n$, hence $k^{m+1}\ell_n+q^*<k^{m+1}\ell_n+r_n=n$. Data Consistency (Lemma~\ref{lem:model-correct}), applied at the index $k^{m+1}\ell_n+q^*\leq n-1$ after $\mathsf{Update}(n-1,\cdot)$, gives $\mathsf{Eval}(k^{m+1}\ell_n+q^*)=x[k^{m+1}\ell_n+q^*]$, which is TP's prediction by the last line of Algorithm~\ref{alg:tp-predict}.

  Thus TP makes a mistake iff $x[k^{m+1}\ell_n+q^*]\neq x[n]=x[k^{m+1}\ell_n+r_n]$. By Lemma~\ref{lem:tp-semantics} part (3), $q^*\approx^{m+1}_{\ell_n}r_n$. Since both $q^*$ and $r_n$ are live at level $\ell_n+1$, the condition $q^*\approx^{m+1}_{\ell_n+1}r_n$ is equivalent to the preceding level-$\ell_n$ equivalence plus the additional equality $x[k^{m+1}\ell_n+q^*]=x[k^{m+1}\ell_n+r_n]$. This additional equality is precisely the negation of the mistake criterion.
\end{proof}

We now partition the mistakes at level $m$ into two classes, according to whether $u_n$ is internal or a leaf in the game's tree just before the time-$n$ activation.

\emph{Internal-vertex mistakes at level $m$.} If $u_n$ is a game-internal vertex at the time of the time-$n$ activation, the activation is a round of type (b.2) in Definition~\ref{def:bp-game}: the learner issues a prediction.

\begin{lemma}
\label{lem:b2-children-id}
  \emph{(b.2 children identification.)} Fix $n\in R_m$. The map $\tilde q\mapsto c^{\tilde q}$ from Lemma~\ref{lem:population} restricts to a bijection between
\[
  \{\tilde q\in D : \nu(\tilde q)\geq1\}
\]
  and the set of children of $u_n$ in the game tree just before the time-$n$ activation.
\end{lemma}

\begin{proof}
  \emph{Map is well-defined.} If $\tilde q\in D$ with $\nu(\tilde q)\geq1$, then $N_n(c^{\tilde q})\geq1$ by Lemma~\ref{lem:population}, so the subtree rooted at $c^{\tilde q}$ contains at least one game element just before any round at time $n$. By the dynamics of Definition~\ref{def:bp-game}, an element only enters a vertex (or its subtree) once that vertex has been attached to the game tree, and game-tree edges are never removed. Hence $c^{\tilde q}$ is a child of $u_n$ in the game tree.

  \emph{Injectivity.} Immediate from part (4) of Lemma~\ref{lem:tp-semantics}.

  \emph{Surjectivity.} Let $c$ be a child of $u_n$ in the game tree. Since the game tree only grows, $c$ must have been attached at some time $t_0<n$. By the dynamics of Definition~\ref{def:bp-game}, this attachment occurred during either (i) the introduction of some element $\hat r'\in E_m$ along a chain containing $u_n$ and $c$ (so $\hat r'$ was placed at $v_{\mathrm{intro}}(\hat r')$, which is $c$ or a descendant of $c$), or (ii) the activation of some element $\hat r'\in E_m$ then located at $u_n$, in which $c$ was created as a new child and $\hat r'$ moved to $c$. In either case, $\hat r'$ entered the subtree of $c$ at time $t_0$. Since elements only descend, the subtree of $c$ retains that element at all subsequent times, hence $N_n(c)\geq1$. Pick any $\hat r\in E_m$ in the subtree of $c$ at time $n$. By Lemma~\ref{lem:depth-formula}, $D_n(\hat r)\in\{\ell_n,\ell_n+1\}$; since $\hat r$ lies in the subtree of $c$ (which sits at depth $\ell_n+1$), we must have $D_n(\hat r)=\ell_n+1$, equivalently $\hat r<r_n$, and moreover $\hat r$ is located exactly at $c$, so $c=([\hat r]^{m+1}_{\ell_n+1},\ell_n+1)$.

  Set $\hat q\mathrel{:=}\hat r\bmod k^m$ and $\hat i\mathrel{:=}\lfloor\hat r/k^m\rfloor$, so that $\hat r=\hat i k^m+\hat q$. From $\hat r\in E_m$ together with $t(\hat r)<n$ we obtain $\hat q\in Q_m^{(n-1)}$. Set $\tilde q\mathrel{:=}\rho_{m+1}^{(n-1)}(\hat r)\in Q_{m+1}^{(n-1)}$; unfolding \eqref{eq:route-def} with the Minimality consequence $\rho_m^{(n-1)}(\hat q)=\hat q$ yields
\[
  \delta_m^{(n-1)}(\hat q,\hat i)=\rho_{m+1}^{(n-1)}(\hat r)=\tilde q.
\]

  We first verify $\tilde q\in D$. Minimality applied to $\tilde q\in Q_{m+1}^{(n-1)}$ gives $\rho_{m+1}^{(n-1)}(\tilde q)=\tilde q$; combined with $\rho_{m+1}^{(n-1)}(\hat r)=\tilde q$, we get
\[
  \mathsf{Eval}^{(n-1)}(k^{m+1}j+\tilde q)
  =
  \operatorname{PEval}_{m+1}^{(n-1)}(\tilde q,j)
  =
  \mathsf{Eval}^{(n-1)}(k^{m+1}j+\hat r)
\]
  for every $j<\ell_n$. For these $j$, both indices lie strictly below $k^{m+1}\ell_n\leq n$, hence both are at most $n-1$. Data Consistency converts this $\mathsf{Eval}$-equality into $x[k^{m+1}j+\tilde q]=x[k^{m+1}j+\hat r]$, whence $\tilde q\approx^{m+1}_{\ell_n}\hat r$. Since $\hat r$ lies in the subtree of $u_n$, $\hat r\approx^{m+1}_{\ell_n}r_n$; transitivity gives $\tilde q\approx^{m+1}_{\ell_n}r_n$, i.e. $\tilde q\in D$ by part (3) of Lemma~\ref{lem:tp-semantics}.

  Now apply part (5) of Lemma~\ref{lem:tp-semantics} at $(\hat q,\hat i)$: the hypotheses $\hat q\in Q_m^{(n-1)}$, $\hat r<r_n$, and $\delta_m^{(n-1)}(\hat q,\hat i)=\tilde q\in D$ yield $\hat r\approx^{m+1}_{\ell_n+1}\tilde q$, hence $c^{\tilde q}=c$. Combined with $\nu(\tilde q)=N_n(c^{\tilde q})=N_n(c)\geq1$, this exhibits a preimage of $c$.
\end{proof}

Since $u_n$ is game-internal by hypothesis, it has at least one child in the game tree, so $\{\tilde q\in D:\nu(\tilde q)\geq1\}$ is non-empty by Lemma~\ref{lem:b2-children-id}. Line 11 of Algorithm~\ref{alg:tp-predict} selects $q^*\mathrel{:=}\operatorname*{arg\,max}_{\tilde q\in D}\nu(\tilde q)$, satisfying $\nu(q^*)\geq1$, and another application of Lemma~\ref{lem:b2-children-id} makes $c^{q^*}$ a child of $u_n$ in the game tree. By Lemma~\ref{lem:population}, $N_n(c^{q^*})=\nu(q^*)=\max_{\tilde q\in D}\nu(\tilde q)$, which by Lemma~\ref{lem:b2-children-id} equals $\max_c N_n(c)$ over game children $c$ of $u_n$. Hence, with ties broken in the manner agreed in this subsection (cf. the remark following Lemma~\ref{lem:depth-formula}), the PB strategy at $u_n$ predicts $c^{q^*}$.

When the time-$n$ activation is processed, nature moves $r_n^{(m)}$ from $u_n$ to its level-$(\ell_n+1)$ subclass child, namely $c^{r_n}=([r_n]^{m+1}_{\ell_n+1},\ell_n+1)$. The PB strategy commits a mistake at this round iff $c^{q^*}\neq c^{r_n}$, equivalently $q^*\not\approx^{m+1}_{\ell_n+1}r_n$. By Lemma~\ref{lem:tp-mistake-criterion} this is precisely the criterion for a TP mistake at round $n$.

Therefore, every TP mistake at a round $n\in R_m$ with $u_n$ game-internal is simultaneously a mistake of the PB strategy in the level-$m$ branch-prediction game.

\emph{Leaf-vertex mistakes at level $m$.} If $u_n$ is a game-leaf just before the time-$n$ activation, the activation is of type (b.1) and produces no mistake in the game. However, TP still issues a prediction, which we must bound separately.

\begin{lemma}
\label{lem:b1-branching}
  \emph{(b.1 mistakes require a branching vertex.)} If TP makes a mistake at a round $n\in R_m$ with $u_n$ a game-leaf, then $u_n$ has at least two children in $T^{(m)}$.
\end{lemma}

\begin{proof}
  By Lemma~\ref{lem:tp-mistake-criterion}, a TP mistake at round $n$ gives $q^*\not\approx^{m+1}_{\ell_n+1}r_n$. By Lemma~\ref{lem:tp-semantics} part (3), $q^*\approx^{m+1}_{\ell_n}r_n$, so $q^*$ and $r_n$ lie in the same level-$\ell_n$ class $C_n$ with $u_n=(C_n,\ell_n)$, but in distinct $\approx^{m+1}_{\ell_n+1}$-classes. Both residues are live at level $\ell_n+1$, and therefore these two distinct classes are two distinct children of $u_n$ in $T^{(m)}$.
\end{proof}

A vertex $v$ of $T^{(m)}$ with at least two children is said to be \emph{branching}. A key consequence of Lemma~\ref{lem:psi-refines} is the following bound.

\begin{lemma}
\label{lem:branching-count}
  \emph{(Branching bound.)} The number of branching vertices of $T^{(m)}$ is at most $s-1$.
\end{lemma}

\begin{proof}
  For a vertex $v=(C,\ell)$ of $T^{(m)}$, define
\[
  \Psi(v)\mathrel{:=}\{\psi^*_{m+1}(r):r\in C\}\subseteq Q^*.
\]
  This set is non-empty. If $u_1,\dots,u_d$ are the children of $v$, then the sets $\Psi(u_1),\dots,\Psi(u_d)$ are pairwise disjoint non-empty subsets of $\Psi(v)$: disjointness follows because two residues with the same $\psi^*_{m+1}$-value are $\approx^{m+1}_{\ell+1}$-equivalent by Lemma~\ref{lem:psi-refines}, and hence cannot lie in distinct children.

  We prove, by induction from the leaves upward, that the number of branching vertices in the subtree rooted at any vertex $v$ is at most $|\Psi(v)|-1$. If $v$ has no children, this is immediate. If $v$ has exactly one child $u$, the induction hypothesis gives at most $|\Psi(u)|-1\leq |\Psi(v)|-1$ branching vertices below $v$, and $v$ itself is not branching. If $v$ has $d\geq2$ children $u_1,\dots,u_d$, then
\[
  \#\{\text{branching vertices below }v\}
  \leq
  1+\sum_{i=1}^d (|\Psi(u_i)|-1)
  \leq
  1+|\Psi(v)|-d
  \leq
  |\Psi(v)|-1.
\]
  Applying this to the root, for which $|\Psi(\Lambda^{m+1}_0,0)|\leq |Q^*|=s$, gives the claim.
\end{proof}

\textbf{Mistake bound.} We are now ready to assemble the bound \eqref{eq:rtl-mistakes}. For each $m\in\{0,\dots,M-1\}$, let $A_m$ denote the number of rounds $n\in R_m$ at which TP makes a mistake and $u_n$ is game-internal, and let $B_m$ denote the corresponding count for $u_n$ a game-leaf. Since the sets $\{R_m\}_{m<M}$ partition $\mathcal{R}$ and each boundary round contributes at most one mistake,
\[
  M_{\mathrm{TP}}(x)\leq |\mathcal{B}|+\sum_{m=0}^{M-1}(A_m+B_m).
\]

\emph{Internal-vertex contribution.} The argument preceding Lemma~\ref{lem:b1-branching} established that every round counted by $A_m$ is simultaneously a mistake of the $\mathsf{PB}$ strategy in the level-$m$ branch-prediction game. Applying Proposition~\ref{prop:branch-pred} to that game and bounding the element set via Lemma~\ref{lem:Em-bound},
\[
  A_m\leq 3|E_m|\left\lfloor\log_2|E_m|\right\rfloor \leq 3ks\log_2(ks).
\]

\emph{Leaf-vertex contribution.} In Definition~\ref{def:bp-game}, an activation at a game-leaf $v$ is a b.1-round and attaches a new child to $v$; thereafter $v$ has at least one child in the game tree and game-tree edges are never removed, so $v$ remains game-internal for the rest of the play. Consequently, at most one b.1-round can occur at any single vertex over the course of the level-$m$ game. By Lemma~\ref{lem:b1-branching}, every round counted by $B_m$ takes place at a vertex of $T^{(m)}$ that is \emph{branching} --- has at least two children in $T^{(m)}$. Hence the map $n\mapsto u_n$ injects the rounds counted by $B_m$ into the set of branching vertices of $T^{(m)}$, and Lemma~\ref{lem:branching-count} yields
\[
  B_m\leq \left|\{\text{branching vertices of }T^{(m)}\}\right|\leq s-1.
\]

\emph{Aggregation.} Summing the two contributions over levels,
\[
  M_{\mathrm{TP}}(x)
  \leq
  |\mathcal{B}|+\sum_{m=0}^{M-1}(3ks\log_2(ks)+(s-1))
  \leq
  |\mathcal{B}|+M\cdot(3ks\log_2(ks)+s).
\]
Recall that $M=\left\lceil\log_k |x|\right\rceil$ and $|\mathcal{B}|=O(\log_k |x|)$. Therefore
\[
  M_{\mathrm{TP}}(x)=O(ks\log(ks)\log |x|),
\]
which is precisely \eqref{eq:rtl-mistakes} upon writing $m\mathrel{:=}\mathrm{AC}^{\mathrm{R}}_k(x)=s$ as in Theorem~\ref{thm:rtl-bounds}.

\subsubsection{Compression Bound}
\label{sec:tp-compression}

We now establish the compression bound \eqref{eq:rtl-compression} of Theorem~\ref{thm:rtl-bounds}. We exhibit a binary encoding of the predictor's persistent state and bound its length component-by-component, using Lemma~\ref{lem:state-bound} to control $\left|Q_m\right|$.

Since $\mathsf{Update}$ and $\mathsf{Predict}$ depend explicitly on the time index $t$ (which appears there as the ``$n$" argument), we need to include it explicitly in the state in order to match Definition~$\ref{def:predictor}$. The state is hence the tuple
\[
  s_t \mathrel{:=} \bigl(t,\; m^*,\; \{Q_m\}_{m \leq m^*},\; \{\delta_m\}_{m < m^*},\; \tau\bigr).
\]

\textbf{A bound on the stored residues.} After the complete execution of $\mathsf{Update}(t, x[t])$ (or after $\mathsf{Initialize}$, where $t = 0$), every $q \in Q_m$ at every level $m \leq m^*$ satisfies
\begin{equation}
  \label{eq:residue-leq-n}
  q \leq n.
\end{equation}
Indeed, $\mathsf{Initialize}$ adds only the value $0$, and a $\mathsf{Preprocess}$ promotion merely copies existing labels; the only mechanism that introduces a fresh value is line~11 of Algorithm~\ref{alg:tp-update} during some $\mathsf{Update}(t, \cdot)$ with $t \leq n$, which adds $q_{\mathrm{new}} = t \bmod k^{m+1} \leq t \leq n$. Since nothing is ever removed from a $Q_m$, \eqref{eq:residue-leq-n} holds. Consequently each residue fits in
\[
  W_t \mathrel{:=} \left\lceil \log_2 \max(t + 1, 2) \right\rceil = O(\log n)
\]
bits.

\textbf{Encoding scheme.} We concatenate five self-delimiting blocks.

\begin{enumerate}
  \item \emph{Index.} The integer $t$ in Elias gamma code, occupying $O(\log(n + 1))$ bits.
  \item \emph{Depth.} The integer $m^*$ in Elias gamma code, occupying $O(\log(m^* + 1))$ bits.
  \item \emph{State sets.} For each $m = 0, 1, \dots, m^*$ in turn: the size $\left|Q_m\right|$ in Elias gamma code, followed by the elements of $Q_m$ listed in increasing order, each as a fixed-width $W_t$-bit binary representation of its value.
  \item \emph{Transitions.} For each $m = 0, 1, \dots, m^* - 1$ in turn: the table $\delta_m$ as a $\left|Q_m\right| \times k$ matrix, where the $(q, i)$ entry stores the rank of $\delta_m (q, i)$ within $Q_{m+1}$ (with elements numbered $0, \dots, \left|Q_{m+1}\right| - 1$ in increasing order), as a fixed-width $\left\lceil \log_2 \max(\left|Q_{m+1}\right|, 2) \right\rceil$-bit field.
  \item \emph{Output function.} The values $\tau(q)$ for $q \in Q_{m^*}$ taken in increasing order, each as a fixed-width $\left\lceil \log_2 \max(\sigma, 2) \right\rceil$-bit symbol.
\end{enumerate}

\textbf{Bounding the length.} By Lemma~\ref{lem:state-bound}, $\left|Q_m\right| \leq s$ for every $m \leq m^*$, where $s \mathrel{:=} \mathrm{AC}^{\mathrm{R}}_k (x)$. The depth satisfies $m^* \leq \left\lceil \log_k n \right\rceil = O(\log n)$ throughout the processing of $x$.
We bound each block in turn.

\emph{Block 1.} Contributes $O(\log(n + 1)) = O(\log n)$ bits.

\emph{Block 2.} Contributes $O(\log(m^* + 1)) = O(\log \log n)$ bits.

\emph{Block 3.} Level $m$ contributes $O(\log(\left|Q_m\right| + 1)) + \left|Q_m\right| \cdot W_t = O(\log(s + 1) + s \log n) = O(s \log n)$ bits. Summing over the $m^* + 1$ levels,
\[
  \sum_{m = 0}^{m^*} O(s \log n) = O\bigl((m^* + 1) \, s \log n\bigr) = O\bigl(s (\log n)^2\bigr).
\]

\emph{Block 4.} Level $m$ contributes $\left|Q_m\right| \cdot k \cdot \left\lceil \log_2 \max(\left|Q_{m+1}\right|, 2) \right\rceil = O(s k \log s)$ bits. Summing over the $m^*$ levels,
\[
  \sum_{m = 0}^{m^* - 1} O(s k \log s) = O(s k \log s \cdot m^*) = O(s k \log s \log n).
\]

\emph{Block 5.} Contributes $\left|Q_{m^*}\right| \left\lceil \log_2 \max(\sigma, 2) \right\rceil = O(s \log \sigma)$ bits.

\textbf{Aggregation.} Combining the five block bounds, and absorbing $\log n$ and $\log \log n$ into the $s (\log n)^2$ term, the total state length satisfies
\[
  S_{\mathrm{TP}} (x) = O\bigl(s (\log n)^2 + s k \log s \log n + s \log \sigma\bigr) = O\bigl(s(k \log s \log n + (\log n)^2 + \log \sigma)\bigr),
\]
which is precisely \eqref{eq:rtl-compression} upon writing $m \mathrel{:=} \mathrm{AC}^{\mathrm{R}}_k (x) = s$ as in Theorem~\ref{thm:rtl-bounds}.

\section{Anytime Multivalued Online Learning}
\label{sec:apx-muwu}

In this appendix we develop the techniques underlying our predictor for ARC. We will work in a setting borrowed from \cite{Joshi2026}, which in their terminology is ``online mistake-unaware learning for contextual bandits", but we prefer to call ``multivalued online learning".

\subsection{Learning Protocol}

Consider the following online learning protocol. Fix an instance space $X$, a finite label space $Y$, and a hypothesis class $\mathcal{H} \subseteq \{X \to 2^Y\}$. Each hypothesis $h \in \mathcal{H}$ maps an instance $x \in X$ to a subset $h(x) \subseteq Y$ of ``admissible'' labels. A hypothesis $h$ with $|h(x)| = 1$ for all $x$ corresponds to a standard (single-valued) hypothesis; a hypothesis with $|h(x)|\leq1$ corresponds to a ``partial concept" in the sense of \cite{Alon2021}; the additional generality of allowing $|h(x)| > 1$ is essential for our applications to sequence prediction.

The learning protocol proceeds in rounds $t = 0,1,2,\dots$ There is an unknown \emph{target} hypothesis $h^* \in \mathcal{H}$.

\begin{enumerate}
  \item The learner receives an instance $x_t \in X$.
  \item The learner outputs a prediction $\hat{y}_t \in Y$.
  \item The true label $y_t \in Y$ is revealed, where $y_t \in h^*(x_t)$.
\end{enumerate}

A round $t$ is a \emph{mistake} if $\hat{y}_t \notin h^*(x_t)$. Note that this is strictly more permissive than requiring $\hat{y}_t = y_t$: the prediction $\hat{y}_t$ is considered correct whenever $\hat{y}_t \in h^*(x_t)$, even if $\hat{y}_t \neq y_t$. Crucially, however, the learner's only feedback is the observed label $y_t$---it does not see $h^*(x_t)$ and therefore does not know whether a given prediction was correct (except when $\hat{y}_t = y_t$, which is guaranteed correct since $y_t \in h^*(x_t)$).

\cite{Joshi2026} propose the Mistake Unaware Weight Update (MUWU) algorithm: see Algorithm~\ref{alg:muwu}. It guarantees the following bound (Theorem 4 there):

\begin{theorem}[Joshi et al.]
\label{thm:online}
Let $\mathcal{H}$ be a finite non-empty hypothesis class of multivalued hypotheses. MUWU guarantees that for any target $h^* \in \mathcal{H}$ and any sequence of instances and labels consistent with $h^*$, the total number of mistakes satisfies
\begin{equation}
  \left|\left\{t : \hat{y}_t \notin h^*(x_t)\right\}\right| \leq \log_2 |\mathcal{H}|.
  \label{eq:muwu-mistakes}
\end{equation}
Moreover, if membership $y \in h(x)$ can be decided in time $T$ for any $h \in \mathcal{H}$, $x \in X$, and $y \in Y$, then the per-round time complexity of MUWU is $\operatorname{poly}(|\mathcal{H}|, |Y|, T)$.
\end{theorem}

We also consider a variant where the hypothesis class can ``grow'' over time. Formally, fix some family $\{X_n \subseteq X\}_{n\in\mathbb{N}}$ s.t. $\bigcup_n X_n=X$. For each $n \in \mathbb{N}$, define $\mathcal{H}_n \subseteq \mathcal{H}$ by
\begin{equation*}
  \mathcal{H}_n \mathrel{:=} \left\{h \in \mathcal{H} \mid \exists m < n,\ x \in X_m : h(x) \neq Y\right\}.
\end{equation*}

Thus $\mathcal{H}_n$ contains the hypotheses that can make a nontrivial restriction on the label of some instance available before time $n$.

We now restrict our protocol s.t. $x_t$ has to lie in $X_t$. We allow $\mathcal{H}$ to be infinite, but each $\mathcal{H}_n$ must be finite. For this case we propose the Anytime-MUWU algorithm (see Algorithm~\ref{alg:muwu-anytime}), yielding the following:

\begin{theorem}
\label{thm:online-anytime}
In the setting above, for any target $h^* \in \mathcal{H}$ and any sequence of instances and labels obeying the protocol and consistent with $h^*$, for any $n \in \mathbb{N}$, Anytime-MUWU guarantees
\begin{equation*}
  \left|\left\{t : t < n \wedge \hat{y}_t \notin h^*(x_t)\right\}\right| \leq \log_2 \max(|\mathcal{H}_n|, 1).
\end{equation*}
Moreover, if it's possible to iterate over $h \in \mathcal{H}_{n+1}$ and decide membership $y \in h(x)$ in time $\operatorname{poly}(n)$ per hypothesis for any $x \in X_n$ and $y \in Y$, then the time complexity of running Anytime-MUWU on round $n$ is $\operatorname{poly}(|\mathcal{H}_{n+1}|, |Y|, n)$.
\end{theorem}

Theorem~\ref{thm:online-anytime} is an anytime guarantee: Anytime-MUWU does not require knowing in advance the horizon $n$ at which it will be evaluated.


\subsection{The Algorithm}
\label{sec:muwu-alg}

The persistent state of MUWU consists of:
\begin{itemize}
  \item A subset $\mathcal{H}^{\mathrm{uf}} \subseteq \mathcal{H}$, the set of currently \emph{unfalsified} hypotheses.
  \item For each $h \in \mathcal{H}^{\mathrm{uf}}$, a non-negative integer counter $m_h \in \mathbb{N}$.
\end{itemize}

These are initialized as $\mathcal{H}^{\mathrm{uf}} \leftarrow \mathcal{H}$ and $m_h \leftarrow 0$ for every $h \in \mathcal{H}$. Each round $t$ proceeds according to Algorithm~\ref{alg:muwu}.

\begin{algorithm}[htbp]
\caption{\textbf{(Joshi et al.)} Mistake-Unaware Weight Update}
\label{alg:muwu}
  \begin{enumerate*}
    \item receive instance $x_t \in X$
    \item \textbf{for} $y \in Y$
      \begin{enumerate*}
        \item $W(y) \leftarrow \sum_{h \in \mathcal{H}^{\mathrm{uf}}\, :\, y \in h(x_t)} 2^{m_h}$
      \end{enumerate*}
    \item $\hat{y}_t \leftarrow \operatorname*{arg\,max}_{y \in Y} W(y)$ \hfill \textnormal{// ties broken arbitrarily}
    \item predict $\hat{y}_t$
    \item receive label $y_t \in Y$
    \item \textbf{for} $h \in \mathcal{H}^{\mathrm{uf}}$
      \begin{enumerate*}
        \item \textbf{if} $y_t \notin h(x_t)$ \textbf{then} remove $h$ from $\mathcal{H}^{\mathrm{uf}}$
        \item \textbf{if} $\hat{y}_t \notin h(x_t)$ \textbf{then} $m_h \leftarrow m_h + 1$
      \end{enumerate*}
  \end{enumerate*}
\end{algorithm}

We use $\mathcal{H}^{\mathrm{uf}}_t$ and $m_h^{(t)}$ to denote the values of the corresponding variables \emph{at the start} of round $t$, i.e., after rounds $0, \dots, t - 1$ have been fully processed but before round $t$ has begun. In particular, $\mathcal{H}^{\mathrm{uf}}_0 = \mathcal{H}$ and $m_h^{(0)} = 0$ for every $h$.

\subsection{Proof of Theorem~\ref{thm:online}}
\label{sec:muwu-proof}

We reproduce the proof for the reader's convenience and to more easily explain the modified analysis of Anytime-MUWU in the following subsections.

\textbf{Mistake bound.} Define the potential
\begin{equation*}
  \Phi_t \mathrel{:=} \sum_{h \in \mathcal{H}^{\mathrm{uf}}_t} 2^{m_h^{(t)}}.
\end{equation*}
Initially, $\Phi_0 = |\mathcal{H}|$.

\begin{lemma}
\label{lem:muwu-potential}
  For every round $t$, $\Phi_{t+1} \leq \Phi_t$.
\end{lemma}

\begin{proof}
  Fix a round $t$, drop the round subscript for brevity, and partition $\mathcal{H}^{\mathrm{uf}} = \mathcal{H}^{\mathrm{uf}}_t$ according to the membership of $\hat{y}_t$ and $y_t$ in $h(x_t)$:
  \[
  \begin{aligned}
  A_{++} &\mathrel{:=} \sum_{h \in \mathcal{H}^{\mathrm{uf}}\, :\, \hat{y}_t \in h(x_t),\; y_t \in h(x_t)} 2^{m_h},
  &\quad
  A_{+-} &\mathrel{:=} \sum_{h \in \mathcal{H}^{\mathrm{uf}}\, :\, \hat{y}_t \in h(x_t),\; y_t \notin h(x_t)} 2^{m_h},\\
  A_{-+} &\mathrel{:=} \sum_{h \in \mathcal{H}^{\mathrm{uf}}\, :\, \hat{y}_t \notin h(x_t),\; y_t \in h(x_t)} 2^{m_h},
  &\quad
  A_{--} &\mathrel{:=} \sum_{h \in \mathcal{H}^{\mathrm{uf}}\, :\, \hat{y}_t \notin h(x_t),\; y_t \notin h(x_t)} 2^{m_h}.
  \end{aligned}
  \]
  These four sums add up to $\Phi_t$. The two weighted votes used by the algorithm decompose as
  \begin{equation}
    W(\hat{y}_t) = A_{++} + A_{+-}, \quad W(y_t) = A_{++} + A_{-+}.
    \label{eq:muwu-W-decomp}
  \end{equation}

  We compute $\Phi_{t + 1}$. The falsification step removes from $\mathcal{H}^{\mathrm{uf}}$ exactly the hypotheses with $y_t \notin h(x_t)$, i.e., the $A_{+-}$ and $A_{--}$ buckets. The counter-update step then increments $m_h$ on every surviving hypothesis with $\hat{y}_t \notin h(x_t)$, i.e., the $A_{-+}$ bucket. The $A_{++}$ bucket survives without modification. Hence
  \[
    \Phi_{t + 1} = A_{++} + 2 A_{-+}.
  \]

  Subtracting,
  \begin{equation}
    \Phi_t - \Phi_{t + 1} = A_{+-} + A_{--} - A_{-+}.
    \label{eq:muwu-diff}
  \end{equation}
  By the choice of $\hat{y}_t$ as an argmax of $W$, we have $W(\hat{y}_t) \geq W(y_t)$, which by \eqref{eq:muwu-W-decomp} is equivalent to $A_{+-} \geq A_{-+}$. Substituting into \eqref{eq:muwu-diff} yields $\Phi_t - \Phi_{t + 1} \geq A_{--} \geq 0$.
\end{proof}

We can now bound the number of true mistakes. Fix any horizon $T \in \mathbb{N}$.

Since the sequence is consistent with $h^*$, $y_t \in h^*(x_t)$ for every $t < T$, so $h^*$ is never removed from $\mathcal{H}^{\mathrm{uf}}$ during rounds $0, \dots, T - 1$. In particular, $h^* \in \mathcal{H}^{\mathrm{uf}}_T$, and therefore
\[
  \Phi_T \geq 2^{m_{h^*}^{(T)}}.
\]
Iterating Lemma~\ref{lem:muwu-potential} and using $\Phi_0 = |\mathcal{H}|$,
\[
  2^{m_{h^*}^{(T)}} \leq \Phi_T \leq \Phi_0 = |\mathcal{H}|,
\]
so $m_{h^*}^{(T)} \leq \log_2 |\mathcal{H}|$.

It remains to identify $m_{h^*}^{(T)}$ with the count of true mistakes. The counter $m_{h^*}$ is incremented at round $t$ if and only if (i) $h^*$ remains in $\mathcal{H}^{\mathrm{uf}}$ after the falsification step of round $t$, and (ii) $\hat{y}_t \notin h^*(x_t)$. Condition (i) holds at every round (as observed above, $h^*$ is never falsified), so $m_{h^*}$ is incremented at round $t$ if and only if $\hat{y}_t \notin h^*(x_t)$, i.e., exactly when round $t$ is a true mistake. Summing,
\[
  \left|\left\{t < T : \hat{y}_t \notin h^*(x_t)\right\}\right| = m_{h^*}^{(T)} \leq \log_2 |\mathcal{H}|.
\]

\subsection{Anytime MUWU}
\label{sec:muwu-anytime-alg}

We now adapt Algorithm~\ref{alg:muwu} to the setting of Theorem~\ref{thm:online-anytime}, where the ambient class $\mathcal{H}$ may be infinite but the finite classes $\mathcal{H}_n$ become available over time. Recall that
\[
  \mathcal{H}_n \mathrel{:=} \left\{h \in \mathcal{H} \mid \exists m < n,\ x \in X_m : h(x) \neq Y\right\}.
\]
The classes $\mathcal{H}_n$ are monotone in $n$, and $\mathcal{H}_0 = \emptyset$.

Anytime-MUWU (Algorithm~\ref{alg:muwu-anytime}) keeps two finite sets. The first, $\mathcal{A}$, is the set of \emph{activated} hypotheses, and the second, $\mathcal{U} \subseteq \mathcal{A}$, is the set of activated hypotheses that are still unfalsified. For every $h \in \mathcal{U}$ it also keeps a counter $m_h$. Initially $\mathcal{A} = \mathcal{U} = \emptyset$. On round $t$, before the prediction is made, the learner activates all hypotheses that have just entered $\mathcal{H}_{t + 1}$, i.e. all hypotheses in
\begin{equation*}
  \mathcal{N}_t \mathrel{:=} \left\{h \in \mathcal{H}_{t + 1} \mid h \notin \mathcal{H}_t\right\}.
\end{equation*}
Each $h \in \mathcal{N}_t$ is inserted into both $\mathcal{A}$ and $\mathcal{U}$ with $m_h \mathrel{:=} 0$. The weighted plurality prediction and the subsequent falsification/counter update are then exactly as in the finite-class algorithm, but restricted to $\mathcal{U}$.

\begin{algorithm}[htbp]
\caption{Anytime Mistake-Unaware Weight Update}
\label{alg:muwu-anytime}
  \begin{enumerate*}
    \item receive instance $x_t \in X_t$
    \item \textbf{for} $h \in \mathcal{H}_{t + 1}$ with $h \notin \mathcal{A}$
      \begin{enumerate*}
        \item insert $h$ into $\mathcal{A}$ and $\mathcal{U}$
        \item $m_h \leftarrow 0$
      \end{enumerate*}
    \item \textbf{for} $y \in Y$
      \begin{enumerate*}
        \item $W(y) \leftarrow \sum_{h \in \mathcal{U}\, :\, y \in h(x_t)} 2^{m_h}$
      \end{enumerate*}
    \item $\hat{y}_t \leftarrow \operatorname*{arg\,max}_{y \in Y} W(y)$ \hfill \textnormal{// ties broken arbitrarily}
    \item predict $\hat{y}_t$
    \item receive label $y_t \in Y$
    \item \textbf{for} $h \in \mathcal{U}$
      \begin{enumerate*}
        \item \textbf{if} $y_t \notin h(x_t)$ \textbf{then} remove $h$ from $\mathcal{U}$
        \item \textbf{else if} $\hat{y}_t \notin h(x_t)$ \textbf{then} $m_h \leftarrow m_h + 1$
      \end{enumerate*}
  \end{enumerate*}
\end{algorithm}

The activation step is harmless with respect to the past. Indeed, if $h \in \mathcal{N}_t$ and $s < t$, then $h(x_s) = Y$: otherwise $x_s \in X_s$ and $h(x_s) \neq Y$ would imply $h \in \mathcal{H}_{s + 1} \subseteq \mathcal{H}_t$, contradicting $h \notin \mathcal{H}_t$. Hence a newly activated hypothesis was neither falsified by any previous observation nor responsible for any previous true mistake. Initializing its counter to zero is therefore consistent with the same interpretation of $m_h$ as in Section~\ref{sec:muwu-alg}.

This algorithm is anytime in the usual sense: on round $t$ it uses only $t$, the current instance, the previous state, and an enumeration of $\mathcal{H}_{t + 1}$. It does not require knowing in advance the horizon $n$ at which it will be evaluated.

\subsection{Proof of Theorem~\ref{thm:online-anytime}}
\label{sec:muwu-anytime-proof}

Let $\mathcal{U}_n$ denote the set $\mathcal{U}$ after rounds $0, \dots, n - 1$ have been fully processed, equivalently before the activation step of round $n$. Let $m_h^{(n)}$ be the corresponding counter of $h \in \mathcal{U}_n$. By construction, $\mathcal{U}_n \subseteq \mathcal{H}_n$. Define the potential
\begin{equation*}
  \Phi_n \mathrel{:=} \sum_{h \in \mathcal{U}_n} 2^{m_h^{(n)}}.
\end{equation*}

\begin{lemma}
\label{lem:muwu-anytime-potential}
  For every $n \in \mathbb{N}$,
  \[
    \Phi_n \leq |\mathcal{H}_n|.
  \]
\end{lemma}

\begin{proof}
  Since $\mathcal{H}_0 = \emptyset$, we have $\Phi_0 = 0$. Fix a round $t$. Immediately after the activation step, the potential has become
  \begin{equation*}
    \tilde{\Phi}_t = \Phi_t + |\mathcal{N}_t|,
  \end{equation*}
  because every newly activated hypothesis is inserted with counter zero and hence contributes weight $1$.

  From this point until the end of the round, Algorithm~\ref{alg:muwu-anytime} performs exactly one MUWU update on the finite active unfalsified class $\mathcal{U}_t \cup \mathcal{N}_t$. The proof of Lemma~\ref{lem:muwu-potential} is local to a single round and does not use any property of the class beyond finiteness. Therefore the update after prediction and feedback cannot increase the potential, giving
  \begin{equation}
    \Phi_{t + 1} \leq \tilde{\Phi}_t = \Phi_t + |\mathcal{N}_t|.
    \label{eq:muwu-anytime-step}
  \end{equation}

  Iterating \eqref{eq:muwu-anytime-step} over $t = 0, \dots, n - 1$ yields
  \[
    \Phi_n \leq \sum_{t = 0}^{n - 1} |\mathcal{N}_t|.
  \]
  The sets $\mathcal{N}_t$ are disjoint and their union is $\mathcal{H}_n$, since $\mathcal{H}_0 = \emptyset$ and the sequence $\mathcal{H}_0 \subseteq \mathcal{H}_1 \subseteq \dots$ is monotone. Hence
  \[
    \Phi_n \leq \sum_{t = 0}^{n - 1} |\mathcal{N}_t| = |\mathcal{H}_n|,
  \]
  as claimed.
\end{proof}

We now prove the mistake bound. Fix a target $h^* \in \mathcal{H}$, a sequence consistent with $h^*$, and a horizon $n$. First suppose $h^* \notin \mathcal{H}_n$. Then for every $t < n$ we must have $h^*(x_t) = Y$: otherwise $x_t \in X_t$ and $h^*(x_t) \neq Y$ would imply $h^* \in \mathcal{H}_{t + 1} \subseteq \mathcal{H}_n$. Thus every prediction belongs to $h^*(x_t)$ for every $t < n$, and there are no true mistakes before the horizon.

It remains to consider the case $h^* \in \mathcal{H}_n$. Since the data are consistent with $h^*$, once $h^*$ is activated it is never removed from $\mathcal{U}$. Therefore $h^* \in \mathcal{U}_n$. Moreover, the counter $m_{h^*}$ is incremented exactly on the true-mistake rounds before $n$. Indeed, after activation this is precisely the counter-update rule in Algorithm~\ref{alg:muwu-anytime}; before activation, say on a round $t$ with $h^* \notin \mathcal{H}_{t + 1}$, we have $h^*(x_t) = Y$ by the same argument as above, so a true mistake is impossible. Consequently
\begin{equation}
  \left|\left\{t : t < n \text{ and } \hat{y}_t \notin h^*(x_t)\right\}\right| = m_{h^*}^{(n)}.
  \label{eq:muwu-anytime-counter}
\end{equation}

Using Lemma~\ref{lem:muwu-anytime-potential} and $h^* \in \mathcal{U}_n$,
\[
  2^{m_{h^*}^{(n)}} \leq \Phi_n \leq |\mathcal{H}_n|.
\]
Taking logarithms and substituting \eqref{eq:muwu-anytime-counter} gives
\[
  \left|\left\{t : t < n \text{ and } \hat{y}_t \notin h^*(x_t)\right\}\right| \leq \log_2 |\mathcal{H}_n|.
\]

\section{MUWU for ARC}
\label{sec:apx-muwu-arc}

In this appendix we prove Theorem~\ref{thm:arc-muwu}. The proof is based on the anytime multivalued online-learning protocol of Theorem~\ref{thm:online-anytime}, instantiated with hypotheses corresponding to equations between arithmetic progressions. We take the instance at round $t$ to be the length-$t$ prefix already observed, so $X_t = \Sigma^t$. The hypotheses are equations between \emph{infinite} arithmetic progressions, but they only make predictions from progression-index $1$ onward. This keeps the active class $\mathcal{H}_n$ polynomial in $n$. The omitted index-$0$ predictions, together with the literal positions, are charged directly to the arithmetic repetition system.

\subsection{Conjunctions of coarsened hypotheses}
\label{sec:arc-conj}

We first record a closure property for conjunctions. The statement is deliberately independent of the particular learning rule: it uses only a given mistake bound for the underlying learner. Throughout this subsection, $X$ is equipped with round-dependent sets $X_t$, the label space is $Y$, and $\mathcal{H}$ is an ambient multivalued hypothesis class.

Let $\mathcal{A}$ be an online learner for this protocol. We say that $\mathcal{A}$ satisfies the single-hypothesis mistake bound $\left(B_N\right)_{N \in \mathbb{N}}$ if, for every $N \in \mathbb{N}$, every $h \in \mathcal{H}$, and every sequence $\left(x_t,y_t\right)_{t<N}$ obeying $x_t \in X_t$ and $y_t \in h(x_t)$ for all $t<N$, the predictions of $\mathcal{A}$ on that sequence satisfy
\begin{equation}
  \left|\left\{t<N : \hat{y}_t \notin h(x_t)\right\}\right| \leq B_N.
  \label{eq:single-hyp-bound}
\end{equation}

\begin{definition}
\label{def:prefix-coarsening}
  Let $h \in \mathcal{H}$. A multivalued hypothesis $g:X \to 2^Y$ is a \emph{prefix coarsening} of $h$ with cutoff $\tau \in \mathbb{N}$ if
  \begin{itemize}
    \item for every $t < \tau$ and every $x \in X_t$, $g(x)=h(x)$;
    \item for every $t \geq \tau$ and every $x \in X_t$, $g(x)=Y$.
  \end{itemize}
  In particular, $h(x) \subseteq g(x)$ for every $x \in X$.
\end{definition}

If $C$ is a finite set of multivalued hypotheses, define its conjunction $h_C$ by
\[
  h_C (x) \mathrel{:=} \left\{y \in Y : \forall g \in C,\ y \in g(x)\right\}.
\]
Thus $h_C$ is the pointwise intersection of the hypotheses in $C$.

\begin{proposition}
\label{prop:conj-lift}
  Let $\mathcal{A}$ be any online learner satisfying the single-hypothesis mistake bound $\left(B_N\right)_{N \in \mathbb{N}}$. Let $C$ be a finite set of multivalued hypotheses such that every $g \in C$ is a prefix coarsening of some $h_g \in \mathcal{H}$ with cutoff $\tau_g$. Then, for every sequence $\left(x_t,y_t\right)_{t<n}$ obeying $x_t \in X_t$ and satisfying $y_t \in h_C (x_t)$ for all $t<n$, the predictions of $\mathcal{A}$ satisfy
  \[
    \left|\left\{t<n : \hat{y}_t \notin h_C (x_t)\right\}\right| \leq \sum_{g \in C} B_{\min(n,\tau_g)}.
  \]
  In particular, if $\left(B_N\right)$ is nondecreasing, then
  \[
    \left|\left\{t<n : \hat{y}_t \notin h_C (x_t)\right\}\right| \leq |C| \cdot B_n.
  \]
\end{proposition}

\begin{proof}
  Fix $g \in C$ and set $N_g\mathrel{:=}\min(n,\tau_g)$. On every round $t \geq \tau_g$ we have $g(x_t)=Y$, so no prediction can lie outside $g(x_t)$. On the remaining rounds $t < N_g$, consistency with $h_C$ implies $y_t \in g(x_t)=h_g (x_t)$. Hence the length-$N_g$ prefix of the sequence is consistent with the hypothesis $h_g \in \mathcal{H}$.

  Since $\mathcal{A}$ is an online learner, its first $N_g$ predictions in the length-$n$ run are the same as its predictions when run only on the truncated sequence of length $N_g$. Applying \eqref{eq:single-hyp-bound} to this truncated sequence gives
  \[
    \left|\left\{t<N_g : \hat{y}_t \notin h_g (x_t)\right\}\right| \leq B_{N_g}.
  \]
  Because $g=h_g$ before the cutoff and $g=Y$ after the cutoff, this also bounds $\left|\left\{t<n : \hat{y}_t \notin g(x_t)\right\}\right|$.

  Finally, if $\hat{y}_t \notin h_C (x_t)$, then $\hat{y}_t \notin g(x_t)$ for at least one $g \in C$. Taking the union of the sets of rounds charged to the hypotheses $g \in C$, and bounding its cardinality by the sum of their cardinalities, proves the first inequality. The second follows immediately when $\left(B_N\right)$ is nondecreasing.
\end{proof}

\subsection{Infinite arithmetic-progression specialists}
\label{sec:arc-specialists}

We now instantiate the hypothesis class. Let $X\mathrel{:=}\Sigma^*$ and $X_t\mathrel{:=}\Sigma^t$, and let $Y\mathrel{:=}\Sigma$. An \emph{infinite arithmetic progression} is a pair $\alpha=(r,p)$ with $r \in \mathbb{N}^+$ and $p \in \mathbb{N}$; it denotes the map
\[
  \alpha(s) \mathrel{:=} p + r s, \quad s \in \mathbb{N}.
\]
Let $\mathcal{P}$ be the set of all infinite arithmetic progressions.

For every ordered pair $(\alpha,\beta) \in \mathcal{P} \times \mathcal{P}$, define the equation specialist $h^{\mathrm{eq}}_{\alpha,\beta}:X \to 2^\Sigma$ as follows. Given a prefix $u \in \Sigma^*$, write $j\mathrel{:=}|u|$. The set $h^{\mathrm{eq}}_{\alpha,\beta}(u)$ consists of all $y \in \Sigma$ satisfying all requirements of the following form, for indices $s \geq 1$:
\begin{itemize}
  \item if $\alpha(s)=j$ and $\beta(s)<j$, then $y=u[\beta(s)]$;
  \item if $\beta(s)=j$ and $\alpha(s)<j$, then $y=u[\alpha(s)]$.
\end{itemize}
If there are no such requirements, then $h^{\mathrm{eq}}_{\alpha,\beta}(u)=\Sigma$. The restriction $s \geq 1$ is essential: the specialist does not predict at index $0$ of the progression pair. Let
\[
  \mathcal{H} \mathrel{:=} \left\{h^{\mathrm{eq}}_{\alpha,\beta} : \alpha,\beta \in \mathcal{P}\right\}.
\]
The predictor ARC-MUWU is the Anytime-MUWU algorithm of Theorem~\ref{thm:online-anytime} run on this class.

\begin{lemma}
\label{lem:arc-hyp-count}
  For the class $\mathcal{H}$ above,
  \[
    \left|\mathcal{H}_n\right| \leq n^4.
  \]
  Moreover, the hypotheses in $\mathcal{H}_n$ can be enumerated in polynomial time in $n$, and membership $y \in h(u)$ is decidable in polynomial time in $n$ for $h \in \mathcal{H}_n$, $u \in X_j$ with $j \leq n$, and $y \in \Sigma$.
\end{lemma}

\begin{proof}
  Suppose $h^{\mathrm{eq}}_{\alpha,\beta} \in \mathcal{H}_n$ with $\alpha(s)=p_\alpha+r_\alpha s$ and $\beta(s)=p_\beta+r_\beta s$. Then for some $j<n$ and some prefix $u \in X_j$, the set $h^{\mathrm{eq}}_{\alpha,\beta}(u)$ is not all of $\Sigma$. Hence at least one of the two progressions reaches the current index $j$ at some index $s \geq 1$, while the paired position at the same $s$ is earlier than $j$.

  If $\alpha(s)=j$, then $p_\alpha<n$ and $r_\alpha<n$, since $j=p_\alpha+r_\alpha s$ with $s \geq 1$. The paired inequality $\beta(s)<j<n$ similarly implies $p_\beta<n$ and $r_\beta<n$. The case $\beta(s)=j$ is symmetric. Thus every active specialist has a description with both starting points $<n$ and both steps $<n$. There are $\leq n^2$ such progressions and hence $\leq n^4$ ordered pairs.

  Enumeration is by these four parameters. To test membership on a prefix of length $j \leq n$, it suffices to check whether $j$ lies on either progression at an index $s \geq 1$ and, when it does, whether the paired position at the same index is earlier than $j$. There are at most two such indices to inspect, and the required comparisons and array lookups are polynomial time.
\end{proof}

In particular,
\begin{equation*}
  \log_2 \max\left(\left|\mathcal{H}_n\right|,1\right) = O(\log n).
\end{equation*}

\subsection{From repetition systems to coarsened specialists}
\label{sec:arc-ars-conj}

Let $\alpha=(r,i,a,b)$ be a finite arithmetic interval. We write $\alpha^\omega$ for the infinite arithmetic progression with the same step and initial value:
\[
  \alpha^\omega (s) \mathrel{:=} r a + i + r s.
\]
Thus $\alpha^\omega (s)=\alpha(s)$ for every $s \in \operatorname{dom}(\alpha)$.

Given a pair $e=(\alpha,\beta)$ of finite arithmetic intervals with $|\alpha|=|\beta|=\ell$, define $g_e:X \to 2^\Sigma$ exactly like the infinite specialist $h^{\mathrm{eq}}_{\alpha^\omega,\beta^\omega}$, except that only indices $s$ with $1 \leq s < \ell$ are allowed to generate requirements. In other words, $g_e$ is the finite equation associated with $e$, with the index-$0$ requirement deliberately omitted.

\begin{lemma}
\label{lem:finite-is-coarsening}
  For every such pair $e=(\alpha,\beta)$, the hypothesis $g_e$ is a prefix coarsening of $h^{\mathrm{eq}}_{\alpha^\omega,\beta^\omega}$. One valid cutoff is
  \[
    \tau_e \mathrel{:=} 1 + \max(\alpha(\ell - 1), \beta(\ell - 1)).
  \]
\end{lemma}

\begin{proof}
  Both $g_e$ and $h^{\mathrm{eq}}_{\alpha^\omega,\beta^\omega}$ ignore the index $s=0$. They also agree on every index $1 \leq s < \ell$. For $s \geq \ell$, both $\alpha^\omega (s)$ and $\beta^\omega (s)$ are strictly larger than the corresponding values at $s=\ell-1$. Hence the later of the two positions at index $s$ is at least $\tau_e$. Therefore no requirement with $s \geq \ell$ can be active on a round $t < \tau_e$. It follows that $g_e$ agrees with $h^{\mathrm{eq}}_{\alpha^\omega,\beta^\omega}$ on all $X_t$ for $t < \tau_e$.

  Conversely, all requirements generated by indices $1 \leq s < \ell$ occur no later than $\max(\alpha(\ell-1),\beta(\ell-1))$. Thus $g_e (x)=\Sigma$ for every $x \in X_t$ with $t \geq \tau_e$. This is precisely the definition of a prefix coarsening.
\end{proof}

\begin{lemma}
\label{lem:ars-conjunction}
  Let $u \in \Sigma^n$, and let $R=(L,\Phi)$ be an ARS for $u$. For each $e \in \Phi$, let $g_e$ be the coarsened equation above, and put $C_R \mathrel{:=} \{g_e : e \in \Phi\}$. Also define
  \[
    Z_R \mathrel{:=} \left\{t<n : \exists\, (\alpha,\beta) \in \Phi,\ \alpha(0) < \beta(0) = t\right\}.
  \]
  Then:
  \begin{enumerate}
    \item for every $t<n$, $u[t] \in h_{C_R}(u[:t])$;
    \item for every $t<n$ such that $(u[t],t) \notin L$ and $t \notin Z_R$, we have $h_{C_R}(u[:t])=\{u[t]\}$.
  \end{enumerate}
  Moreover, $|C_R| \leq |\Phi|$ and $|Z_R| \leq |\Phi|$.
\end{lemma}

\begin{proof}
  We first record a simple consequence of the equation condition in the definition of an ARS. If $e=(\alpha,\beta) \in \Phi$ and $\ell\mathrel{:=}|\alpha|=|\beta|$, then for every $s<\ell$ such that $\alpha(s)<n$ and $\beta(s)<n$,
  \[
    u[\alpha(s)] = u[\beta(s)].
  \]
  Indeed, the arithmetic intervals are strictly increasing on their common domain $[\ell]$, so when both positions are $<n$, the two symbols above occur in the same $s$-th position of the two words $u[\alpha]$ and $u[\beta]$. These words are equal because $R$ is an ARS for $u$.

  We prove the first claim. Fix $j<n$ and put $w\mathrel{:=}u[:j]$. It is enough to show that $u[j] \in g_e (w)$ for every $e=(\alpha,\beta) \in \Phi$. Let $\ell\mathrel{:=}|\alpha|=|\beta|$, and consider any requirement generated by $g_e$ on the prefix $w$. If this requirement is generated by an index $s$ with $1 \leq s < \ell$, $\alpha(s)=j$, and $\beta(s)<j$, then $w[\beta(s)]=u[\beta(s)]$ and the preceding observation gives
  \[
    u[j] = u[\alpha(s)] = u[\beta(s)] = w[\beta(s)],
  \]
  so $u[j]$ satisfies this requirement. The case $\beta(s)=j$ and $\alpha(s)<j$ is symmetric, giving
  \[
    u[j] = u[\beta(s)] = u[\alpha(s)] = w[\alpha(s)].
  \]
  Hence $u[j]$ satisfies every requirement imposed by $g_e$ on the prefix $w$. Since $e$ was arbitrary, $u[j]$ belongs to every $g_e (u[:j])$, and therefore
  \[
    u[j] \in h_{C_R}(u[:j]).
  \]

  We prove the second claim. Fix $j<n$ such that $(u[j],j) \notin L$ and $j \notin Z_R$, and put $w\mathrel{:=}u[:j]$. Since $R$ is an ARS for $u$ and the position $j$ is not literal, there are some $e=(\alpha,\beta) \in \Phi$ and some $s \in \operatorname{dom}(\alpha)$ such that
  \[
    \alpha(s) < \beta(s) = j.
  \]
  Writing $\ell\mathrel{:=}|\alpha|=|\beta|$, this means $s<\ell$. Moreover $s$ cannot be $0$: if $s=0$, then $\alpha(0)<\beta(0)=j$, which would put $j$ in $Z_R$. Hence $1 \leq s < \ell$. For this particular $e$, the finite equation $g_e$ imposes the requirement
  \[
    y = w[\alpha(s)].
  \]
  Since $\alpha(s)<j$, the right-hand side is $u[\alpha(s)]$. By the observation at the beginning of the proof,
  \[
    u[\alpha(s)] = u[\beta(s)] = u[j].
  \]
  Therefore $g_e (u[:j]) \subseteq \{u[j]\}$, and so also $h_{C_R}(u[:j]) \subseteq \{u[j]\}$. The first claim gives the reverse inclusion in the form $u[j] \in h_{C_R}(u[:j])$, hence
  \[
    h_{C_R}(u[:j]) = \{u[j]\}.
  \]

  Finally, $C_R$ is the image of the function from $\Phi$ that sends $e$ to $g_e$, so $|C_R| \leq |\Phi|$. Likewise, each $e=(\alpha,\beta) \in \Phi$ contributes at most one element $\beta(0)$ to $Z_R$, and duplicates can only decrease cardinality. Therefore $|Z_R| \leq |\Phi|$.
\end{proof}

\begingroup
\renewcommand{\proofname}{Proof of Theorem~\ref{thm:arc-muwu}.}
\begin{proof}
  Let $x \in \Sigma^n$, and let $R=(L,\Phi)$ be a minimum-size ARS for $x$, so $|L|+|\Phi|=\operatorname{ARC}(x)$. Write $m\mathrel{:=}\operatorname{ARC}(x)$. The predictor ARC-MUWU is the Anytime-MUWU learner on the infinite specialist class $\mathcal{H}$ defined above, with instances $X_t=\Sigma^t$.

  By Theorem~\ref{thm:online-anytime}, ARC-MUWU satisfies the single-hypothesis mistake bound
  \[
    B_N \mathrel{:=} \log_2 \max\left(\left|\mathcal{H}_N\right|,1\right).
  \]
  Since the classes $\mathcal{H}_N$ are monotone, $\left(B_N\right)$ is nondecreasing. By Lemma~\ref{lem:arc-hyp-count},
  \begin{equation}
    B_N = O(\log(N+1)).
    \label{eq:arc-muwu-single-bound}
  \end{equation}

  Let $C_R$ and $Z_R$ be as in Lemma~\ref{lem:ars-conjunction}. By Lemma~\ref{lem:finite-is-coarsening}, every member of $C_R$ is a prefix coarsening of a hypothesis in $\mathcal{H}$. By the first part of Lemma~\ref{lem:ars-conjunction}, the observed sequence is consistent with the conjunction $h_{C_R}$. Applying Proposition~\ref{prop:conj-lift} and \eqref{eq:arc-muwu-single-bound} gives
  \begin{equation}
    \left|\left\{t<n : \hat{x}_t \notin h_{C_R}(x[:t])\right\}\right| \leq |C_R| \cdot O(\log n) \leq O(|\Phi| \log n).
    \label{eq:arc-conj-mistakes}
  \end{equation}

  We now pass from mistakes against the conjunction to ordinary prediction mistakes. Rounds in
  \[
    P_R \mathrel{:=} \left\{t<n : (x[t],t) \in L\right\} \cup Z_R
  \]
  contribute at most $|P_R| \leq |L|+|\Phi|$ ordinary mistakes. On every round $t \notin P_R$, the second part of Lemma~\ref{lem:ars-conjunction} gives $h_{C_R}(x[:t])=\{x[t]\}$. Hence an ordinary mistake $\hat{x}_t \neq x[t]$ on such a round implies $\hat{x}_t \notin h_{C_R}(x[:t])$, and is counted in \eqref{eq:arc-conj-mistakes}. Therefore
  \[
    M_{\text{ARC-MUWU}} (x) \leq |L|+|\Phi|+O(|\Phi| \log n) = O(m \log n).
  \]
  This is the claimed $O(m \log n)$ bound.

\end{proof}
\endgroup

\section{Automaticity Separations}
\label{sec:apx-auto}

In Example~\ref{ex:thue_morse} we've seen the Thue-Morse sequence, in which both the LTR $2$-automaticity and the RTL $2$-automaticity of all prefixes are uniformly bounded. More generally, for any $u \in \Sigma^\omega$, $\mathrm{AC}^{\mathrm{R}}_k(u[:n]) = O(1)$ iff $\mathrm{AC}^{\mathrm{L}}_k(u[:n]) = O(1)$. Any such $u$ is called a \emph{$k$-automatic sequence}, and their study is an interesting subfield of combinatorics on words \citep{Allouche2003}. However, the magnitude of $\mathrm{AC}^{\mathrm{R}}_k$ and $\mathrm{AC}^{\mathrm{L}}_k$ can be quite different. Indeed, while properties like the sum of digits are direction-agnostic, other properties strongly favor one reading direction over the other. This leads to exponential separations in complexity.

\begin{example}
  Let $\Sigma=\{0,1\}$ and consider any $n\geq1$. Let $u_n \in \Sigma^{4^n}$ be defined s.t. $u_n[0]=0$ and for any $0<t<4^n$,
  \[
    u_n[t]=\langle t \rangle_2^{2 n}\bigl[\min(f(t) + n,2n-1)\bigr]
  \]
  where $f(t)<2n$ is defined to satisfy $\langle t \rangle_2^{2 n}[f(t)]=1$ and $\langle t \rangle_2^{2 n}[i]=0$ for any $i<f(t)$ (i.e. $f(t)$ is the position of the left-most 1 in the binary expansion of $t$).
  \begin{itemize}
    \item \textbf{LTR Automaticity:} $\mathrm{AC}^{\mathrm{L}}_2(u_n) = O(n)$. The automaton reads from the left (MSDF). It stays in an initial state until it sees the first ``1''. It then transitions to a chain of $n$ states to count down the positions. The output is determined by the bit seen at the $n$-th step after the first ``1''.
    \item \textbf{RTL Automaticity:} $\mathrm{AC}^{\mathrm{R}}_2(u_n) = \Omega(2^n)$. Reading from the right (LSDF), the automaton does not know when the string will end (identifying the most significant 1) until the process terminates. To correctly output the bit located $n$ positions away from the eventual most significant 1, it must maintain a sliding window of the last $n$ bits seen. This requires a state space exponential in $n$.
  \end{itemize}
  Let $v_n \in \Sigma^{4^n}$ be defined s.t. $v_n[0]=0$ and for any $0<t<4^n$,
  \[
    v_n[t]=\langle t \rangle_2^{2 n}\bigl[\max(g(t) - n,0)\bigr]
  \]
  where $g(t)<2n$ is defined to satisfy $\langle t \rangle_2^{2 n}[g(t)]=1$ and $\langle t \rangle_2^{2 n}[i]=0$ for any $i>g(t)$ (i.e. $g(t)$ is the position of the right-most 1 in the binary expansion of $t$).
  \begin{itemize}
    \item \textbf{RTL Automaticity:} $\mathrm{AC}^{\mathrm{R}}_2(v_n) = O(n)$, by a symmetric argument.
    \item \textbf{LTR Automaticity:} $\mathrm{AC}^{\mathrm{L}}_2(v_n) = \Omega(2^n)$, by a symmetric argument.
  \end{itemize}
\end{example}

The automaticity of a sequence is also highly sensitive to the choice of base $k$. A sequence that is simple in one base may be complex in another.

\begin{example}
\label{ex:valuation}
  Let $u$ be the infinite binary sequence defined by the parity of the 2-adic valuation of $t+1$. That is, $u[t] \equiv \nu_2(t+1) \pmod{2}$, where $\nu_2(n)$ is the maximal natural number s.t. $2^{\nu_2(n)}$ divides $n$.

  While it's easy to upper bound $\mathrm{AC}^{\mathrm{R}}_2$, it is possible to also show that automaticity in odd bases is high \citep[Theorem~15.4.3]{Allouche2003}. Specifically, for the prefix of length $n$:
  \begin{itemize}
    \item \textbf{Base 2:} $\mathrm{AC}^{\mathrm{R}}_2(u[:n]) = O(1)$.
    \item \textbf{Base $k$ (odd, $k \geq 3$):} $\mathrm{AC}^{\mathrm{R}}_k(u[:n]) = \Omega(\sqrt{n})$.
  \end{itemize}
\end{example}

This is a serious issue with automaticity-based predictors as candidate algorithms for practical applications, because in most cases there is no way to single out a preferred value of $k$. To overcome this limitation we need the arithmetic repetition complexity of Section~\ref{sec:arc}.

\section{Upper Bounds on ARC and ZLC}
\label{sec:apx-arc-bounds}

In this appendix we prove Proposition~\ref{prp:arc-vs-lzc}, Proposition~\ref{prp:arc-vs-zlc}, Proposition~\ref{prp:arc-vs-szlc} and Proposition~\ref{prp:zlc-vs-acr}.

\subsection{LZ77 and proof of Proposition~\ref{prp:arc-vs-lzc}}
\label{sec:apx-arc-vs-lzc}

We recall the definition of LZ77 compression, originally proposed in \citet{Ziv1977}. The empty word has the empty LZ77 factorization. For $n > 0$, an \emph{LZ77-type factorization} of a word $x \in \Sigma^n$ is a decomposition
\[
  x = f_0 f_1 \dots f_{z-1}
\]
into non-empty factors satisfying the following condition. For every $j < z$, denote the starting position of $f_j$ by
\[
  l_j \mathrel{:=} |f_0 f_1 \dots f_{j-1}|,
\]
with $l_0 = 0$. Then either:
\begin{itemize}
  \item $|f_j| = 1$. In this case, $f_j$ is a \emph{literal factor}.
  \item There exists a source position $i_j < l_j$ such that
  \[
    f_j = x[i_j : i_j + |f_j|].
  \]
  In this case, $f_j$ is a \emph{copy factor} with source $i_j$.
\end{itemize}

A copy factor is allowed to overlap with its source: it may happen that $i_j + |f_j| > l_j$.

The \emph{greedy LZ77 factorization} is the unique LZ77-type factorization in which every factor is chosen as long as possible given its starting position. That is, the following two conditions hold:
\begin{enumerate}
  \item For every literal factor $f_i$, the symbol $f_i$ doesn't appear in $x[:l_i]$.
  \item For every $j < z - 1$ there is no source position $i < l_j$ such that
  \[
    x[l_j : l_j + |f_j| + 1] = x[i : i + |f_j| + 1].
  \]
\end{enumerate}
The \emph{LZ77 complexity} $\operatorname{LZC}(x)$ is the number of factors in this greedy factorization. This is also the minimum number of factors in any LZ77-type factorization of $x$.\footnote{See Theorem~10 in \citet{Storer1982}.}

\begingroup
\renewcommand{\proofname}{Proof of Proposition~\ref{prp:arc-vs-lzc}.}
\begin{proof}
  If $x$ is empty, the empty ARS witnesses $\operatorname{ARC}(x)=0=\operatorname{LZC}(x)$. Assume from now on that $n\mathrel{:=}|x| > 0$. Let $z \mathrel{:=} \operatorname{LZC}(x)$, and choose an LZ77-type factorization
  \[
    x = f_0 f_1 \dots f_{z-1}
  \]
  with exactly $z$ factors. For every $j < z$, let $l_j \mathrel{:=} |f_0 f_1 \dots f_{j-1}|$ and $m_j \mathrel{:=} |f_j|$.

  We construct an arithmetic repetition system $R=(L,\Phi)$ for $x$. For every literal factor $f_j$, put the literal
  \[
    (x[l_j], l_j)
  \]
  into $L$. For every copy factor $f_j$ with source position $i_j < l_j$, define two arithmetic intervals of step $1$:
  \[
    \alpha_j \mathrel{:=} (1, 0, i_j, i_j + m_j), \quad \beta_j \mathrel{:=} (1, 0, l_j, l_j + m_j),
  \]
  and put $(\alpha_j,\beta_j)$ into $\Phi$.

  We verify that $R$ is an ARS in the sense of Definition~\ref{def:ars}. The literal constraints are correct by construction. If $f_j$ is a copy factor, then $|\alpha_j|=|\beta_j|=m_j$, and for every $t < m_j$,
  \[
    \alpha_j(t) = i_j + t, \quad \beta_j(t) = l_j + t.
  \]
  Since $f_j = x[i_j : i_j + m_j]$ and $f_j = x[l_j : l_j + m_j]$, we have
  \[
    x[\alpha_j] = x[i_j : i_j + m_j] = x[l_j : l_j + m_j] = x[\beta_j].
  \]
  Thus all equality constraints in $\Phi$ are valid.

  It remains to check coverage. Let $p < n$, and let $j$ be the unique index such that $l_j \leq p < l_j + m_j$. Write $p = l_j + t$ with $t < m_j$. If $f_j$ is a literal factor, then $m_j=1$ and $p=l_j$, so $(x[p],p) \in L$. If $f_j$ is a copy factor, then $(\alpha_j,\beta_j) \in \Phi$ satisfies
  \[
    \beta_j(t) = l_j + t = p, \quad \alpha_j(t) = i_j + t < l_j + t = p,
  \]
  because $i_j < l_j$. Hence $p$ is covered by an earlier position as required.

  Finally, each factor contributes exactly one object: a literal contributes one element to $L$, and a copy factor contributes one element to $\Phi$. Therefore
  \[
    \operatorname{ARC}(x) \leq |R| = |L| + |\Phi| = z = \operatorname{LZC}(x).
  \]
\end{proof}
\endgroup

\subsection{Proof of Proposition~\ref{prp:arc-vs-szlc}}
\label{sec:apx-arc-vs-szlc}

Let $x \in \Sigma^n$, and let $P=(Q,q_0,\delta)$ be a synchronous ZLP with $x \sqsubseteq \operatorname{val}(P)$. Write $\overline{Q} \mathrel{:=} Q \sqcup \Sigma$, $d_q \mathrel{:=} |\delta(q)|$ for $q \in Q$, and $\ell(z) \mathrel{:=} |\operatorname{val}_P(z)|$ for $z \in \overline{Q}$. The case $n=0$ is witnessed by the empty ARS, so assume $n>0$.

\textbf{Strategy.} Each internal vertex $q \in Q$ corresponds to a family of arithmetic-progression copies of $\operatorname{val}_P(q)$ sitting inside $x$, all of common step $\phi(q)$. We single out one of them as the \emph{representative} copy of $q$, and record the others through equality-with-the-representative equations in $\Phi$. Edges to terminal children produce literals in $L$. The bookkeeping is arranged so that each edge of $P$ contributes at most one constraint to $R$, giving $|R| \leq |P|$.

\textbf{Synchrony.} Fix a synchronizing map $\phi:Q\to\mathbb{N}$ for $P$:
\begin{equation}
  \phi(q_0)=1, \quad \text{and} \quad \phi(q') = d_q \cdot \phi(q) \ \text{whenever } q' \in Q \text{ appears in } \delta(q).
  \label{eq:szlc-sync}
\end{equation}

\textbf{Occurrence residues.} For each $q \in Q$, define recursively the set $C_q$ of \emph{occurrence residues}:
\begin{itemize}
  \item $C_{q_0} \mathrel{:=} \{0\}$;
  \item if $c \in C_q$ and $\delta(q)[j] = q' \in Q$ for some $j < d_q$, then $c + \phi(q)j \in C_{q'}$.
\end{itemize}
By \eqref{eq:szlc-sync}, $C_q \subseteq [\phi(q)]$. The ZLP axioms (acyclicity together with $q_0$ being the unique source) imply that every internal vertex is reachable from $q_0$, so every $C_q$ is non-empty. Set
\[
  c_q \mathrel{:=} \min C_q.
\]

\textbf{Key identity.} The geometric meaning of $C_q$ is captured by the following claim, which we prove first.

\begin{lemma}
\label{lem:szlp-occ}
For every $q \in Q$, every $c \in C_q$, and every $t \in \mathbb{N}$ with $c + \phi(q)t < n$:
\[
  t < \ell(q) \quad \text{and} \quad x[c+\phi(q)t] = \operatorname{val}_P(q)[t].
\]
\end{lemma}

\begin{proof}
We argue by induction on the depth of $q$ in the graph of $P$.

\emph{Base ($q=q_0$).} Here $c=0$ and $\phi(q_0)=1$, so the claim reduces to: $t<n$ implies $t<\ell(q_0)$ and $x[t]=\operatorname{val}_P(q_0)[t]$. Both follow from $x \sqsubseteq \operatorname{val}_P(q_0)$, which already requires $n \leq \ell(q_0)$.

\emph{Step ($q\neq q_0$).} Fix $c \in C_q$. By the recursive definition, $c$ enters $C_q$ via some triple $(q_p,c_p,j)$ with $q_p \in Q$, $c_p \in C_{q_p}$, $j < d_{q_p}$, $\delta(q_p)[j]=q$, and
\[
  c = c_p + \phi(q_p)j.
\]
Using $\phi(q)=d_{q_p}\phi(q_p)$ from \eqref{eq:szlc-sync}, set $s \mathrel{:=} d_{q_p}t+j$, so that
\[
  c + \phi(q)t = c_p + \phi(q_p)(d_{q_p}t+j) = c_p + \phi(q_p)s.
\]
The inductive hypothesis at $(q_p,c_p,s)$ gives $s < \ell(q_p)$ and
\begin{equation}
  x[c+\phi(q)t] = x[c_p+\phi(q_p)s] = \operatorname{val}_P(q_p)[s].
  \label{eq:szlc-occ-step}
\end{equation}

We must now relate $\operatorname{val}_P(q_p)[s]$ to $\operatorname{val}_P(q)[t]$. Unfolding
\[
  \operatorname{val}_P(q_p) = \operatorname{zip}(\operatorname{val}_P(\delta(q_p)[0]), \dots, \operatorname{val}_P(\delta(q_p)[d_{q_p}-1]))
\]
and writing $m \mathrel{:=} \min_{i<d_{q_p}} \ell(\delta(q_p)[i])$ for the shortest child length and $i_*$ for the smallest index achieving it, the definition of $\operatorname{zip}$ gives $\ell(q_p)=d_{q_p}m+i_*$. From $s=d_{q_p}t+j < d_{q_p}m+i_*$ we deduce $t \leq m$, with two sub-cases:
\begin{itemize}
  \item $t<m$: every child is at least this long, in particular $t<\ell(q)$;
  \item $t=m$ and $j<i_*$: by minimality of $i_*$, the index $j$ does not achieve the min, so $\ell(q)>m=t$.
\end{itemize}

Either way, $t<\ell(q)$. Furthermore, the $\operatorname{zip}$ identity \eqref{eq:zip} yields
\[
  \operatorname{val}_P(q_p)[s]
  = \operatorname{val}_P(\delta(q_p)[s \bmod d_{q_p}])\left[\left\lfloor\frac{s}{d_{q_p}}\right\rfloor\right]
  = \operatorname{val}_P(q)[t],
\]
which combined with \eqref{eq:szlc-occ-step} finishes the step.
\end{proof}

\textbf{Visible part of an occurrence.} For $c \in C_q$, define
\[
  B_q(c) \mathrel{:=} \left|\{t \in \mathbb{N} : c+\phi(q)t < n\}\right|.
\]
The set being counted is an initial segment of $\mathbb{N}$, so $t < B_q(c)$ exactly when $c+\phi(q)t<n$. The map $c \mapsto B_q(c)$ is non-increasing in $c$, so $c_q$---being the minimum of $C_q$---gives the longest visible occurrence:
\[
  B_q(c) \leq B_q(c_q) \mathrel{=:} b_q \quad \text{for every } c \in C_q.
\]
Note that $b_{q_0}=n$.

\textbf{Construction of the ARS.} For every edge $(q,j)$ of $P$ with $q \in Q$ and $j < d_q$, set $z \mathrel{:=} \delta(q)[j]$ and
\[
  c_{q,j} \mathrel{:=} c_q + \phi(q)j.
\]
The positions of the $j$-th child within the representative occurrence of $q$ are
\begin{equation}
  c_q + \phi(q)(d_q t+j) = c_{q,j} + d_q\phi(q)t, \quad t=0,1,2,\dots
  \label{eq:szlc-child-positions}
\end{equation}
If $c_{q,j} \geq n$---equivalently $j \geq b_q$---the edge contributes no constraint: its first position is already past $x$. Assume henceforth $c_{q,j}<n$.

\begin{itemize}
  \item \emph{Terminal child.} If $z=a \in \Sigma$, add the literal
  \[
    (a,c_{q,j})
  \]
  to $L$.

  \item \emph{Internal child.} If $z=q' \in Q$, then by \eqref{eq:szlc-sync}, $\phi(q')=d_q\phi(q)$, and the recursive definition of $C_{q'}$ places $c_{q,j}$ in it. Set $b_{q,j} \mathrel{:=} B_{q'}(c_{q,j})$, which is positive since $c_{q,j}<n$. Since $c_{q'}$ is the minimum of $C_{q'}$, we have $c_{q'} \leq c_{q,j}$, with two sub-cases:
  \begin{itemize}
    \item $c_{q'} < c_{q,j}$: add the equation
    \begin{equation}
      ((\phi(q'), c_{q'}, 0, b_{q,j}), (\phi(q'), c_{q,j}, 0, b_{q,j}))
      \label{eq:szlc-added-equation}
    \end{equation}
    to $\Phi$.
    \item $c_{q'} = c_{q,j}$: this child occurrence already \emph{is} the representative occurrence of $q'$; add nothing.
  \end{itemize}
\end{itemize}

\textbf{Validity.} For a literal $(a,c_{q,j})$ added by the construction, applying Lemma~\ref{lem:szlp-occ} to the representative occurrence of $q$ at index $j < b_q$ gives
\[
  x[c_{q,j}] = \operatorname{val}_P(q)[j] = a.
\]

For an equation \eqref{eq:szlc-added-equation}, both tuples are arithmetic intervals: the residues $c_{q'}, c_{q,j} \in [\phi(q')]$ satisfy the residue bound, the common step is $\phi(q')$, and the common length $b_{q,j}$ is positive. For every $t < b_{q,j}$, the definition of $b_{q,j}$ gives $c_{q,j}+\phi(q')t<n$, and from $c_{q'}<c_{q,j}$ also $c_{q'}+\phi(q')t<n$. Two applications of Lemma~\ref{lem:szlp-occ} to the occurrences of $q'$ at residues $c_{q,j}$ and $c_{q'}$ yield
\[
  x[c_{q'}+\phi(q')t] = \operatorname{val}_P(q')[t] = x[c_{q,j}+\phi(q')t],
\]
and the strict inequality $c_{q'}<c_{q,j}$ ensures that the source position is strictly less than the target throughout.

\textbf{Coverage.} We prove the following claim by reverse induction in the graph order of $P$: \emph{for every $q \in Q$ and every $s < b_q$, the position $c_q+\phi(q)s$ is covered by $R$---meaning either it appears as a literal in $L$, or it is the target of some equation in $\Phi$.}

Fix $q$ and $s < b_q$, and write $j \mathrel{:=} s \bmod d_q$, $t \mathrel{:=} \lfloor s/d_q\rfloor$, $z \mathrel{:=} \delta(q)[j]$. The position to cover is
\[
  p \mathrel{:=} c_q+\phi(q)s = c_{q,j}+d_q\phi(q)t,
\]
by \eqref{eq:szlc-child-positions}. Since $j \leq s < b_q$, we have $c_{q,j}<n$, so the edge $(q,j)$ contributed a constraint.

\begin{itemize}
  \item If $z=a \in \Sigma$: a terminal has length one, so $t=0$ and $s=j$. The literal $(a,c_{q,j})$ covers $p=c_{q,j}$.

  \item If $z=q' \in Q$: from \eqref{eq:szlc-sync}, $d_q\phi(q)=\phi(q')$, so $p=c_{q,j}+\phi(q')t$, and $p<n$ gives $t<b_{q,j}$.
  \begin{itemize}
    \item If $c_{q'} < c_{q,j}$: at index $t < b_{q,j}$, the equation \eqref{eq:szlc-added-equation} has source position $c_{q'}+\phi(q')t$ strictly less than target position $c_{q,j}+\phi(q')t=p$. Hence $p$ is covered.
    \item If $c_{q'} = c_{q,j}$: then $p=c_{q'}+\phi(q')t$ is a representative position of $q'$, and
    \[
      t < b_{q,j} = B_{q'}(c_{q'}) = b_{q'}.
    \]
    The induction hypothesis at $q'$ covers $p$.
  \end{itemize}
\end{itemize}

This proves the claim. Specializing at the root, $c_{q_0}=0$, $\phi(q_0)=1$, $b_{q_0}=n$, so every position of $x$ is covered. Together with validity, this shows that $R$ is an ARS for $x$.

\textbf{Size bound.} Each edge of $P$ contributes at most one element of $L \sqcup \Phi$, so
\[
  \operatorname{ARC}(x) \leq |R| = |L| + |\Phi| \leq \sum_{q\in Q}|\delta(q)| = |P|.
\]
Taking the minimum over synchronous ZLPs $P$ with $x \sqsubseteq \operatorname{val}(P)$ gives
\[
  \operatorname{ARC}(x) \leq \operatorname{sZLC}(x).
\]

\subsection{Proof of Proposition~\ref{prp:arc-vs-zlc}}
\label{sec:apx-arc-vs-zlc}

We first observe an elementary ``covering'' fact about arithmetic progressions which will be used to replace the single representative occurrence of the synchronous case.

\begin{lemma}
\label{lem:rep-progressions}
Let $n \geq 1$, and let $S$ be a finite set of pairwise disjoint infinite arithmetic progressions
\[
  \rho(k) = mk+i, \quad m \geq 2, \quad 0 \leq i < m.
\]
Then there is a subset $A \subseteq S$ with $|A| < \sqrt[3]{2n}+1$ such that for every $\rho \in S$ and every $k \in \mathbb{N}$ with $\rho(k)<n$, either $\rho \in A$, or there is $\eta \in A$ such that
\begin{equation}
  \eta(k) < \rho(k).
  \label{eq:rep-progressions-cover}
\end{equation}
Moreover, if $\rho \in S$, $\rho \notin A$, and $B_\rho \mathrel{:=} |\{k \in \mathbb{N} \mid \rho(k)<n\}| > 0$, then $[B_\rho]$ can be partitioned into at most two integer intervals $I_0,I_1$ such that for each non-empty $I_l$ there is some $\eta_l \in A$ with
\[
  \eta_l(k) < \rho(k) \quad \text{for all } k \in I_l.
\]
\end{lemma}

\begin{proof}
Call a subset $A \subseteq S$ \emph{covering below} $n$ if it satisfies \eqref{eq:rep-progressions-cover}. The full set $S$ is covering, hence there is an inclusion-minimal covering subset $A$. We bound its size.

If $A$ is empty there is nothing to prove. Write $r \mathrel{:=} |A|$, and for each $\rho \in A$ choose a witness $k_\rho$ to the minimality of $A$: after removing $\rho$, the covering property fails at some progression and some parameter $k_\rho$. We claim that $\rho$ is the unique minimizer of the values $\eta(k_\rho)$, $\eta \in A$, and that $\rho(k_\rho)<n$.

Indeed, let $\xi \in S$ be a progression for which the cover fails after removing $\rho$. If $\xi \in A$ and $\xi \neq \rho$, then $\xi$ would cover itself even after $\rho$ is removed, so either $\xi=\rho$ or $\xi \notin A$. Since $A$ itself covers $\xi$, the failure after removing $\rho$ implies that $\rho(k_\rho)<\xi(k_\rho)$ in the second case, and that no $\eta \in A$ with $\eta \neq \rho$ satisfies $\eta(k_\rho)<\xi(k_\rho)$. Hence no such $\eta$ is below $\rho$ at $k_\rho$. In the first case the same conclusion is immediate. Ties are also impossible: if $\eta(k_\rho)=\rho(k_\rho)$ for $\eta \neq \rho$, then the two progressions intersect, contrary to the disjointness assumption. Thus $\rho$ is the unique minimizer, and in both cases $\rho(k_\rho)<n$.

The witnesses $k_\rho$ are therefore distinct. Enumerate $A$ as $\rho_1,\dots,\rho_r$ in increasing order of their witnesses, and write
\[
  \rho_j(k)=m_jk+i_j, \quad k_j \mathrel{:=} k_{\rho_j}.
\]
Then $k_j \geq j-1$. For $j<r$, uniqueness of the minima at $k_j$ and $k_{j+1}$ gives
\[
  m_j k_j+i_j < m_{j+1}k_j+i_{j+1},
\]
and
\[
  m_{j+1}k_{j+1}+i_{j+1} < m_j k_{j+1}+i_j.
\]
Since $k_j < k_{j+1}$, these inequalities imply $m_j>m_{j+1}$ and $i_j<i_{j+1}$. In particular,
\[
  i_{j+1}-i_j > (m_j-m_{j+1})k_j \geq k_j \geq j-1.
\]
The left hand side is an integer, so $i_{j+1}-i_j \geq j$. Therefore
\[
  i_r \geq \sum_{j=1}^{r-1} j = \frac{r(r-1)}{2}.
\]

Since $\rho_r(k_r)<n$, $k_r \geq r-1$, and $i_r < m_r$, we have, for $r\geq 2$,
\[
  n > m_r k_r+i_r \geq (i_r+1)(r-1)+i_r = r i_r+r-1 \geq \frac{r^2(r-1)}{2}.
\]
Hence $r-1<\sqrt[3]{2n}$ and $r<\sqrt[3]{2n}+1$. The case $r=1$ is trivial, so the first assertion follows.

It remains to prove the two-representative refinement. Fix $\rho \in S$ with $\rho \notin A$ and $B_\rho>0$. For each $\eta \in A$, the set
\[
  I_\eta \mathrel{:=} \{k \in [B_\rho] \mid \eta(k)<\rho(k)\}
\]
is, inside the interval $[B_\rho]$, either an initial interval, a final interval, the whole interval, or the empty interval, since it is defined by a single strict linear inequality in $k$. The covering property says that the union of the $I_\eta$ over $\eta \in A$ is all of $[B_\rho]$.

If some $I_\eta$ is all of $[B_\rho]$, then one interval suffices. Otherwise, take among the initial intervals one with maximal right endpoint, and among the final intervals one with minimal left endpoint. The union of these two intervals must cover $[B_\rho]$: if it did not, the uncovered point would be covered neither by any initial interval nor by any final interval, contradicting the previous paragraph. Refining the overlap, if any, gives the desired partition into at most two intervals, assigning to each interval the corresponding representative.
\end{proof}

\begingroup
\renewcommand{\proofname}{Proof of Proposition~\ref{prp:arc-vs-zlc}.}
\begin{proof}
  Let $x \in \Sigma^n$, and let $P=(Q,q_0,\delta)$ be a ZLP with $x \sqsubseteq \operatorname{val}(P)$. Write
  \[
    \overline{Q} \mathrel{:=} Q \sqcup \Sigma, \quad d_q \mathrel{:=} |\delta(q)| \text{ for } q \in Q, \quad \ell(z) \mathrel{:=} |\operatorname{val}_P(z)| \text{ for } z \in \overline{Q}.
  \]
  The case $n=0$ is witnessed by the empty ARS, so assume $n>0$.

  \textbf{Occurrences.} A \emph{path occurrence} of an internal vertex is the arithmetic progression obtained by following a path from the root. Formally, define finite sets $\mathcal{O}_q$ of progressions recursively as follows:
  \begin{itemize}
    \item $\mathcal{O}_{q_0}$ contains the root occurrence $\rho_0(k) \mathrel{:=} k$;
    \item if $\rho(k)=rk+i$ lies in $\mathcal{O}_q$, $j<d_q$, and $\delta(q)[j]=q'\in Q$, then $\mathcal{O}_{q'}$ contains the child occurrence
    \begin{equation}
      \rho^{(j)}(k) \mathrel{:=} r d_q k+i+rj.
      \label{eq:zlc-child-occurrence}
    \end{equation}
  \end{itemize}
  Equivalently, each occurrence is a pair $(r,i)$ with $0 \leq i < r$, interpreted as the map $k \mapsto rk+i$. For every $q \neq q_0$, all occurrences in $\mathcal{O}_q$ have step at least $2$.

  For a fixed $q$, the occurrences in $\mathcal{O}_q$ are pairwise disjoint as subsets of $\mathbb{N}$. Indeed, take two distinct path occurrences $\rho_1,\rho_2 \in \mathcal{O}_q$, and choose paths from $q_0$ to $q$ producing them. If $q=q_0$, there is only the root occurrence, so assume $q \neq q_0$. The two paths cannot be such that one is a proper prefix of the other: otherwise, after the shorter path has reached $q$, the longer path would give a non-empty directed path from $q$ back to $q$, contradicting acyclicity of $P$.

  Hence the paths have a longest common prefix and then diverge. Let this common prefix end at an occurrence $\rho_p(k)=rk+i$ of some internal vertex $p$, and let the next edge indices be $j_1 \neq j_2$. The two corresponding child occurrences are
  \[
    \chi_b(k) \mathrel{:=} r d_p k+i+rj_b, \quad b \in \{1,2\}.
  \]
  These two progressions are disjoint: a common value would imply, for some $k_1,k_2 \in \mathbb{N}$,
  \[
    r d_p k_1+i+rj_1 = r d_p k_2+i+rj_2,
  \]
  and hence
  \[
    d_p(k_1-k_2)=j_2-j_1.
  \]
  But $0 \leq j_1,j_2 < d_p$ and $j_1 \neq j_2$, so the right-hand side is a non-zero integer of absolute value strictly smaller than $d_p$, impossible for a multiple of $d_p$.

  Finally, every occurrence obtained further down a path is contained in the occurrence from which it descends, since \eqref{eq:zlc-child-occurrence} gives $\rho^{(j)}(k)=\rho(d_qk+j)$. Therefore $\rho_1$ and $\rho_2$ lie in disjoint child occurrences after the first divergence, and are themselves disjoint.

  \textbf{Occurrence identity.} We will use the following value identity for occurrences: for every $q \in Q$, every $\rho(k)=rk+i$ in $\mathcal{O}_q$, and every $t \in \mathbb{N}$ with $\rho(t)<n$,
  \begin{equation}
    t<\ell(q) \quad \text{and} \quad x[\rho(t)] = \operatorname{val}_P(q)[t].
    \label{eq:zlc-occ-value}
  \end{equation}

  To prove this, choose a path that produces $\rho$, and argue by induction on its length. For the root occurrence of $q_0$, the claim is exactly $x \sqsubseteq \operatorname{val}_P(q_0)$. For the induction step, suppose $\rho$ is obtained from an occurrence $\rho_p(k)=r_pk+i_p$ of a parent $p$ through the edge $\delta(p)[j]=q$. Thus $\rho(t)=\rho_p(d_pt+j)$. If $\rho(t)<n$, the induction hypothesis applied to $\rho_p$ at $s \mathrel{:=} d_pt+j$ gives
  \[
    s < \ell(p) \quad \text{and} \quad x[\rho(t)] = \operatorname{val}_P(p)[s].
  \]
  Since $\operatorname{val}_P(p)$ is the zip of its $d_p$ children and $s=d_pt+j$ is a valid position of this zip, the definition of $\operatorname{zip}$ implies $t<\ell(q)$ and
  \[
    \operatorname{val}_P(p)[s]=\operatorname{val}_P(q)[t].
  \]
  Combining these identities proves \eqref{eq:zlc-occ-value}.

  \textbf{Representatives.} For $q=q_0$, set $A_q \mathrel{:=} \{\rho_0\}$. For every $q \neq q_0$, apply Lemma~\ref{lem:rep-progressions} to the finite pairwise disjoint family $\mathcal{O}_q$, and let $A_q \subseteq \mathcal{O}_q$ be the resulting representative family. Thus
  \begin{equation}
    |A_q| < \sqrt[3]{2n}+1 \quad \text{for every } q \in Q,
    \label{eq:zlc-rep-size}
  \end{equation}
  where the root case is included because $n>0$.

  \textbf{Construction of the ARS.} We construct $R=(L,\Phi)$. Consider an internal vertex $q \in Q$, a representative occurrence $\rho(k)=rk+i \in A_q$, and an edge index $j<d_q$. Put $z \mathrel{:=} \delta(q)[j]$ and define the child progression
  \begin{equation}
    \chi(k) \mathrel{:=} r d_q k+i+rj.
    \label{eq:zlc-child-chi}
  \end{equation}
  Let
  \[
    B_\chi \mathrel{:=} |\{k \in \mathbb{N} \mid \chi(k)<n\}|.
  \]
  If $B_\chi=0$, this edge occurrence contributes no constraint.

  Suppose first that $z=a \in \Sigma$ and $B_\chi>0$. Then, we add the literal $(a,\chi(0))$ to $L$.

  Now suppose that $z=q' \in Q$ and $B_\chi>0$. Then $\chi \in \mathcal{O}_{q'}$. If $\chi \in A_{q'}$, we add no constraint: this child occurrence is itself representative and will be handled independently. If $\chi \notin A_{q'}$, the second part of Lemma~\ref{lem:rep-progressions} gives a partition of $[B_\chi]$ into at most two non-empty integer intervals $I=[a_I,b_I)$, each assigned a representative $\eta_I(k)=r_Ik+i_I \in A_{q'}$ satisfying
  \begin{equation}
    \eta_I(k) < \chi(k) \quad \text{for all } k \in I.
    \label{eq:zlc-source-before-target}
  \end{equation}
  For each such interval, add to $\Phi$ the equation
  \begin{equation}
    ((r_I,i_I,a_I,b_I), (r d_q,i+rj,a_I,b_I)).
    \label{eq:zlc-added-equation}
  \end{equation}

  \textbf{Validity.} We verify that all constraints are true in $x$. For a literal added from a terminal edge $\delta(q)[j]=a$, we have $\chi(0)=\rho(j)<n$. By \eqref{eq:zlc-occ-value} applied to $\rho$ at $j$,
  \[
    x[\chi(0)] = \operatorname{val}_P(q)[j].
  \]
  Since the $j$-th child is the one-letter word $a$, the zip identity gives $\operatorname{val}_P(q)[j]=a$.

  For an equation \eqref{eq:zlc-added-equation}, fix $k \in [a_I,b_I)$. The target position is $\chi(k)<n$ by the definition of $B_\chi$, and the source position is earlier by \eqref{eq:zlc-source-before-target}, hence it is also $<n$. Applying \eqref{eq:zlc-occ-value} to the two occurrences $\eta_I$ and $\chi$ of the same internal vertex $q'$ gives
  \[
    x[\eta_I(k)] = \operatorname{val}_P(q')[k] = x[\chi(k)].
  \]
  Thus the two arithmetic intervals read the same word in $x$, and the source position is strictly smaller than the target position at every parameter in the interval.

  \textbf{Coverage.} We prove, by reverse induction in any topological ordering of the acyclic dependency graph of $P$, the following claim, where ``covered'' means either appearing as a literal in $L$ or appearing as the target of an equation in $\Phi$ with a strictly earlier source:
  \begin{equation}
    \text{for every } q \in Q,\ \rho \in A_q,\ s \in \mathbb{N}, \quad \text{if } \rho(s)<n, \text{ then } \rho(s) \text{ is covered by } R.
    \label{eq:zlc-coverage-claim}
  \end{equation}

  Fix $q$, assume the claim already holds for all children of $q$ in $Q$, and let $\rho \in A_q$ with $\rho(k)=rk+i$ and $s \in \mathbb{N}$ with $\rho(s)<n$. Write
  \[
    j \mathrel{:=} s \bmod d_q, \quad t \mathrel{:=} \left\lfloor \frac{s}{d_q}\right\rfloor, \quad z \mathrel{:=} \delta(q)[j],
  \]
  and let $\chi$ be the child progression \eqref{eq:zlc-child-chi}. Then
  \[
    \rho(s)=\chi(t)
  \]
  and, in particular, $t<B_\chi$.

  If $z=a \in \Sigma$, then \eqref{eq:zlc-occ-value} applied to $\rho$ at $s$ gives $s<\ell(q)$. Since the $j$-th child has length $1$, the zip position $s=d_qt+j$ can be valid only when $t=0$. Therefore $\rho(s)=\chi(0)$, and the literal $(a,\chi(0))$ added above covers this position.

  If $z=q' \in Q$, then $\chi \in \mathcal{O}_{q'}$. There are two cases. If $\chi \in A_{q'}$, then $\chi(t)=\rho(s)$ is covered by the induction hypothesis for $q'$. If $\chi \notin A_{q'}$, the construction added equations for a partition of $[B_\chi]$; since $t<B_\chi$, it lies in one of the intervals. The corresponding equation has target $\chi(t)=\rho(s)$ and source strictly earlier, so it covers $\rho(s)$.

  This proves \eqref{eq:zlc-coverage-claim}. Applying the claim to the root occurrence $\rho_0 \in A_{q_0}$, whose visible positions are exactly $0,\dots,n-1$, shows that every position of $x$ is covered. Together with validity, $R$ is an ARS for $x$.

  \textbf{Size bound.} For every representative occurrence of every internal vertex, each outgoing edge contributes at most two constraints: one literal in the terminal case, or at most two equations in the internal case. Therefore, using \eqref{eq:zlc-rep-size},
  \[
    |R| \leq 2 \sum_{q\in Q} |A_q| d_q \leq 2\left(\sqrt[3]{2n}+1\right)\sum_{q\in Q} d_q = 2\left(\sqrt[3]{2n}+1\right)\cdot |P|.
  \]
  Hence $\operatorname{ARC}(x) \leq 2(\sqrt[3]{2n}+1)\cdot |P|$. Taking the minimum over all ZLPs $P$ with $x \sqsubseteq \operatorname{val}(P)$ gives
  \[
    \operatorname{ARC}(x) \leq 2\left(\sqrt[3]{2n}+1\right)\cdot \operatorname{ZLC}(x).
  \]
\end{proof}
\endgroup

\subsection{Preliminaries for Proposition~\ref{prp:zlc-vs-acr}}
\label{sec:apx-szlc-vs-acr-prelim}

Throughout this subsection, given a DFAO $N=(Q,q_0,\delta,\tau)$ over input alphabet $[k]$, write $\hat{\delta}$ for the standard extension of $\delta$ to a map $Q \times [k]^* \to Q$. For an integer $j \geq 0$, define the \emph{level-$j$ reading map} $\psi^N_j:[k^j]\to Q$ by
\[
  \psi^N_j(r) \mathrel{:=} \hat{\delta}(q_0,d_0d_1\dots d_{j-1}), \quad d_i \mathrel{:=} \left\lfloor \frac{r}{k^i}\right\rfloor \bmod k.
\]
Thus $\psi^N_j(r)$ is the state of $N$ after reading the $j$ least-significant base-$k$ digits of $r$ in LSDF order. The chain rule
\[
  \psi^N_{j'}(r) = \hat{\delta}(\psi^N_j(r \bmod k^j), d_jd_{j+1}\dots d_{j'-1})
\]
holds for $0 \leq j \leq j'$ and $r \in [k^{j'}]$. When the underlying DFAO is clear from context, we drop the superscript and write $\psi_j$.

\begin{lemma}
\label{lem:can-depth-rtl}
Let $k\geq 2$, let $x \in \Sigma^n$, and assume $n\geq 1$. There is a DFAO $M=(Q,q_0,\delta,\tau)$ with $\mathrm{AC}^{\mathrm{R}}_k(x)$ states such that, for $m\mathrel{:=}\lceil\log_k n\rceil$,
\[
  x[t] = \tau(\psi^M_m(t)) \quad \text{for every } t < n.
\]
\end{lemma}

\begin{proof}
Let $M'=(Q,q_0,\delta,\tau)$ be a DFAO that witnesses $\mathrm{AC}^{\mathrm{R}}_k(x)$ at some depth $m_0$, so $n \leq k^{m_0}$. Necessarily $m_0 \geq m$. If $m_0=m$, there is nothing to prove. If $m_0>m$, then for every $t<n$,
\[
  (\langle t\rangle^{m_0}_k)^{\mathrm{R}} = (\langle t\rangle^{m}_k)^{\mathrm{R}} 0^{m_0-m}.
\]
Define a new output map
\[
  \tau'(q) \mathrel{:=} \tau(\hat{\delta}(q,0^{m_0-m})).
\]
Replacing $\tau$ by $\tau'$ gives a DFAO with the same state set that witnesses $x$ at depth $m$.
\end{proof}

\subsection{Proof of Proposition~\ref{prp:zlc-vs-acr}}
\label{sec:apx-szlc-vs-acr}

Let $s \mathrel{:=} \mathrm{AC}^{\mathrm{R}}_k(x)$ and set $m \mathrel{:=} \lceil\log_k n\rceil$. By Lemma~\ref{lem:can-depth-rtl} there is a DFAO
\[
  M^* = (Q^*,q_0^*,\delta^*,\tau^*)
\]
with $|Q^*|=s$ such that
\[
  x[t] = \tau^*(\psi_m(t)) \quad \text{for every } t<n,
\]
where $\psi \mathrel{:=} \psi^{M^*}$.

For $j \in \{0,\dots,m\}$, let
\[
  \mathcal{S}_j \mathrel{:=} \{\psi_j(r) : r \in [k^j]\}.
\]
Thus $|\mathcal{S}_j| \leq s$ for every $j$.

\textbf{Construction.} We construct a ZLP $P=(Q_P,q_P,\delta_P)$. For each $j<m$ and $q \in \mathcal{S}_j$, introduce a fresh internal vertex $v_{j,q}$, and set
\[
  Q_P \mathrel{:=} \{v_{j,q} : 0 \leq j < m \text{ and } q \in \mathcal{S}_j\}, \quad q_P \mathrel{:=} v_{0,q_0^*}.
\]
The outgoing list of every internal vertex has arity $k$. If $j=m-1$, set
\[
  \delta_P(v_{m-1,q}) \mathrel{:=} (\tau^*(\delta^*(q,0)), \tau^*(\delta^*(q,1)), \dots, \tau^*(\delta^*(q,k-1))).
\]
If $j<m-1$, set
\[
  \delta_P(v_{j,q}) \mathrel{:=} (v_{j+1,\delta^*(q,0)}, v_{j+1,\delta^*(q,1)}, \dots, v_{j+1,\delta^*(q,k-1)}).
\]
This is well-defined because whenever $q=\psi_j(r)$ and $d \in [k]$, we have
\[
  \delta^*(q,d)=\psi_{j+1}(r+d k^j) \in \mathcal{S}_{j+1}.
\]
Every internal edge goes from level $j$ to level $j+1$, so the dependency graph is acyclic. Moreover, every vertex at level $j>0$ appears in the outgoing list of some vertex at level $j-1$, by decomposing a witness $r \in [k^j]$ as $r=r' + d k^{j-1}$. Therefore $q_P$ is the unique internal vertex that does not appear in any outgoing list, and $P$ is a ZLP.

\textbf{Validity.} We now verify the value of $P$. For $\ell \geq 0$ and $h \in [k^\ell]$, write $\operatorname{Dig}_\ell(h)$ for the length-$\ell$ LSDF base-$k$ digit word of $h$:
\[
  \operatorname{Dig}_\ell(h) \mathrel{:=} e_0 e_1 \dots e_{\ell-1}, \quad e_i \mathrel{:=} \left\lfloor \frac{h}{k^i}\right\rfloor \bmod k.
\]
For $\ell=0$ this is the empty word. For $0 \leq j \leq m$ and $q \in \mathcal{S}_j$, define the word $w_{j,q} \in \Sigma^{k^{m-j}}$ by
\[
  w_{j,q}[h] \mathrel{:=} \tau^*(\widehat{\delta^*}(q,\operatorname{Dig}_{m-j}(h)))
\]
for $h \in [k^{m-j}]$. Thus $w_{m,q}$ is the one-symbol word $\tau^*(q)$.

We claim that, for every $j<m$ and $q \in \mathcal{S}_j$,
\begin{equation}
  \operatorname{val}_P(v_{j,q}) = w_{j,q}.
  \label{eq:apx-szlc-value-claim}
\end{equation}
This is proved by backward induction on $j$.

For the base case $j=m-1$, fix $q \in \mathcal{S}_{m-1}$ and put
\[
  a_d \mathrel{:=} \tau^*(\delta^*(q,d))
\]
for $d \in [k]$. Then
\[
  \operatorname{val}_P(v_{m-1,q}) = \operatorname{zip}(a_0,a_1,\dots,a_{k-1}).
\]
Each $a_d$ is a one-symbol word, so this zip has length $k$ and its $h$-th symbol is $a_h$ for $h \in [k]$. Since $\operatorname{Dig}_1(h)$ is the one-letter word $h$,
\[
  \begin{aligned}
  w_{m-1,q}[h]
    &= \tau^*(\widehat{\delta^*}(q,\operatorname{Dig}_1(h))) \\
    &= \tau^*(\delta^*(q,h)) \\
    &= a_h.
  \end{aligned}
\]
Thus $\operatorname{val}_P(v_{m-1,q})=w_{m-1,q}$.

For the induction step, assume the claim has been proved at level $j+1$, where $j<m-1$, and fix $q \in \mathcal{S}_j$. Set
\[
  u_d \mathrel{:=} \operatorname{val}_P(v_{j+1,\delta^*(q,d)})
\]
for $d \in [k]$. By the induction hypothesis,
\[
  u_d = w_{j+1,\delta^*(q,d)},
\]
so all words $u_d$ have the common length $k^{m-j-1}$. Hence
\[
  \operatorname{val}_P(v_{j,q}) = \operatorname{zip}(u_0,u_1,\dots,u_{k-1})
\]
has length $k^{m-j}$. Let $h \in [k^{m-j}]$, and write
\[
  d \mathrel{:=} h \bmod k, \quad h' \mathrel{:=} \left\lfloor\frac{h}{k}\right\rfloor.
\]
Then $h' \in [k^{m-j-1}]$, and the definition of $\operatorname{zip}$ gives
\[
  \begin{aligned}
  \operatorname{val}_P(v_{j,q})[h]
    &= u_d[h'] \\
    &= w_{j+1,\delta^*(q,d)}[h'] \\
    &= \tau^*(\widehat{\delta^*}(\delta^*(q,d),\operatorname{Dig}_{m-j-1}(h'))) \\
    &= \tau^*(\widehat{\delta^*}(q,d\operatorname{Dig}_{m-j-1}(h'))).
  \end{aligned}
\]
Since $h=d+kh'$, the length-$(m-j)$ LSDF digit word of $h$ is obtained by prefixing the digit $d$ to the length-$(m-j-1)$ LSDF digit word of $h'$, i.e.
\[
  \operatorname{Dig}_{m-j}(h) = d\operatorname{Dig}_{m-j-1}(h').
\]
Therefore
\[
  \operatorname{val}_P(v_{j,q})[h] = \tau^*(\widehat{\delta^*}(q,\operatorname{Dig}_{m-j}(h))) = w_{j,q}[h].
\]
Since this holds for every $h \in [k^{m-j}]$, we get $\operatorname{val}_P(v_{j,q})=w_{j,q}$, completing the induction.

Taking $j=0$ and $q=q_0^*$ in \eqref{eq:apx-szlc-value-claim}, we get
\[
  \operatorname{val}(P) = \operatorname{val}_P(q_P) = w_{0,q_0^*}.
\]
For every $t<n$, the word $\operatorname{Dig}_m(t)$ is precisely the length-$m$ LSDF base-$k$ representation of $t$. Hence
\[
  \begin{aligned}
  \operatorname{val}(P)[t]
    &= w_{0,q_0^*}[t] \\
    &= \tau^*(\widehat{\delta^*}(q_0^*,\operatorname{Dig}_m(t))) \\
    &= \tau^*(\psi_m(t)) \\
    &= x[t].
  \end{aligned}
\]
Since $n \leq k^m = |\operatorname{val}(P)|$, this shows that $x \sqsubseteq \operatorname{val}(P)$.

\textbf{Synchronicity.} The ZLP $P$ is synchronous: define $\phi_P(v_{j,q}) \mathrel{:=} k^j$. If $v_{j+1,q'}$ appears in $\delta_P(v_{j,q})$, then
\[
  \phi_P(v_{j+1,q'}) = k^{j+1} = k \cdot k^j = |\delta_P(v_{j,q})| \cdot \phi_P(v_{j,q}).
\]
Also $\phi_P(q_P)=1$.

\textbf{Size bound.} Every internal vertex has exactly $k$ outgoing edges, so
\[
  \begin{aligned}
  |P|
    &= \sum_{j=0}^{m-1}\sum_{q\in\mathcal{S}_j}|\delta_P(v_{j,q})| \\
    &= k\sum_{j=0}^{m-1}|\mathcal{S}_j| \\
    &\leq ksm.
  \end{aligned}
\]
Recalling that $s=\mathrm{AC}^{\mathrm{R}}_k(x)$ and $m=\lceil\log_k n\rceil$, we obtain
\[
  \operatorname{sZLC}(x) \leq |P| \leq ksm = k\lceil\log_k n\rceil \cdot \mathrm{AC}^{\mathrm{R}}_k(x).
\]

\section{Mix-Automatic Sequences}
\label{sec:apx-mix-automatic}

In this appendix, we will explain the definition of mix-automatic sequences from \citet{Endrullis2013}, and show that they exhibit polylogarithmic $\operatorname{sZLC}$ growth. We start by recalling the definition of $k$-automatic sequences (see \citet{Allouche2003}) which were already mentioned in Appendix~\ref{sec:apx-auto}. We will then see how mix-automatic sequences are a natural generalization of $k$-automatic sequences that doesn't depend on $k$.

Fix $k \geq 2$. For $m \in \mathbb{N}$ and $t < k^m$, write $\operatorname{Dig}_k^m(t) \in [k]^m$ for the length-$m$ LSDF base-$k$ representation of $t$:
\[
  \operatorname{Dig}_k^m(t) = d_0 d_1 \dots d_{m-1}, \quad d_j \mathrel{:=} \left\lfloor \frac{t}{k^j} \right\rfloor \bmod k.
\]
Thus $\operatorname{Dig}_k^m(t)$ is $(\langle t \rangle_k^m)^{\mathrm{R}}$ in the notation of Section~\ref{sec:auto}. The canonical LSDF representation is denoted $\operatorname{Dig}_k(t)$: we set $\operatorname{Dig}_k(t)\mathrel{:=}\operatorname{Dig}_k^\ell(t)$, where $\ell\in \mathbb{N}$ is minimal with $t<k^\ell$.

A \emph{$k$-DFAO} over $\Sigma$ is a finite automaton with output $N=(Q,q_0,\delta,\tau)$, where $Q$ is finite, $q_0 \in Q$, $\delta:Q \times [k] \to Q$, and $\tau:Q \to \Sigma$. We extend $\delta$ to $\widehat{\delta}:Q \times [k]^* \to Q$ by
\[
  \widehat{\delta}(q,\epsilon) \mathrel{:=} q, \quad
  \widehat{\delta}(q,ui) \mathrel{:=} \delta(\widehat{\delta}(q,u),i).
\]

\begin{definition}
\label{def:k-auto-lsdf}
  A sequence $x \in \Sigma^\omega$ is \emph{$k$-automatic} if there exists a $k$-DFAO $N=(Q,q_0,\delta,\tau)$ over $\Sigma$ such that, for every $t \in \mathbb{N}$,
  \[
    x[t] = \tau(\widehat{\delta}(q_0,\operatorname{Dig}_k(t))).
  \]
\end{definition}

This is equivalent to the way we described $k$-automatic sequences in Appendix~\ref{sec:apx-auto}, due to the following.

\begin{proposition}
\label{prp:acr-bdd-iff-k-auto}
  For fixed $k\geq 2$ and $x \in \Sigma^\omega$, the following are equivalent:
  \begin{itemize}
    \item $x$ is $k$-automatic in the sense of Definition~\ref{def:k-auto-lsdf};
    \item $\mathrm{AC}^{\mathrm{R}}_k(x[:n])=O(1)$ as $n \to \infty$.
  \end{itemize}
\end{proposition}

\begin{proof}
  First suppose that $x$ is generated by $N=(Q,q_0,\delta,\tau)$ as in Definition~\ref{def:k-auto-lsdf}. We transform $N$ into a DFAO $N^+$ that is insensitive to the extra zeros introduced by the padding in Definition~\ref{def:auto}. Let
  \[
    Q^+ \mathrel{:=} \Sigma \times Q, \quad q^+_0 \mathrel{:=} (\tau(q_0),q_0),
  \]
  let $\tau^+(a,q)\mathrel{:=}a$, and define transitions by
  \[
    \delta^+((a,q),0) \mathrel{:=} (a,\delta(q,0)), \quad
    \delta^+((a,q),i) \mathrel{:=} (\tau(\delta(q,i)),\delta(q,i)) \quad (0<i<k).
  \]
  The second coordinate simulates $N$. The first coordinate stores the output that $N$ would have after the last nonzero digit read so far; if no nonzero digit has been read, it stores $\tau(q_0)$. Therefore, for every $t \in \mathbb{N}$ and every $r \in \mathbb{N}$,
  \begin{equation}
    \tau^+(\widehat{\delta^+}(q^+_0,\operatorname{Dig}_k(t)0^r))
    = \tau(\widehat{\delta}(q_0,\operatorname{Dig}_k(t)))
    = x[t].
    \label{eq:k-pad-insensitive}
  \end{equation}
  Indeed, if $t=0$ then no nonzero digit is read and both sides are $\tau(q_0)$; if $t>0$, the final digit of $\operatorname{Dig}_k(t)$ is nonzero, so at that digit the first coordinate is refreshed after all lower-significance digits have already been processed, and all subsequent zeros leave it unchanged.

  Now fix $n\geq 1$ and choose $m$ with $n \leq k^m$. For every $t<n$, the padded LSDF word $\operatorname{Dig}_k^m(t)$ is $\operatorname{Dig}_k(t)0^{m-|\operatorname{Dig}_k(t)|}$. Hence \eqref{eq:k-pad-insensitive} shows that $N^+$ witnesses
  \[
    \mathrm{AC}^{\mathrm{R}}_k(x[:n]) \leq |Q^+| = |\Sigma| |Q|.
  \]
  In particular, $\mathrm{AC}^{\mathrm{R}}_k(x[:n])=O(1)$.

  Conversely, assume that $\mathrm{AC}^{\mathrm{R}}_k(x[:n]) \leq C$ for all $n\geq 1$. For each $n$, choose a DFAO $N_n$ with at most $C$ states witnessing $\mathrm{AC}^{\mathrm{R}}_k(x[:n])$ at depth $m_n$ (i.e. via padding to length $m_n$). Up to isomorphism, there are only finitely many such DFAOs. Hence, there is an unbounded sequence $n_j \to \infty$, a single DFAO $N=(Q,q_0,\delta,\tau)$, and depths $m_j \to \infty$ such that $N$ witnesses $x[:n_j]$ at depth $m_j$ for every $j$:
  \begin{equation}
    x[t] = \tau(\widehat{\delta}(q_0,\operatorname{Dig}_k^{m_j}(t))) \quad \text{for all } t < n_j.
    \label{eq:bounded-acr-subsequence}
  \end{equation}

  Let $z:Q \to Q$ be the zero-transition map $z(q)\mathrel{:=}\delta(q,0)$. Since $Q$ is finite, there exist $B \in \mathbb{N}$ and $p \in \mathbb{N}^+$ such that
  \[
    z^{r+p}(q) = z^r(q) \quad \text{for all } q \in Q \text{ and all } r\geq B.
  \]
  Passing to a further subsequence, we may assume w.l.o.g. that all $m_j$ have the same residue $c \in [p]$ modulo $p$.

  We now build a single canonical-input DFAO $N^*$. Its states are $Q \times [p]$, its initial state is $(q_0,0)$, and its transition function is
  \[
    \delta^*((q,a),i) \mathrel{:=} (\delta(q,i), a+1 \bmod p).
  \]
  For each $a \in [p]$, choose $e_a\geq B$ such that $e_a \equiv c-a \pmod p$, and define the output map by
  \[
    \tau^*(q,a) \mathrel{:=} \tau(z^{e_a}(q)).
  \]
  Let $t \in \mathbb{N}$, put $\ell\mathrel{:=}|\operatorname{Dig}_k(t)|$, and let $q\mathrel{:=}\widehat{\delta}(q_0,\operatorname{Dig}_k(t))$. Choose $j$ so large that $t<n_j$ and $m_j-\ell\geq B$. From \eqref{eq:bounded-acr-subsequence},
  \[
  \begin{aligned}
    x[t]
      &= \tau(\widehat{\delta}(q_0,\operatorname{Dig}_k(t)0^{m_j-\ell})) \\
      &= \tau(z^{m_j-\ell}(q)).
  \end{aligned}
  \]
  Since $m_j \equiv c \pmod p$, we have $m_j-\ell \equiv c-\ell \pmod p$. By the eventual periodicity of $z$,
  \[
    \tau(z^{m_j-\ell}(q)) = \tau(z^{e_{\ell \bmod p}}(q)).
  \]
  On the other hand, after reading the canonical word $\operatorname{Dig}_k(t)$, the automaton $N^*$ is in state $(q,\ell \bmod p)$, whose output is exactly $\tau(z^{e_{\ell \bmod p}}(q))$. Thus $N^*$ computes $x[t]$. Since $t$ was arbitrary, $x$ is $k$-automatic.
\end{proof}

In order to define a generalization that doesn't depend on $k$, we make the digit alphabet state-dependent. A $k$-DFAO has exactly $k$ successors at every state; in a mix-automaton, the number of successors may depend on the current state, and therefore the base of the next digit is determined by the lower-significance digits that have already been read.

\begin{definition}
\label{def:mix-automaton}
  A \emph{mix-automaton} over $\Sigma$ is a tuple $M=(Q,q_0,\delta,\tau)$, where $Q$ is finite, $q_0 \in Q$, $\delta:Q \to Q^*$, and $\tau:Q \to \Sigma$. For $q \in Q$, write
  \[
    d_q \mathrel{:=} |\delta(q)|.
  \]
  We require $d_q \geq 2$ for every $q \in Q$. The \emph{directed graph of $M$} is defined to have vertex set $Q$ and an edge $q \to q'$ whenever $q'$ appears in $\delta(q)$.
\end{definition}

At state $q$, the available digits are $[d_q]$, and reading digit $i<d_q$ moves the automaton to $\delta(q)[i]$. A $k$-DFAO is the special case in which $d_q=k$ for all $q$.

For a mix-automaton $M=(Q,q_0,\delta,\tau)$, define
\[
  \operatorname{val}_M: Q \to \Sigma^\omega
\]
by recursion on the position inside the sequence. For any $q \in Q$, set
\[
  \operatorname{val}_M(q)[0] \mathrel{:=} \tau(q),
\]
and, whenever $d_q t+i>0$ with $i<d_q$, set
\begin{equation}
  \operatorname{val}_M(q)[d_q t+i] \mathrel{:=} \operatorname{val}_M(\delta(q)[i])[t].
  \label{eq:mix-ev-recursion}
\end{equation}
The recursion is well-founded because the right-hand side is evaluated at the smaller index $t<d_q t+i$. We define the \emph{value} of $M$ as $\operatorname{val}(M)\mathrel{:=}\operatorname{val}_M(q_0)$. A sequence $x \in \Sigma^\omega$ is \emph{mix-automatic} if $x=\operatorname{val}(M)$ for some mix-automaton $M$.

It will be convenient to use a zero-normalized presentation.

\begin{definition}
\label{def:mix-zero-normalized}
  A mix-automaton $M=(Q,q_0,\delta,\tau)$ is \emph{zero-normalized} if
  \[
    \tau(\delta(q)[0]) = \tau(q) \quad \text{for every } q \in Q.
  \]
\end{definition}

\begin{lemma}
\label{lem:mix-zero-normalization}
  Every mix-automatic sequence is generated by a zero-normalized mix-automaton.
\end{lemma}

\begin{proof}
  Let $M=(Q,q_0,\delta,\tau)$ be any mix-automaton. Define
  \[
    M^+ \mathrel{:=} (Q^+,q^+_0,\delta^+,\tau^+)
  \]
  by
  \[
    Q^+ \mathrel{:=} \Sigma \times Q, \quad q^+_0 \mathrel{:=} (\tau(q_0),q_0), \quad \tau^+(a,q)\mathrel{:=}a.
  \]
  The arity of $(a,q)$ is $d_q$, and its successor list is defined by
  \[
    \delta^+(a,q)[i] \mathrel{:=}
    \begin{cases}
      (a,\delta(q)[0]) & \text{if } i=0,\\
      (\tau(\delta(q)[i]),\delta(q)[i]) & \text{if } 0<i<d_q.
    \end{cases}
  \]
  Thus $M^+$ is again a mix-automaton. It is zero-normalized because
  \[
    \tau^+(\delta^+(a,q)[0]) = \tau^+(a,\delta(q)[0]) = a = \tau^+(a,q).
  \]

  We claim that, for every $(a,q) \in Q^+$,
  \begin{equation}
    \operatorname{val}_{M^+}(a,q)[0] = a, \quad \text{and} \quad \operatorname{val}_{M^+}(a,q)[n] = \operatorname{val}_M(q)[n] \text{ for all } n>0.
    \label{eq:normalization-claim}
  \end{equation}
  The first equality follows directly from the definition of $\operatorname{val}_{M^+}$. We prove the second by induction on $n>0$.

  Write $n=d_q t+i$ with $i<d_q$, and set
  \[
    q_i \mathrel{:=} \delta(q)[i], \quad
    a_i \mathrel{:=}
    \begin{cases}
      a & \text{if } i=0,\\
      \tau(q_i) & \text{if } i>0,
    \end{cases}
  \]
  so that $\delta^+(a,q)[i]=(a_i,q_i)$.

  If $t=0$, then $i>0$, and hence
  \[
  \begin{aligned}
    \operatorname{val}_{M^+}(a,q)[n]
      &= \operatorname{val}_{M^+}(a_i,q_i)[0]\\
      &= a_i\\
      &= \tau(q_i)\\
      &= \operatorname{val}_M(q_i)[0]\\
      &= \operatorname{val}_M(q)[n].
  \end{aligned}
  \]
  If $t>0$, then the induction hypothesis applied to the child state $(a_i,q_i)$ gives
  \[
  \begin{aligned}
    \operatorname{val}_{M^+}(a,q)[n]
      &= \operatorname{val}_{M^+}(a_i,q_i)[t]\\
      &= \operatorname{val}_M(q_i)[t]\\
      &= \operatorname{val}_M(q)[n].
  \end{aligned}
  \]
  This proves \eqref{eq:normalization-claim}. Taking $(a,q)=(\tau(q_0),q_0)$ gives $\operatorname{val}(M^+)=\operatorname{val}(M)$.
\end{proof}

For the rest of the appendix, $M$ denotes a zero-normalized mix-automaton. For infinite sequences $x_0,\dots,x_{r-1} \in \Sigma^\omega$, we use the same zip notation as for finite words:
\[
  \operatorname{zip}(x_0,\dots,x_{r-1})[rt+i] \mathrel{:=} x_i[t] \quad (i<r).
\]

By \eqref{eq:mix-ev-recursion} and Definition~\ref{def:mix-zero-normalized}, every state $q$ satisfies the exact zip identity
\begin{equation}
  \operatorname{val}_M(q) = \operatorname{zip}(\operatorname{val}_M(\delta(q)[0]),\dots,\operatorname{val}_M(\delta(q)[d_q-1])).
  \label{eq:mix-zip-identity}
\end{equation}
Indeed, for an index $d_q t+i$ with $(t,i)\neq(0,0)$ this is precisely \eqref{eq:mix-ev-recursion}, while at index $0$ it is the zero-normalization condition.

We also need notation for walks in $M$. Let $\operatorname{dom}_M \subseteq \mathbb{N}^*$ be the set of words that label walks from $q_0$ in $M$, and let $\operatorname{run}_M:\operatorname{dom}_M \to Q$ be the corresponding run map. Thus $\epsilon \in \operatorname{dom}_M$, $\operatorname{run}_M(\epsilon)=q_0$, and whenever $u \in \operatorname{dom}_M$ and $i<d_{\operatorname{run}_M(u)}$, we have $ui \in \operatorname{dom}_M$ and
\[
  \operatorname{run}_M(ui) = \delta(\operatorname{run}_M(u))[i].
\]
For $u \in \operatorname{dom}_M$, define
\begin{equation}
  D_M(u) \mathrel{:=} \prod_{h<|u|} d_{\operatorname{run}_M(u[:h])}.
  \label{eq:mix-path-scale}
\end{equation}
This is the product of the bases encountered along the path $u$.

\begin{definition}
\label{def:syrk}
  Let $M$ be a mix-automaton. For a directed cycle $C \in Q^*$ in the graph of $M$, set
  \[
    w(C) \mathrel{:=} \sum_{j<|C|} \ln d_{C[j]}.
  \]
  The \emph{synchronization rank} of $M$ is
  \[
    \operatorname{syrk}(M) \mathrel{:=} \dim_{\mathbb{Q}} \operatorname{span}_{\mathbb{Q}} \{ w(C) : C \text{ is a directed cycle in } M \}.
  \]
\end{definition}

The synchronization rank is relevant to us because of the following property.

\begin{lemma}
\label{lem:mix-product-count}
  Let $M$ be a mix-automaton, and put $r\mathrel{:=}\operatorname{syrk}(M)$. Then
  \[
    \left|\{D_M(u) : u \in \operatorname{dom}_M \text{ and } D_M(u)<n\}\right| = O((\log n)^r).
  \]
  The constant hidden in the $O$-notation depends on $M$ but not on $n$.
\end{lemma}

\begin{proof}
  Write $a_q \mathrel{:=} \ln d_q$. For $u \in \operatorname{dom}_M$, set $q_h \mathrel{:=} \operatorname{run}_M(u[:h])$ for $0 \leq h \leq |u|$. Then
  \[
    \ln D_M(u) = \sum_{h<|u|} a_{q_h}.
  \]

  Let $\mathcal{C}$ be the finite set of directed cycles in the directed graph of $M$. We first use the decomposition of a walk into a path and some number of cycles. Namely, for every $u \in \operatorname{dom}_M$, there exist a directed path $\pi$ starting at $q_0$ and integers $m_C \in \mathbb{N}$, indexed by $C \in \mathcal{C}$, such that
  \begin{equation}
    \ln D_M(u) = b_\pi + \sum_{C \in \mathcal{C}} m_C w(C),
    \label{eq:mix-walk-cycle-decomp}
  \end{equation}
  where, if $\pi=p_0 \dots p_\ell$, then $b_\pi \mathrel{:=} \sum_{j<\ell} a_{p_j}$.

  Indeed, consider the state walk $q_0 q_1 \dots q_{|u|}$. If it has a repeated vertex, let $q_j$ be the first vertex, in the order of the walk, that has appeared before, and let $i<j$ be its previous occurrence. Then the vertices $q_i,q_{i+1},\dots,q_{j-1}$ are distinct, and together with the edge from $q_{j-1}$ to $q_j=q_i$ they form a directed cycle $C$. Removing the segment $q_i \dots q_{j-1}$ preserves the property of being a walk, decreases the length, and subtracts exactly $w(C)=\sum_{h=i}^{j-1} a_{q_h}$ from the weight sum. Iterating this operation terminates with a directed path, proving \eqref{eq:mix-walk-cycle-decomp}.

  It is enough to prove the bound for $n\geq 2$. Put $L\mathrel{:=}\ln n$ and $\lambda\mathrel{:=}\min_{C \in \mathcal{C}} w(C)$. Since every $d_q \geq 2$ and $\mathcal{C}$ is non-empty (our graph has no sinks), we have $\lambda>0$. If $D_M(u)<n$, then \eqref{eq:mix-walk-cycle-decomp} gives
  \[
    \sum_{C \in \mathcal{C}} m_C w(C) \leq \ln D_M(u) < L,
  \]
  and therefore
  \begin{equation}
    \sum_{C \in \mathcal{C}} m_C \leq \frac{L}{\lambda}.
    \label{eq:mix-cycle-mult-bound}
  \end{equation}

  Choose directed cycles $C_0,\dots,C_{r-1} \in \mathcal{C}$ such that $\eta_i\mathrel{:=}w(C_i)$ is a $\mathbb{Q}$-basis of the $\mathbb{Q}$-span of the cycle weights. Since $\mathcal{C}$ is finite, there is an integer $A\geq 1$ and integers $z_{C,i}$ such that, for every $C \in \mathcal{C}$,
  \begin{equation*}
    A \cdot w(C) = \sum_{i<r} z_{C,i} \eta_i.
  \end{equation*}
  Let $K\mathrel{:=}\max_{C \in \mathcal{C},\ i<r}|z_{C,i}|$. For any tuple $(m_C)_{C \in \mathcal{C}}$ satisfying \eqref{eq:mix-cycle-mult-bound}, define
  \[
    z_i \mathrel{:=} \sum_{C \in \mathcal{C}} m_C z_{C,i} \quad (i<r).
  \]
  Then $|z_i| \leq K L / \lambda$ for every $i<r$. Hence the number of possible integer vectors $(z_0,\dots,z_{r-1})$ is $O(L^r)$, with a constant depending only on $M$.

  The vector $(z_i)_{i<r}$ determines the real number
  \[
  \begin{aligned}
    \sum_{C \in \mathcal{C}} m_C w(C)
      &= \sum_{C \in \mathcal{C}} m_C \cdot \frac{1}{A} \sum_{i<r} z_{C,i}\eta_i \\
      &= \frac{1}{A}\sum_{i<r} z_i\eta_i
  \end{aligned}
  \]
  uniquely. Therefore the number of possible cycle contributions in \eqref{eq:mix-walk-cycle-decomp} is $O(L^r)$. Multiplying by the finite number of possible path contributions $b_\pi$, the number of possible values of $\ln D_M(u)$ with $D_M(u)<n$ is $O(L^r)$. Since the logarithm is injective, the same bound holds for the products $D_M(u)$ themselves.
\end{proof}

\begin{proposition}
\label{prp:mix-auto-szlc}
  Let $M$ be a zero-normalized mix-automaton and let $x\mathrel{:=}\operatorname{val}(M)$ and $r\mathrel{:=}\operatorname{syrk}(M)$. Then, for every $n\geq 2$,
  \[
    \operatorname{sZLC}(x[:n]) = O((\log n)^r).
  \]
  The constant hidden in the $O$-notation depends on $M$ but not on $n$.
\end{proposition}

\begin{proof}
  Fix $n \geq 2$. Write $d_q \mathrel{:=} |\delta(q)|$ as in Definition~\ref{def:mix-automaton}, so that $d_q \geq 2$ for every $q$. By \eqref{eq:mix-path-scale}, the quantity $D_M(u)$ is the product of the arities $d_{\operatorname{run}_M(u[:h])}$ read along the walk $u$; we call it the \emph{scale} of $u$. Scales multiply along edges: for $u \in \operatorname{dom}_M$ and $i < d_{\operatorname{run}_M(u)}$,
  \[
    D_M(ui) = D_M(u) \cdot d_{\operatorname{run}_M(u)}.
  \]
  Since every $d_q \geq 2$, we have $D_M(u) \geq 2^{|u|}$, so the only walk of scale $1$ is $u = \epsilon$, with $\operatorname{run}_M(\epsilon) = q_0$.

  \textbf{The program.} Call a pair $(q, D) \in Q \times \mathbb{N}^+$ \emph{attainable} if $q = \operatorname{run}_M(u)$ and $D = D_M(u)$ for some $u \in \operatorname{dom}_M$. For every attainable pair with $D < n$ introduce a fresh vertex $v_{q,D}$, whose \emph{scale} is by definition its second coordinate $D$, and let
  \[
    V \mathrel{:=} \{ v_{q,D} : (q,D) \text{ attainable}, D < n \}.
  \]
  We make $V$ into a zipline program $P = (V, v_{q_0,1}, \delta_P)$ over $\Sigma$. Each $v_{q,D}$ has arity $d_q$, and for $i < d_q$ we set
  \[
    \delta_P(v_{q,D})[i] \mathrel{:=}
    \begin{cases}
      v_{\delta(q)[i], D d_q} & \text{if } D d_q < n,\\
      \tau(\delta(q)[i]) & \text{if } D d_q \geq n.
    \end{cases}
  \]
  Notice that each vertex is either \emph{internal}, with all children vertices in $V$, or a \emph{leaf}, with all children symbols in $\Sigma$. In the internal case the children are indeed valid: if $(q,D)$ is attainable via $u$, then by \eqref{eq:mix-path-scale} the walk $ui$ attains $(\delta(q)[i], D d_q)$, which has scale $D d_q < n$.

  \textbf{$P$ is a synchronous ZLP.} We verify the requirements of Definition~\ref{def:zlp} together with synchronicity (Definition~\ref{def:sync-zlp}).
  \begin{itemize}
    \item \emph{Arity.} Every vertex has arity $d_q \geq 2$.
    \item \emph{Acyclicity.} Each edge between vertices runs from $v_{q,D}$ to some $v_{q',D d_q}$, increasing the scale from $D$ to $D d_q \geq 2D > D$; since the scale strictly increases along edges, there is no directed cycle.
    \item \emph{Unique source.} No child is ever $v_{q_0,1}$, since every child in $V$ has scale $D d_q \geq 2 > 1$ whereas $v_{q_0,1}$ has scale $1$; thus $v_{q_0,1}$ occurs in no list $\delta_P(\cdot)$. Conversely, every element of $V$ occurs as a child: let $v_{q,D} \in V$ with $(q,D) \neq (q_0,1)$. Then $D > 1$, because $(q_0,1)$ is the only attainable pair of scale $1$, so any witnessing walk is non-empty; write it as $u'i$ and set $q' \mathrel{:=} \operatorname{run}_M(u')$. By \eqref{eq:mix-path-scale}, $(q', D_M(u'))$ is attainable with $D_M(u') = D / d_{q'} < n$, so $v_{q',D_M(u')} \in V$, and its $i$-th child is $v_{\delta(q')[i],D_M(u')d_{q'}} = v_{q,D}$.
    \item \emph{Synchronicity.} Set $\phi(v_{q,D}) \mathrel{:=} D$, the scale of the vertex. Then $\phi(v_{q_0,1}) = 1$, and along every edge $v_{q,D} \to v_{q',D d_q}$ we have $\phi(v_{q',D d_q}) = D d_q = |\delta_P(v_{q,D})| \cdot \phi(v_{q,D})$.
  \end{itemize}

  \textbf{Value.} We show, by downward induction on the (finitely many) attainable scales $D < n$, that every $v_{q,D} \in V$ satisfies
  \begin{equation}
    \operatorname{val}_P(v_{q,D}) \sqsubseteq \operatorname{val}_M(q) \quad \text{and} \quad |\operatorname{val}_P(v_{q,D})| \geq \left\lceil \frac{n}{D} \right\rceil.
    \label{eq:mix-szlc-induction}
  \end{equation}

  Fix $v_{q,D} \in V$, put $D' \mathrel{:=} Dd_q$, and let $c_0, \dots, c_{d_q - 1}$ be its children. We first claim that, for every $i < d_q$,
  \begin{equation}
    \operatorname{val}_P(c_i) \sqsubseteq \operatorname{val}_M(\delta(q)[i]) \quad \text{and} \quad |\operatorname{val}_P(c_i)| \geq \left\lceil \frac{n}{D'} \right\rceil.
    \label{eq:mix-child-bound}
  \end{equation}
  If $D' \geq n$, then $c_i = \tau(\delta(q)[i])$, so $\operatorname{val}_P(c_i)$ is the one-symbol word $\tau(\delta(q)[i])$. Since
  \[
    \operatorname{val}_M(\delta(q)[i])[0] = \tau(\delta(q)[i])
  \]
  by definition of $\operatorname{val}_M$, we get
  \[
    \operatorname{val}_P(c_i) \sqsubseteq \operatorname{val}_M(\delta(q)[i]),
  \]
  while $|\operatorname{val}_P(c_i)| = 1 = \lceil n / D'\rceil$ because $D' \geq n$. If $D' < n$, then $c_i = v_{\delta(q)[i],D'}$ and \eqref{eq:mix-child-bound} is the inductive hypothesis \eqref{eq:mix-szlc-induction} at the larger scale $D'$.

  By definition, $\operatorname{val}_P(v_{q,D}) = \operatorname{zip}(\operatorname{val}_P(c_0), \dots, \operatorname{val}_P(c_{d_q - 1}))$. For the prefix claim, take $\iota < |\operatorname{val}_P(v_{q,D})|$ and write $\iota = d_q t + i$ with $i < d_q$. By \eqref{eq:zip}, $\operatorname{val}_P(v_{q,D})[\iota] = \operatorname{val}_P(c_i)[t]$, which is in particular a genuine symbol of $\operatorname{val}_P(c_i)$, so $t < |\operatorname{val}_P(c_i)|$. Using \eqref{eq:mix-child-bound} and then \eqref{eq:mix-zip-identity},
  \[
  \begin{aligned}
    \operatorname{val}_P(v_{q,D})[\iota]
      &= \operatorname{val}_P(c_i)[t] \\
      &= \operatorname{val}_M(\delta(q)[i])[t] \\
      &= \operatorname{val}_M(q)[d_q t + i] \\
      &= \operatorname{val}_M(q)[\iota],
  \end{aligned}
  \]
  so $\operatorname{val}_P(v_{q,D}) \sqsubseteq \operatorname{val}_M(q)$. For the length claim, the definition of $\operatorname{zip}$ gives $|\operatorname{zip}(u_0, \dots, u_{d_q - 1})| \geq d_q \min_{i < d_q} |u_i|$, whence by \eqref{eq:mix-child-bound}
  \[
    |\operatorname{val}_P(v_{q,D})| \geq d_q \min_{i < d_q} |\operatorname{val}_P(c_i)| \geq d_q \left\lceil \frac{n}{D'} \right\rceil = d_q \left\lceil \frac{n}{D d_q} \right\rceil \geq \left\lceil \frac{n}{D} \right\rceil,
  \]
  the last inequality because $d_q \lceil n / (D d_q)\rceil$ is an integer that is at least $n/D$. This proves \eqref{eq:mix-szlc-induction}.

  Applying \eqref{eq:mix-szlc-induction} at the root $v_{q_0,1}$ (scale $1$) yields $\operatorname{val}(P) = \operatorname{val}_P(v_{q_0,1}) \sqsubseteq \operatorname{val}_M(q_0) = x$ and $|\operatorname{val}(P)| \geq \lceil n / 1\rceil = n$. Therefore $x[:n] \sqsubseteq \operatorname{val}(P)$.

  \textbf{Size.} Each vertex contributes $d_q \leq \max_{q \in Q} d_q$ outgoing edges, so $|P| = \sum_{v_{q,D} \in V} d_q \leq (\max_{q \in Q} d_q) \cdot |V|$. Every attainable scale $D < n$ is paired with at most $|Q|$ states, hence
  \[
    |V| \leq |Q| \cdot \left|\{D_M(u) : u \in \operatorname{dom}_M, D_M(u) < n\}\right| = O((\log n)^r)
  \]
  by Lemma~\ref{lem:mix-product-count}. As $|Q|$ and $\max_{q \in Q} d_q$ depend only on $M$,
  \[
    \operatorname{sZLC}(x[:n]) \leq |P| = O((\log n)^r).
  \]
\end{proof}

By Lemma~\ref{lem:mix-zero-normalization} and Proposition~\ref{prp:mix-auto-szlc}, every mix-automatic sequence has polylogarithmic $\operatorname{sZLC}$ growth. By Proposition~\ref{prp:arc-vs-szlc}, the same is true for $\operatorname{ARC}$. Consequently, combining Proposition~\ref{prp:mix-auto-szlc} with Theorem~\ref{thm:arc-muwu}, every fixed mix-automatic sequence is predicted by ARC-MUWU with a polylogarithmic mistake bound; the exponent may depend on the generating mix-automaton.

\section{Properties of ARS and ZLP}
\label{sec:apx-props}

Before establishing the basic properties of ARS and ZLP, we illustrate Definition~\ref{def:ars} with a concrete example. We use the Thue-Morse sequence (Example~\ref{ex:thue_morse}). Note how, in contrast to the step-1 copies underlying LZ77 factorizations (Proposition~\ref{prp:arc-vs-lzc}), the example makes essential use of progressions of step greater than 1.

\begin{example}
\label{ex:ars-tm}
  Let $\Sigma = \{0,1\}$ and let $x \in \Sigma^\omega$ be the Thue-Morse sequence $u_\infty$ of Example~\ref{ex:thue_morse}. Identifying the symbols with bits, the digit-sum characterization of $u_\infty$ yields the recurrences
  \begin{equation}
    x[2t] = x[t], \quad x[2t+1] = 1 - x[t].
    \label{eq:tm-rec}
  \end{equation}
  The second recurrence is a complementation rather than an equality, so it cannot be used in an ARS directly. However, iterating \eqref{eq:tm-rec} once gives two genuine equalities between arithmetic-progression subsequences:
  \[
    x[4t+1] = 1 - x[2t] = 1 - x[t] = x[2t+1],
  \]
  \[
    x[4t+3] = 1 - x[2t+1] = x[t].
  \]

  Fix $n \geq 4$ and consider the prefix $u \mathrel{:=} x[:n]$. We define $R=(L,\Phi)$ as follows. The literals anchor the first two positions:
  \[
    L \mathrel{:=} \{(0,0),(1,1)\}.
  \]
  The set $\Phi$ consists of three equations, obtained from the three identities above by truncating each pair of infinite progressions to its visible part. Writing $\ell_1 \mathrel{:=} \lceil n/2\rceil$, $\ell_2 \mathrel{:=} \lceil (n-1)/4\rceil$ and $\ell_3 \mathrel{:=} \lceil (n-3)/4\rceil$ for the number of indices $t$ with $2t<n$, $4t+1<n$ and $4t+3<n$ respectively, we set
  \[
  \begin{aligned}
    \Phi \mathrel{:=} \{&((1,0,0,\ell_1),(2,0,0,\ell_1)),\\
                       &((2,1,0,\ell_2),(4,1,0,\ell_2)),\\
                       &((1,0,0,\ell_3),(4,3,0,\ell_3))\}.
  \end{aligned}
  \]

  We verify the three conditions of Definition~\ref{def:ars}. The literal condition holds since $x[0]=0$ and $x[1]=1$. For the equation condition, consider e.g. the second pair $(\alpha,\beta)$ with $\alpha(t)=2t+1$ and $\beta(t)=4t+1$: by the choice of $\ell_2$, every target position $\beta(t)=4t+1$ with $t<\ell_2$ satisfies $\beta(t)<n$, and the corresponding source position $\alpha(t)=2t+1\leq 4t+1$ is then also smaller than $n$; hence $u[\alpha]$ and $u[\beta]$ are both words of length $\ell_2$, equal symbol-by-symbol by the identity $x[4t+1]=x[2t+1]$. The other two pairs are handled in the same way using $x[2t]=x[t]$ and $x[4t+3]=x[t]$.

  For coverage, fix $2\leq i<n$. If $i$ is even, write $i=2t$ with $t\geq 1$; the first equation covers $i$, since $t<2t=i$. If $i\equiv 3 \pmod{4}$, write $i=4t+3$; the third equation covers $i$, since $t<4t+3=i$. Finally, if $i\equiv 1 \pmod{4}$ and $i\geq 5$, write $i=4t+1$ with $t\geq 1$; the second equation covers $i$, since $2t+1<4t+1=i$. (Note that at $t=0$ the second equation degenerates to the trivial identity $x[1]=x[1]$ and covers nothing, which is why position 1 must be a literal.)

  Therefore $R$ is an ARS for $x[:n]$, and
  \[
    \operatorname{ARC}(x[:n]) \leq |L|+|\Phi| = 5
  \]
  uniformly in $n$.
\end{example}

It is instructive to compare this direct construction with the generic route through automaticity: since $\mathrm{AC}^{\mathrm{R}}_2(x[:n])=2$, combining Proposition~\ref{prp:zlc-vs-acr} with Proposition~\ref{prp:arc-vs-szlc} only yields $\operatorname{ARC}(x[:n])=O(\log n)$, whereas the system above witnesses $\operatorname{ARC}(x[:n])=O(1)$.

\begin{proposition}
\label{prp:arc-measure}
  $\operatorname{ARC}$ is a word complexity measure.
\end{proposition}

\begin{proof}
  We verify the two conditions of Definition~\ref{def:word-complexity}.

  \emph{Polynomial bound.}
  For any $u \in \Sigma^*$, let $n\mathrel{:=}|u|$, $L\mathrel{:=}\{(u[i],i): i<n\}$, and $\Phi\mathrel{:=}\emptyset$.
  Then $R\mathrel{:=}(L,\Phi)$ is an ARS for $u$: the first condition of Definition~\ref{def:ars} holds by construction, the second is vacuous, and the third holds because $(u[i],i)\in L$ for every $i<n$. Hence
  \[
    \operatorname{ARC}(u) \leq |R| = n,
  \]
  yielding the polynomial bound with $p(x)=x$.

  \emph{Approximate monotonicity.}
  We prove the stronger statement
  \begin{equation}
    \operatorname{ARC}(u[:k]) \leq \operatorname{ARC}(u)
    \label{eq:arc-monotone}
  \end{equation}
  for every $u \in \Sigma^*$ and every $k\leq |u|$, which yields the required bound with $q\equiv 1$.

  Let $R=(L,\Phi)$ be an ARS for $u$. We construct an ARS $R'=(L',\Phi')$ for $u[:k]$ with $|R'|\leq |R|$, which establishes \eqref{eq:arc-monotone} upon taking $R$ to be of minimum size.

  Set
  \[
    L' \mathrel{:=} \{(a,i)\in L : i<k\}.
  \]
  For each $(\alpha,\beta)\in\Phi$ with $\alpha=(r_\alpha,i_\alpha,a_\alpha,b_\alpha)$ and $\beta=(r_\beta,i_\beta,a_\beta,b_\beta)$, define
  \[
    T(\alpha,\beta) \mathrel{:=} \max \{t \in \{0,1,\dots,|\alpha|\} : \alpha(s)<k \text{ and } \beta(s)<k \text{ for all } s<t\}.
  \]

  If $T(\alpha,\beta)\geq 1$, set
  \[
    \alpha' \mathrel{:=} (r_\alpha,i_\alpha,a_\alpha,a_\alpha+T(\alpha,\beta)), \quad
    \beta' \mathrel{:=} (r_\beta,i_\beta,a_\beta,a_\beta+T(\alpha,\beta)),
  \]
  and include $(\alpha',\beta')$ in $\Phi'$; otherwise, discard the pair.

  We verify that $R'$ satisfies the three conditions of Definition~\ref{def:ars} with respect to $u[:k]$.

  \emph{Literal positions.}
  For any $(a,i)\in L'$, we have $(a,i)\in L$, so $u[i]=a$. Since $i<k$, $u[:k][i]=u[i]=a$, as required.

  \emph{Equation pairs.}
  Fix $(\alpha',\beta')\in\Phi'$ derived from $(\alpha,\beta)\in\Phi$, and write $T\mathrel{:=}T(\alpha,\beta)\geq 1$. By construction, $|\alpha'|=|\beta'|=T$, and for every $t<T$,
  \[
    \alpha'(t)=\alpha(t)<k, \quad \beta'(t)=\beta(t)<k.
  \]

  Therefore, $u[:k][\alpha']=u[\alpha']=u[\alpha][:T]$ and $u[:k][\beta']=u[\beta']=u[\beta][:T]$. Since $u[\alpha]=u[\beta]$, we conclude that $u[:k][\alpha']=u[:k][\beta']$.

  \emph{Coverage.}
  Fix $i<k$. The third condition of Definition~\ref{def:ars} applied to $R$ at position $i$ has two cases. If $(u[i],i)\in L$, then $i<k$ gives $(u[i],i)\in L'$. Otherwise, there exist $(\alpha,\beta)\in\Phi$ and $t\in\operatorname{dom}(\alpha)$ with $\alpha(t)<\beta(t)=i$. Then $\beta(t)=i<k$ and $\alpha(t)<i<k$, so $t+1$ lies in the set defining $T(\alpha,\beta)$, whence $T(\alpha,\beta)\geq t+1\geq 1$. The truncated pair $(\alpha',\beta')$ therefore belongs to $\Phi'$, and
  \[
    \alpha'(t)=\alpha(t)<\beta(t)=\beta'(t)=i,
  \]
  with $t\in\operatorname{dom}(\alpha')$ since $t<T(\alpha,\beta)$.

  \emph{Size.}
  By construction, $L'\subseteq L$, and $\Phi'$ is the image of a subset of $\Phi$ under the truncation map $(\alpha,\beta)\mapsto(\alpha',\beta')$. Hence $|L'|\leq |L|$ and $|\Phi'|\leq |\Phi|$, so $|R'|\leq |R|$.
\end{proof}

We now turn to zipline programs, and illustrate Definition~\ref{def:zlp} and Definition~\ref{def:sync-zlp} with two examples: a synchronous ZLP generating the prefixes of the Thue-Morse sequence, and a small ZLP which is not synchronous.

\begin{example}
\label{ex:zlp-tm}
  Let $\Sigma=\{0,1\}$ and let $u_\infty$ be the Thue-Morse sequence of Example~\ref{ex:thue_morse}, with $u_m=u_\infty[:2^m]$. The recurrences \eqref{eq:tm-rec}, restricted to $t<2^m$, can be restated as the zip identities
  \begin{equation}
    u_{m+1} = \operatorname{zip}(u_m,\overline{u}_m), \quad \overline{u}_{m+1} = \operatorname{zip}(\overline{u}_m,u_m).
    \label{eq:tm-zip}
  \end{equation}

  Fix $m\geq 2$ and define $P_m=(Q,v_0,\delta)$ with $Q\mathrel{:=}\{v_0,\dots,v_{m-1}\}\cup\{w_1,
  \dots,w_{m-1}\}$ and
  \begin{itemize}
    \item $\delta(v_j)\mathrel{:=}(v_{j+1},w_{j+1})$ and $\delta(w_j)\mathrel{:=}(w_{j+1},v_{j+1})$ for $j<m-1$;
    \item $\delta(v_{m-1})\mathrel{:=}(0,1)$ and $\delta(w_{m-1})\mathrel{:=}(1,0)$.
  \end{itemize}

  This is a valid ZLP: every arity is $2$, every edge goes from index $j$ to index $j+1$ (so $\delta$ is acyclic), and every vertex other than $v_0$ appears in some outgoing list. A reverse induction on $j$ shows
  \[
    \operatorname{val}_{P_m}(v_j)=u_{m-j}, \quad \operatorname{val}_{P_m}(w_j)=\overline{u}_{m-j}.
  \]
  Indeed, at the bottom level $\operatorname{val}_{P_m}(v_{m-1})=\operatorname{zip}(0,1)=01=u_1$ and $\operatorname{val}_{P_m}(w_{m-1})=\operatorname{zip}(1,0)=10=\overline{u}_1$, and the induction step is precisely \eqref{eq:tm-zip}. In particular, $\operatorname{val}(P_m)=u_m$.

  Moreover, $P_m$ is synchronous: take $\phi(v_j)=\phi(w_j)\mathrel{:=}2^j$. Then $\phi(v_0)=1$, and along every edge $\phi$ is multiplied by $2$, the arity of the parent. Counting $2$ outgoing edges per internal vertex,
  \[
    \operatorname{sZLC}(u_\infty[:2^m]) \leq |P_m| = 2(2m-1)=4m-2,
  \]
  and hence $\operatorname{sZLC}(u_\infty[:n])=O(\log n)$. This matches, up to a constant factor, the bound obtained by combining $\mathrm{AC}^{\mathrm{R}}_2(u_\infty[:n])=2$ with Proposition~\ref{prp:zlc-vs-acr}.
\end{example}

\begin{example}
\label{ex:zlp-nonsync}
  Let $\Sigma=\{0,1\}$ and consider $P=(Q,q_0,\delta)$ with $Q=\{q_0,q_1,q_2\}$ and
  \[
    \delta(q_0)\mathrel{:=}(q_1,q_2), \quad \delta(q_1)\mathrel{:=}(q_2,q_2), \quad \delta(q_2)\mathrel{:=}(0,1).
  \]
  Again, all conditions of Definition~\ref{def:zlp} hold: every arity is $2$, edges only descend along the order $q_0,q_1,q_2$, and $q_0$ is the unique source. Computing values bottom-up,
  \[
    \operatorname{val}_P(q_2)=\operatorname{zip}(0,1)=01,
  \]
  \[
    \operatorname{val}_P(q_1)=\operatorname{zip}(01,01)=0011,
  \]
  \[
    \operatorname{val}_P(q_0)=\operatorname{zip}(0011,01)=00011.
  \]
  In the last step, the shorter child $\operatorname{val}_P(q_2)$ has length $m=2$ and index $j=1$, so by \eqref{eq:zip-length} the zip has length $2\cdot 2+1=5$, stopping just before the round-robin process requests the undefined symbol $\operatorname{val}_P(q_2)[2]$.

  However, $P$ is not synchronous. Suppose $\phi$ were a synchronizing map. Then $\phi(q_0)=1$, and the edge $q_0\to q_2$ forces $\phi(q_2)=|\delta(q_0)|\cdot\phi(q_0)=2$; on the other hand, the path through $q_1$ forces $\phi(q_1)=2$ and then $\phi(q_2)=|\delta(q_1)|\cdot\phi(q_1)=4$---a contradiction.

  The failure of synchronicity is visible in the occurrence structure of $q_2$, in the sense of the discussion following Definition~\ref{def:sync-zlp}. Via the direct edge $q_0\to q_2$, prefixes of $\operatorname{val}_P(q_2)=01$ occupy the step-2 progression $2t+1$ inside $\operatorname{val}(P)$ (positions $1,3$, reading $01$). Via the path through $q_1$, they occupy the two step-4 progressions $4t$ and $4t+2$ (positions $0,4$, reading $01$, and position $2$, reading $0$). The same vertex thus corresponds to arithmetic progressions of two different steps, $2$ and $4$---exactly what the synchronicity condition rules out.
\end{example}

Because zipping is not associative, we cannot reduce an arbitrary ZLP to a binary one. However, we can reduce it to a ZLP where all the arities (out-degrees) are prime, thanks to the following observation.

\begin{proposition}
\label{prp:group-zip}
  Consider any $n,m>0$ and $u_0,u_1,\dots,u_{nm-1}\in\Sigma^*$. For any $i<n$, denote
  \[
    w_i\mathrel{:=}\operatorname{zip}(u_i,u_{i+n},u_{i+2n},\dots,u_{i+(m-1)n}).
  \]
  Then,
  \[
    \operatorname{zip}(u_0,u_1,\dots,u_{nm-1}) = \operatorname{zip}(w_0,w_1,\dots,w_{n-1}).
  \]
\end{proposition}

\begin{proof}
  We use the following immediate reformulation of the definition of $\operatorname{zip}$ in \eqref{eq:zip}. For words $v_0,\dots,v_{r-1}$ with $r>0$, the length of $\operatorname{zip}(v_0,\dots,v_{r-1})$ is the first position at which the round-robin process requests an undefined symbol. Equivalently, it is the least integer of the form $rt+a$, where $a<r$ and $v_a[t]$ is undefined; every smaller position $rt+a$ contributes the symbol $v_a[t]$.

  Let $L$ be the least integer of the form $nmt+nj+i$ such that $i<n$, $j<m$, and $u_{i+nj}[t]$ is undefined. Applying the preceding characterization to the left-hand side gives
  \[
    |\operatorname{zip}(u_0,\dots,u_{nm-1})|=L.
  \]

  We next compute the length of the right-hand side. Fix $i<n$. By the same characterization applied to the definition of $w_i$, the length $|w_i|$ is the least integer of the form $mt+j$ such that $j<m$ and $u_{i+nj}[t]$ is undefined. Therefore $n|w_i|+i$ is the least integer of the form $nmt+nj+i$ satisfying the same condition for this fixed value of $i$. Taking the least value over $i<n$, and applying the characterization once more to the outer zip, gives
  \[
    |\operatorname{zip}(w_0,\dots,w_{n-1})|=L.
  \]

  It remains to compare the symbols. Fix $s<L$. There are unique $t\in\mathbb{N}$, $j<m$, and $i<n$ such that
  \[
    s=nmt+nj+i.
  \]
  Equivalently,
  \[
    s=n(mt+j)+i.
  \]

  Since $s<L$, the symbol $u_{i+nj}[t]$ is defined. By \eqref{eq:zip} applied to the left-hand side,
  \[
    \operatorname{zip}(u_0,\dots,u_{nm-1})[s] = u_{i+nj}[t].
  \]
  Since the right-hand side also has length $L$, \eqref{eq:zip} applied to the outer zip gives
  \[
    \operatorname{zip}(w_0,\dots,w_{n-1})[s]=w_i[mt+j].
  \]
  Finally, applying \eqref{eq:zip} to the definition of $w_i$ gives
  \[
    w_i[mt+j]=u_{i+nj}[t].
  \]

  Hence,
  \[
    \operatorname{zip}(u_0,\dots,u_{nm-1})[s]=\operatorname{zip}(w_0,\dots,w_{n-1})[s].
  \]

  This holds for every $s<L$, and the two words have the same length $L$. Therefore they are equal.
\end{proof}

From this, it is straightforward to show the following.

\begin{proposition}
  For any ZLP $P$, there exists a ZLP $P'=(Q',q_0,\delta')$ s.t.
  \begin{itemize}
    \item $\operatorname{val}(P')=\operatorname{val}(P)$;
    \item $|P'|\leq 2|P|$;
    \item For all $q\in Q'$, $|\delta'(q)|$ is a prime number.
  \end{itemize}
  Moreover, if $P$ is synchronous then $P'$ can be chosen to be synchronous as well.
\end{proposition}

\begin{proof}
  Write $P=(Q,q_0,\delta)$ and $\overline{Q}\mathrel{:=}Q\sqcup\Sigma$. For $q\in Q$, put $d_q\mathrel{:=}|\delta(q)|$. Choose a prime factorization
  \[
    d_q = p_{q,0}p_{q,1}\dots p_{q,\ell_q-1}.
  \]
  Since $d_q\geq 2$, we have $\ell_q\geq 1$. For $s\leq \ell_q$, set
  \[
    D_{q,s}\mathrel{:=}\prod_{r<s} p_{q,r}.
  \]
  Thus $D_{q,0}=1$ and $D_{q,\ell_q}=d_q$.

  \emph{Construction.} We replace each vertex $q$ by a prime-arity gadget. For a tuple $\eta=(\eta_0,\dots,\eta_{s-1})$, where $\eta_r<p_{q,r}$ for every $r<s$, write
  \[
    I_q(\eta) \mathrel{:=} \sum_{r<s} D_{q,r}\eta_r.
  \]
  We identify the empty tuple vertex with $q$ itself, and for every nonempty such $\eta$ with $s<\ell_q$ we introduce a fresh vertex $q_\eta$. These fresh vertices, together with the old vertices in $Q$, form $Q'$.

  It remains to define $\delta'$. Let $\eta$ be a tuple of length $s<\ell_q$. The vertex $q_\eta$ has arity $p_{q,s}$. If $s<\ell_q-1$, define
  \[
    \delta'(q_\eta)[i] \mathrel{:=} q_{\eta i}
  \]
  for $i<p_{q,s}$, where $\eta i$ denotes $\eta$ with $i$ appended. If $s=\ell_q-1$, define
  \[
    \delta'(q_\eta)[i] \mathrel{:=} \delta(q)[I_q(\eta)+D_{q,s}i]
  \]
  for $i<p_{q,s}$. This completes the definition of $P'=(Q',q_0,\delta')$.

  The arity of every vertex of $P'$ is one of the primes $p_{q,s}$. Also, $P'$ is a valid ZLP. Indeed, every fresh vertex has a unique parent inside its gadget, and every old vertex other than $q_0$ still has positive in-degree because it appeared in some list $\delta(q)$ in $P$. The vertex $q_0$ still does not appear in any outgoing list. Acyclicity follows from the acyclicity of $P$: a directed path in $P'$ either moves strictly deeper inside one gadget, or exits the gadget to a child that already lay below the original vertex in $P$.

  \emph{Value.} We prove that the values are unchanged. Fix $q\in Q$, and let $U_i\mathrel{:=}\operatorname{val}_{P'}(\delta(q)[i])$ for $i<d_q$. We claim that the value produced at a gadget vertex $q_\eta$, where $|\eta|=s$, is
  \[
    \operatorname{zip}(U_{I_q(\eta)}, U_{I_q(\eta)+D_{q,s}}, \dots, U_{I_q(\eta)+(d_q/D_{q,s}-1)D_{q,s}}).
  \]

  This is proved by reverse induction on $s$. If $s=\ell_q-1$, the assertion is immediate from the definition of the outgoing list of $q_\eta$. For the induction step, assume the assertion has been proved for tuples of length $s+1$, where $s<\ell_q-1$, and fix a tuple $\eta$ of length $s$. Put
  \[
    n_s\mathrel{:=}p_{q,s}
  \]
  and
  \[
    m_s\mathrel{:=}\frac{d_q}{D_{q,s+1}}.
  \]

  For $a<n_s m_s$, define
  \[
    V_a\mathrel{:=}U_{I_q(\eta)+D_{q,s}a}.
  \]
  For each $i<n_s$, the tuple $\eta i$ has length $s+1$, and
  \[
    I_q(\eta i)=I_q(\eta)+D_{q,s}i.
  \]
  By the induction hypothesis applied to $q_{\eta i}$,
  \[
    \operatorname{val}_{P'}(q_{\eta i}) = \operatorname{zip}(V_i,V_{i+n_s},\dots,V_{i+(m_s-1)n_s}).
  \]
  The outgoing list of $q_\eta$ in $P'$ is $(q_{\eta 0},\dots,q_{\eta(n_s-1)})$. Hence, applying Proposition~\ref{prp:group-zip} to the words $V_0,\dots,V_{n_s m_s-1}$ gives
  \[
    \operatorname{val}_{P'}(q_\eta)=\operatorname{zip}(V_0,V_1,\dots,V_{n_s m_s-1}).
  \]
  Since $n_s m_s=d_q/D_{q,s}$, this is precisely the claimed formula for $q_\eta$. This completes the reverse induction on $s$.

  For the empty tuple $\eta$, $I_q(\eta)=0$ and $D_{q,0}=1$, so the claim gives
  \[
    \operatorname{val}_{P'}(q)=\operatorname{zip}(U_0,U_1,\dots,U_{d_q-1}).
  \]
  We now prove that $\operatorname{val}_{P'}(z)=\operatorname{val}_{P}(z)$ for every $z\in\overline{Q}$ by reverse induction in a topological ordering of the dependency graph of the original program $P$. The assertion is immediate for $z\in\Sigma$. Fix $q\in Q$, and assume it is already known for every internal child of $q$ in $P$. For each $i<d_q$, if $\delta(q)[i]$ is terminal then $U_i=\operatorname{val}_{P}(\delta(q)[i])$ by definition, while if $\delta(q)[i]\in Q$ then the same equality follows from the induction hypothesis. Therefore,
  \[
    \operatorname{val}_{P'}(q)=\operatorname{zip}(\operatorname{val}_{P}(\delta(q)[0]),\dots,\operatorname{val}_{P}(\delta(q)[d_q-1])).
  \]
  The right-hand side is exactly $\operatorname{val}_{P}(q)$ by the definition of the value of a ZLP. Applying this to $q_0$ gives $\operatorname{val}(P')=\operatorname{val}(P)$.

  \emph{Size.} For fixed $q\in Q$, let $E_q$ be the number of edges in the gadget replacing $q$. At tuple-depth $s$, there are $D_{q,s}$ gadget vertices, each of arity $p_{q,s}$. Thus
  \[
    E_q \mathrel{:=} \sum_{s<\ell_q} D_{q,s}p_{q,s}.
  \]
  Using the definition of $D_{q,s}$, we get
  \[
    E_q = \sum_{s<\ell_q} D_{q,s+1}.
  \]
  For every $h\in\{0,\dots,\ell_q-1\}$, the term $D_{q,\ell_q-h}$ is at most $d_q/2^h$, since all prime factors are at least $2$. Hence
  \[
    E_q \leq d_q \sum_{h=0}^{\ell_q-1} 2^{-h}.
  \]
  The geometric sum is less than $2$, so $E_q<2d_q$. The gadgets partition the edge set of $P'$, and therefore
  \[
    |P'|=\sum_{q\in Q} E_q.
  \]
  It follows that
  \[
    |P'|<2\sum_{q\in Q}d_q.
  \]
  Since $|P|=\sum_{q\in Q}d_q$, we conclude that $|P'|\leq 2|P|$.

  \emph{Synchrony.} Finally, suppose that $P$ is synchronous, and let $\phi:Q\to\mathbb{N}$ be a synchronizing map for $P$. Define $\phi':Q'\to\mathbb{N}$ by
  \[
    \phi'(q_\eta)\mathrel{:=}D_{q,s}\phi(q)
  \]
  whenever $q_\eta$ is a gadget vertex corresponding to a tuple $\eta$ of length $s$, including the empty tuple identified with $q$. In particular, $\phi'(q_0)=\phi(q_0)$, so $\phi'(q_0)=1$.

  Consider first an internal gadget edge $q_\eta\to q_{\eta i}$, where $\eta$ has length $s<\ell_q-1$. We have
  \[
    \phi'(q_{\eta i})=D_{q,s+1}\phi(q).
  \]
  Since $D_{q,s+1}=p_{q,s}D_{q,s}$, this implies
  \[
    \phi'(q_{\eta i})=p_{q,s}\phi'(q_\eta).
  \]
  This is the synchrony condition for that edge, because $|\delta'(q_\eta)|=p_{q,s}$.

  It remains to consider a last-level gadget edge from $q_\eta$, where $|\eta|=\ell_q-1$, to an internal vertex $q'\in Q$. Such an edge can occur only when $q'=\delta(q)[I_q(\eta)+D_{q,s}i]$ for some $i<p_{q,s}$, where $s=\ell_q-1$. Since $P$ is synchronous,
  \[
    \phi(q')=d_q\phi(q).
  \]
  Also,
  \[
    d_q=p_{q,s}D_{q,s}.
  \]
  Therefore
  \[
    \phi'(q')=p_{q,s}\phi'(q_\eta).
  \]
  Again $|\delta'(q_\eta)|=p_{q,s}$, so the synchrony condition holds. Terminal children impose no condition. Hence $P'$ is synchronous whenever $P$ is synchronous.
\end{proof}

\end{document}